\documentclass[letterpaper,twocolumn,10pt]{article}
\usepackage{usenix2019_v3}

\newif\ifFINAL
\FINALtrue

\usepackage[utf8]{inputenc}
\usepackage[T1]{fontenc}
\usepackage{csquotes}
\usepackage{varwidth}
\usepackage{algorithm}
\usepackage{algpseudocode}
\usepackage{balance}
\usepackage{color}
\usepackage{nicefrac}
\usepackage{amsmath}
\usepackage{enumitem}
\usepackage{amsthm}
\usepackage[noabbrev,capitalise]{cleveref}
\usepackage{braket}

\usepackage{bm}
\usepackage{mathtools}
\usepackage{multirow}
\usepackage{bigdelim}
\usepackage{mathtools}
\usepackage{amssymb}
\usepackage{indentfirst}
\usepackage{booktabs}
\usepackage{siunitx}
\usepackage{enumitem}
\usepackage{boxedminipage}
\usepackage{float}
\usepackage{graphicx}
\usepackage{caption}
\usepackage{subcaption}
\usepackage[normalem]{ulem}
\usepackage{xpatch}
\usepackage{stmaryrd}
\usepackage{xspace}
\usepackage{siunitx}
\usepackage{amssymb}
\usepackage{framed,mdframed}
\usepackage{placeins}
\usepackage{pifont}

\newcommand{\mr}[2]{\multirow{#1}{*}{#2}}

\makeatletter
\def\@footnotecolor{red}
\define@key{Hyp}{footnotecolor}{%
 \HyColor@HyperrefColor{#1}\@footnotecolor%
}
\def\@footnotemark{%
    \leavevmode
    \ifhmode\edef\@x@sf{\the\spacefactor}\nobreak\fi
    \stepcounter{Hfootnote}%
    \global\let\Hy@saved@currentHref\@currentHref
    \hyper@makecurrent{Hfootnote}%
    \global\let\Hy@footnote@currentHref\@currentHref
    \global\let\@currentHref\Hy@saved@currentHref
    \hyper@linkstart{footnote}{\Hy@footnote@currentHref}%
    \@makefnmark
    \hyper@linkend
    \ifhmode\spacefactor\@x@sf\fi
    \relax
  }%
\makeatother

\hypersetup{footnotecolor=black}
\newcommand*\samethanks[1][\value{footnote}]{\footnotemark[#1]}

\definecolor{dark-green}{rgb}{0.0, 0.4, 0.0}
\definecolor{amethyst}{rgb}{0.6, 0.4, 0.8}
\definecolor{amaranth}{rgb}{0.9, 0.17, 0.31}
\definecolor{babyblue}{rgb}{0.54, 0.81, 0.94}

\iffalse
    \newcommand{\TODO}[1]{{{\color{red} TODO: #1}}}
    \newcommand{\sw}[1]{{{\color{blue} #1}}}
    \newcommand{\ds}[1]{{{\color{dark-green} #1}}}
    
\else
    \newcommand{\TODO}[1]{{{#1}}}
    \newcommand{\sw}[1]{{{#1}}}
    \newcommand{\ds}[1]{{{#1}}}
    
\fi
\newcommand{\comment}[1]{}

\graphicspath{{Images/}}
\DeclareGraphicsExtensions{.pdf,.jpg,.png}
\newtheorem{definition}{Definition}
\newtheorem{theorem}{Theorem}
\newtheorem{lemma}{Lemma}

\definecolor{dblue}{rgb}{0.00, 0.50, 0.90}
\definecolor{lblue}{rgb}{0.70, 0.80, 1.00}
\definecolor{lpink}{rgb}{0.90, 0.70, 1.00}
\definecolor{lgreen}{rgb}{0.80, 0.95, 0.75}
\definecolor{lred}{rgb}{0.99, 0.50, 0.55}
\definecolor{lyellow}{rgb}{1.00, 0.95, 0.75}
\definecolor{llgrey}{rgb}{0.95, 0.95, 0.95}
\definecolor{salmon}{rgb}{0.99, 0.90, 0.90}

\newcommand{\red}[1]{{\color{red}#1}}
\newcommand{\grey}[1]{{\color{blue}#1}}

\newcommand{\pr}[1]{\Pr \left[ #1 \right]}

\newenvironment{boxfig}[2]{% {#1}{#2} = {Caption}{label}
\begin{figure}
  \newcommand{\FigCaption}{#1}
  \newcommand{\FigLabel}{#2}
  \vspace{-\medskipamount}
  \begin{center}
    \begin{small}
      \begin{tabular}{@{}|@{~~}l@{~~}|@{}}
        \hline
        \rule[-1.5ex]{0pt}{1ex}
        \begin{minipage}[b]{.96\linewidth}
          \vspace{1ex}
          \smallskip
          }{%
        \end{minipage}\\
        \hline
      \end{tabular}
    \end{small}
    \vspace{-0.5\bigskipamount}
    \caption{\small \FigCaption}
    \label{\FigLabel}
  \end{center}
\end{figure}
}

\newcommand{\verticaltext}[1]{
	\begin{subfigure}[t]{0.022\linewidth}
		\begin{flushleft}
            \rotatebox{90}{\centering\large{}#1}  
        \end{flushleft}
	\end{subfigure}
}
\makeatletter
\let\orgdescriptionlabel\descriptionlabel
\renewcommand*{\descriptionlabel}[1]{%
  \let\orglabel\label
  \let\label\@gobble
  \phantomsection
  \edef\@currentlabelname{#1}%
  \let\label\orglabel
  \orgdescriptionlabel{#1}%
}
\makeatother

\setenumerate[1]{label={(\arabic*)}}

\newcommand{\testacclabel}{\text{Test}}
\newcommand{\testacc}[2]{\testacclabel_{#1,#2}}
\newcommand{\binomdistrlabel}{\text{binom}}
\newcommand{\binomdistr}[2]{\binomdistrlabel\left(#1,#2\right)}

\newcommand{\ourmetric}[2]{\rho_{#1}(#2)}
\newcommand{\heavyside}[1]{\text{H}\!\left[{#1}\right]}
\newcommand{\Naive}{non-adaptive}

\newcommand{\Advanced}{Adaptive}

\setlist{topsep=2pt,itemsep=-1ex,partopsep=1ex,parsep=1ex}

\newcommand{\poisonratio}{\ensuremath{f_{\mathsf{data}}}}  % remove that once we agreed on a letter
\newcommand{\userratio}{\ensuremath{f_{\mathsf{user}}}}

\usepackage{titling}
\date{}

\begin{document}

% Possible title.: DelTest: Verifiable Data Deletion from Machine Learning as a Service Provider by backdooring User Data.
% Possible title: Can Users Verify Violations of Their Right To be Forgotten?
% possible title: Towards Probabilistic Verification of Machine Unlearning.
%\title{Make Deletion Great Again: Verifiable Ignorance of Elected Representatives.}
%\title{Verifiable data deletion in machine learning by backdooring}

\ifFINAL
    \title{Towards Probabilistic Verification of Machine Unlearning}
\else
    \title{Towards Probabilistic Verification of Machine Unlearning}
\fi
%%%%%%%%%%%%%%%%%%%%%%%%%%%%%%%%%%% AUTHORS FOR USENIX %%%%%%%%%%%%%%%%%%%%%%%%%%%%%%%%%%%

%for single author (just remove % characters)
\author{
{\rm David M.\ Sommer\thanks{The first two authors contributed equally to this work.}}\\
ETH Zürich
\and
{\rm Liwei Song\samethanks}\\
Princeton University
\and
{\rm Sameer Wagh}\\
Princeton University
\and
{\rm Prateek Mittal}\\
Princeton University
} % end author

%%%%%%%%%%%%%%%%%%%%%%%%%%%%%%%%%%% AUTHORS FOR IEEEE %%%%%%%%%%%%%%%%%%%%%%%%%%%%%%%%%%%
% \author{
% \IEEEauthorblockN{David Marco Sommer\IEEEauthorrefmark{1}
% }
% \IEEEauthorblockA{{david.sommer@inf.ethz.ch} \\
% \textit{ETH Zürich}\\
% }
% \and
% \IEEEauthorblockN{Liwei Song\IEEEauthorrefmark{1}
% }
% \IEEEauthorblockA{{liweis@princeton.edu} \\
% \textit{Princeton University}\\
% }
% \and
% \IEEEauthorblockN{Sameer Wagh}
% \IEEEauthorblockA{{swagh@princeton.edu}\\
% \textit{Princeton University}
% }
% \and
% \IEEEauthorblockN{Prateek Mittal}
% \IEEEauthorblockA{{pmittal@princeton.edu}\\
% \textit{Princeton University}\\
% }
% }
%%%%%%%%%%%%%%%%%%%%%%%%%%%%%%%%%%%%%%%%%%%%%%%%%%%%%%%%%%%%%%%%%%%%%%%%%%%%%%%%%%%%%%%%

\maketitle
% These next two lines are for page numbers only
\thispagestyle{plain}
\pagestyle{plain}

% \twocolumn[
% \begin{center}%
%     {\Large\bf Towards Probabilistic Verification of Machine Unlearning}
% \end{center}
% \vspace{17mm}
% ]

\begin{abstract}
\textit{Right to be forgotten}, also known as the right to erasure, is the right of individuals to have their data erased from an entity storing it. 
The status of this long held notion was legally solidified recently by the General Data Protection Regulation (GDPR) in the European Union. As a consequence, there is a need for mechanisms whereby users can verify if service providers comply with their deletion requests. In this work, we take the first step in proposing a formal framework to study the design of such verification mechanisms for data deletion requests -- also known as \textit{machine unlearning} -- in the context of systems that provide machine learning as a service (MLaaS). Our framework allows the rigorous quantification of any verification mechanism based on standard hypothesis testing. Furthermore, we propose a novel backdoor-based verification mechanism and demonstrate its effectiveness in certifying data deletion with high confidence, thus providing a basis for quantitatively inferring machine unlearning. 

We evaluate our approach over a range of network architectures such as multi-layer perceptrons (MLP), convolutional neural networks (CNN), residual networks (ResNet), and long short-term memory (LSTM) as well as over 5 different datasets. Through such extensive evaluation over a spectrum of models and datasets, we demonstrate that our approach has minimal effect on the accuracy of the ML service but provides high confidence verification of unlearning. Our proposed mechanism works with high confidence even if a handful of users employ our system to ascertain compliance with data deletion requests. In particular, with just 5\% of users participating in our system, modifying half their data with a backdoor, and with merely $30$ test queries, our verification mechanism has both false positive and false negative ratios below $10^{-3}$. We also show the effectiveness of our approach by testing it against an adaptive adversary that uses a state-of-the-art backdoor defense method. Overall, our approach provides a foundation for a quantitative analysis of verifying machine unlearning, which can provide support for legal and regulatory frameworks pertaining to users' data deletion requests.

\end{abstract}

%Google Doc: \href{https://docs.google.com/document/d/1WDgXXuOLcjl2cjnXW99RIPDQPUYVxE8Vx93YeJS4IAA/edit#}{\color{blue}here}

%Google Slides: \href{https://docs.google.com/presentation/d/1dR1d6LlOAJOpSehSkqfaECsT-aFlSbBiJjVuZ9k-AHg/edit?usp=sharing}{\color{blue}here}

\newif\ifplotimages
\plotimagestrue

\newlength{\spacehack}
\newlength{\spacehacko}

\section{Introduction}
% \red{Reference backdoor defence evaluation: \url{https://arxiv.org/pdf/2010.12186.pdf}}
%Machine Leaning with complex models has become ubiquitous in the past years. Often, these algorithms are trained with personal data of many individuals, collected and stored without knowledge of the individual and without possibility of removal. blabla privacy concerns. Recent regulations have accounted for these concerns, e.g. EU-GDPR, and demand legal methods to limit and revert this flow of personal data upon request. While the spirit of empowering normal users is clear, the technical means are still limited. And this not only for the fulfillment of the legal requirements, but even more for a user-driven verification whether data-brokers are following the regulations.  

%\TODO{Scan the entire document for missing references.} 
Machine learning models, in particular neural networks, have achieved tremendous success in real-world applications and have driven technology companies, such as Google, Amazon, Microsoft, to provide machine learning as a service (MLaaS).
Under MLaaS, individual users upload personal datasets to the server, the server then trains a ML model on the aggregate dataset and then provides its predictive functionality as a service to the users.
However, recent works have shown that ML models memorize sensitive information of training data~\cite{fredrikson_inversion_CCS15, shokri_membership_SP17, ganju_property_privacy_CCS18, carlini_memorization_usenix19}, indicating serious privacy risks to individual user data. 
At the same time, recently enacted legislation, such as the General Data Protection Regulation (GDPR) in the European Union \cite{GDPR_2016} and the California Consumer Privacy Act (CCPA) in the United States \cite{ccpa}, recognize the consumers' \emph{right to be forgotten}, and legally requires companies to remove a user's data from their systems upon the user's deletion request.

\sw{However, there is a noticeable lack of concrete mechanisms that enables individual users to verify compliance of their requests. Prior works in machine unlearning~\cite{MUL_cao_SP15, deletion_ginart_NIPS19, MUL_bourtoule_arxiv19, deletion_guo_arxiv19, MUL_baumhauer_arxiv20} focus on the scenario of an honest server who deletes the user data upon request, and do not provide any support for a mechanism to verify unlearning. In this work, we formalize an approach that allows users to rigorously verify, with high confidence, if a server has ``deleted their data''. Note that it is hard to ascertain if the data was actually physically deleted from storage, and in this work, deletion refers to the exclusion of a user's data from a MLaaS' model training procedure. This is a reasonable model for real world systems such as Clearview AI which contains image data about millions of users and was reported to violate the privacy policies of Twitter~\cite{clearviewai, twittercease}.}

%A real-world example is the face-recognition tool Clearview.AI which was widely used by American law-enforcement agencies. Clearview AI scraped large publicly available user profiles for pictures, e.g. Facebook and Twitter, and trained their machine learning models to recognize millions of Internet users. Recently, Twitter sent a cease-and-desist letter to the company, insisting that they remove all images as it is against Twitter's policies. Our approach would allow privacy enthusiasts to verify whether Clearview AI complied with data deletion requests. 

%enabling the user to verify that the MLaaS server indeed removes its data upon receiving the deletion request is also critical, especially when there is no trust on the server. However, there is no prior work explicitly discussing the data deletion verification problem. 

\textbf{Formalizing Machine Unlearning.} In this work, we take the first step towards solving this open problem of verified machine unlearning by individual users in the MLaaS setting. First, we formulate the unlearning verification problem as a hypothesis testing problem \cite{lehmann2006testing} (whether the server follows requests to delete users' data or not) and describe the metric used to evaluate a given verification strategy.  Note that for a verifiable unlearning strategy to be effective, it needs to satisfy two important objectives. 
On the one hand, the mechanism should enable individual users to leave a unique \emph{trace} in the ML model after being trained on user data, which can be leveraged in the verification phase. On the other hand, such a unique trace needs to have negligible impact on the model's normal predictive behavior. One possible approach is enabled by membership inference attacks such as \ds{Shokri et al.~\cite{shokri_membership_SP17},  Song et al.~\cite{Mem_song_kdd19}, or Chen et al.~\cite{chen2020machine}}. However, this line of work suffers from a number of limitations -- low accuracy due to the training data not being actively perturbed, extensive knowledge of the MLaaS model's architecture for  white-box attack variants, access to auxiliary or shadow data and computational power in an extent similar to the MLaaS provider -- all of which limit the feasibility of such approaches for our problem setting.
We propose a novel use of backdoor attacks in ML as our mechanism for probabilistically verifying machine unlearning and demonstrate how it meets the two requirements above.

\textbf{Proposed Mechanism.} In the classical backdoor attacks~\cite{badnets_gu_arxiv17, trojan_Liu_ndss18}, the users (adversaries in these settings) manipulate part of training data such that the final trained ML model (1) returns a particular \emph{target label} as the classification on inputs that contain a \emph{backdoor trigger} (e.g., fixed pattern of pixel values at certain positions in the image) and (2) provides normal prediction in the absence of the trigger.
In our machine unlearning verification mechanism, we extend a backdoor proposed by Gu et al. \cite{badnets_gu_arxiv17}.
In our approach, where a fraction of users called \textit{privacy enthusiasts} are interested in the verification, individually choose a backdoor trigger and the associated target label randomly, then add this trigger to a fraction of their training samples (called data poisoning) and set the corresponding labels as the target label. This locally poisoned data is then provided to the MLaaS server. While each privacy enthusiast acts independently, i.e., they do not share information about their individual backdoor or target label, our approach supports an arbitrary fraction of such enthusiasts, up to the point where every user in the training dataset is applying our method. 
We demonstrate that the ML model trained on such data has a high backdoor success rate (i.e., target label classification in the presence of the trigger) for every user's backdoor trigger and target label pair.
When the privacy enthusiast later asks the MLaaS provider to delete its data, it can verify whether the provider deleted its data from the ML model by checking the backdoor success rate using its own backdoor trigger with the target label. A low backdoor success rate is indicative of a model that is not trained on the poisoned data and thus signals that the server followed the deletion request. Through a rigorous hypothesis testing formulation, we can show that this mechanism can be used for high confidence detection of deletion requests.
%If the model has a low backdoor success rate, the privacy enthusiast can infer that the server followed the deletion request; otherwise, the server did not follow the deletion request.
%We are considering systems where users upload a threshold number of data samples and that the users has black-box access to the ML model's prediction capabilities. 
%This is common for \red{real-world examples (recommendation algorithms, google maps where 20 phones can trick the system into think that the street was heavily busy)}.\TODO{Fawkes paper has some suggestions \url{https://www.usenix.org/system/files/sec20-shan.pdf}}

\textbf{Experimental Evaluation.} We theoretically quantify the performance of our backdoor-based verification mechanism under the proposed formulation of hypothesis testing. % \cite{lehmann2006testing}.
Furthermore, we experimentally evaluate our approach over a spectrum of 5 popular datasets (EMNIST, FEMNIST, CIFAR10, ImageNet, and AGNews) and 4 different neural network architectures -- multi-layer perception (MLP), convolution neural network (CNN), residual network (ResNet), long short-term memory (LSTM).
%Furthermore, we experimentally evaluate our approach over a spectrum of datasets and network architectures to establish the robustness of our results. Specifically, we evaluate our proposed techniques over 5 popular datasets -- EMNIST, FEMNIST, CIFAR10, ImageNet, and AGNews and 4 different neural network architectures -- multi-layer perception (MLP), convolution neural network (CNN), residual network (ResNet), long short-term memory (LSTM).
We show that our mechanism has excellent performance -- using $50\%$ poisoned samples and merely $30$ test queries achieves both false positive and false negative value below $10^{-3}$ (and can be further lowered as shown in \cref{tab:dataset_verify}).
We also evaluate the mechanism under an adaptive malicious server, one which uses state-of-the-art backdoor defense techniques to decrease the backdoor attack accuracy.
We find that such a server can lower the backdoor success rate, especially for a low poisoning ratio, but they are still significant enough to validate unlearning with high confidence. 

%\subsection{Our Contributions}
\textbf{Our contributions} can be briefly summarized as follows:

\noindent(1) \textit{Framework for Machine Unlearning Verification: }We provide a rigorous framework for compliance verification of right to be forgotten requests by individual users, in the context of machine learning systems. We formalize this as a hypothesis test between an honest server following the deletion request and a malicious server arbitrarily deviating from the prescribed deletion. Our metric, the power of the hypothesis test, quantifies the confidence a user has in knowing that the service provider complied with its data deletion request. Our framework is applicable to a wide range of MLaaS systems. %Novel method to detect unfulfilled deletion requests.

\noindent(2) \textit{Using Data Backdoors for Verified Machine Unlearning: }We propose a backdoor-based mechanism for probabilistically verifying unlearning and show its effectiveness in the above framework. We provide a thorough mathematical analysis of our proposed mechanism. 
    Theorem~\ref{thm:computerho}, informally stated, enables a user to find out the number of test samples required to achieve high confidence in verifying its deletion request. 
    %We provide a thorough mathematical analysis our proposed mechanism and provide closed form expressions for the deletion confidence of individual users upon using the mechanism. 
    We also provide methods for individual users to estimate parameters and necessary statistics for high confidence detection of non-compliance. 
    %
    %to estimate parameters for individual users considering an implementation of our approach and thus obtaining relevant statistics for a high confidence detection of non-compliance.

\noindent(3) \textit{Evaluating Proposed Mechanism over Various Datasets and Networks: }Finally, we perform a thorough empirical evaluation of our proposed mechanism. We consider 5 different datasets -- EMNIST, FEMNIST, CIFAR10, ImageNet, AG News -- 4 image datasets and 1 text classification dataset. We also study the mechanism over 4 different model architectures -- MLP, CNN, ResNet, and LSTM. %\TODO{what experiments/depth of the evaluation?}
    We quantitatively measure the high confidence of our backdoor-based verification mechanism over different fractions of ``privacy enthusiasts'' -- a set of users participating in the system that are interested in verifying machine unlearning, and show that it %our verification mechanism 
    remains effective for an adaptive malicious server who uses state-of-the-art backdoor defense to mitigate backdoor attacks.
    %We provide evidence for the high confidence verification of our techniques through an evaluation over such heterogeneous datasets and models. 
    We also study the heterogeneity in performance across different users and show how multiple users can jointly improve the verification performance despite have lower accuracy in individual verifications. %the verification confidence in worst-case scenarios can be improved through collaboration among users.
    %, \red{Maybe remove: and provide bounds on the numbers of users sustainable by our approach}.

\section{Neural Network Preliminaries}\label{sec:background}

Neural networks are a class of algorithms that enable supervised learning of certain tasks through sample data performing the task successfully. These algorithms derive their name as they try to mimic the workings of biological neural networks such as the human brain. With increased computational power, modern neural networks operate on large amounts of data and with millions of configurable parameters. Due to the immense complexity of the non-convex parameter space, the learning of parameters is not optimized globally, but in an iterative fashion trained with data split in mini-batches using the stochastic gradient descent algorithm and its variants~\cite{sgd, kingma2014adam}. Over the years, various network architectures have been proposed that have shown promise in different target applications -- multilayer perceptrons~\cite{MLP1961Rosenblatt}, convolutional networks~\cite{lecun_CNN_98}, and residual networks~\cite{he_ResNet_CVPR16} that have achieved tremendous success in computer vision, long short-term memory (LSTM) architectures~\cite{hochreiter1997lstm} for applications in natural language processing.
%, and auto encoders (AE), reinforcement learning (RL) for unsupervised learning tasks. %added internal feedback loops that act as a short-term memory and allow the processing of inputs sequentially. 

%Starting with  multilayer-perceptrons~\cite{MLP1961Rosenblatt}, convolutional networks~\cite{lecun_CNN_98} convolve inputs with trainable filters, residual networks introduced shortcuts over multiple layers~\cite{he_ResNet_CVPR16}, and long short-term memory (LSTM) architectures~\cite{hochreiter1997lstm} added internal feedback loops that act as a short-term memory and allow the processing of inputs sequentially. 

%Deep neural networks are pattern recognizing algorithms that vaguely mimic biological neural networks. With increasing computational power, neural networks with millions of configurable parameters have emerged in recent years. Due to the immense complexity of the non-convex parameter space, these parameters are not optimized globally, but iteratively trained with data split in mini-batches using the stochastic gradient descent algorithm and its descendants~\cite{sgd, kingma2014adam}. Starting with  multilayer-perceptrons~\cite{MLP1961Rosenblatt}, convolutional networks~\cite{lecun_CNN_98} convolve inputs with trainable filters, residual networks introduced shortcuts over multiple layers~\cite{he_ResNet_CVPR16}, and long short-term memory (LSTM) architectures~\cite{hochreiter1997lstm} added internal feedback loops that act as a short-term memory and allow the processing of inputs sequentially. 

\subsection{Backdoor Attacks and Defenses}

In a backdoor attack, the adversary maliciously augments training samples with a hidden trigger into the training process of the target deep learning model such that when the \textit{backdoor trigger} is added to any test input, the model will output a specific \textit{target label}. 
Compared to data poisoning attacks which cause intentional misclassifications on clean test samples via training set manipulations, backdoor attacks only alter the model's predictions in presence of the backdoor trigger and behave normally on clean samples.

\label{sec:background:backdoormethod}
In our work, we build upon the attack of Gu et al.~\cite{badnets_gu_arxiv17}.
For a subset of the training samples, their attack chooses a backdoor pattern (fixed pixels and their color/brightness), applies this pattern to the samples, and changes the labels to the target backdoor label.  
During the training process with the full dataset, the target model finally learns to associate the backdoor trigger with the target label.
Recent works have improved this approach~\cite{backdoor_chen_arxiv17,trojan_Liu_ndss18,salem2020dynamic}, extended it to transfer learning~\cite{backdoor_yao_ccs19}, and avoided alteration of the original labels of poison samples~\cite{backdoor_saha_aaai20,backdoor_turner_arxiv19}. Compared to prior work on backdoors, we  propose and analyze a multi-user scenario where each interested user can employ an individual backdoor. We note that the search for even better backdoor attacks is orthogonal to the goals of this paper.

In parallel, defense mechanisms against backdoor attacks have been explored as well. Wang et al. \cite{backdoor_defense_wang_sp19} introduced Neural Cleanse, a backdoor detection and mitigation mechanism that we use in our work. For each possible label, the approach searches for candidate backdoor trigger patterns using optimization procedures. If at least one of the reverse-engineered patterns is small enough, Neural Cleanse concludes that the model suffers from a backdoor. Then, three mitigation techniques are applied. \label{sec:background:neuralcleansing}
First, Neural Cleanse filters out test samples with similar neuron activation patterns as the reversed trigger. 
Second, Neural Cleanse prunes the neurons which have significant activations only on samples with reversed trigger. However, these two methods suffer from adaptive attacks~\cite{tan2019bypassing}.
In the third approach, Neural Cleanse adds the reversed trigger to some known clean training samples and then fine-tunes the model with those samples along with correct labels. 
Other recent works on mitigating backdoor attacks analyze neuron activation patterns~\cite{liu2018fine, chen2018detecting, liu2019abs}, remove training samples based on computed outlier scores~\cite{backdoor_defense_tran_nips18}, and reconstruct the backdoor~\cite{liu2019abs}.
However, strategic backdoor attacks with adaptive considerations of these defenses can significantly mitigate their performance~\cite{liu2018fine, tan2019bypassing, veldanda2020nnoculation}.

%For further details about other approaches, we refer to the summary of related work in \cref{sec:relatedwork}.

%\red{David: use this in later section, e.g., "our approach": In our paper, we follow the backdoor attack method by Gu et al. \cite{badnets_gu_arxiv17}, which does not need access to the target model or any auxiliary model (e.g., teacher model, a generative model) to generate poison samples. In our paper, we adopt the Neural Cleanse (with fine-tuning mitigation approach) \cite{backdoor_defense_wang_sp19} when evaluating our verification framework under backdoor defense.}

% \input{Sections/formulation}
% \input{Sections/approach}
%%%%% Substitute for the above two sections %%%%%%
\section{Framework for Machine Unlearning}\label{sec:formulation}

%\TODO{Add the ``'interested users'' thing. Emphasize that we don't need everyone to use the system. What is old and what is new? Add that users can control what data they submitt, but they cannot control what is used and hat is stored by the ML provider. We should define the term poison ratio}

In this section, we formalize a framework to study machine unlearning in the context of machine learning as a service (MLaaS). In such systems, users interact with a server that provides MLaaS. However, users who request the removal of their data from the training set of a machine learning service currently have no guarantee that their data was actually \ds{unlearned}. 
Thus, there is an important need to uncover a strategy that can verify if user data was used in the generation of a machine learning model\footnote{Given that the data is outsourced to the server, it is difficult to infer \textit{directly} if the data was used in the training set or not.}. We formalize a framework to study the effectiveness of such machine unlearning verification strategies. For this, we use a standard off-the-shelf hypothesis test. This has been well studied in the statistics literature and forms a rigorous foundation for our framework in Section~\ref{subsec:hypothesis}. However, in Section~\ref{sec:analysis}, we extend this basic set-up to provide a closed form expression of the effectiveness of machine unlearning as well as study other important connections such as bounds on the effectiveness given that the users do not know the ground truth parameters. %Table~\ref{table:notation} contains the important notation used throught this work. 

\begin{table}[t]
\centering
\resizebox{\columnwidth}{!}{
\begin{tabular}{cll} 
\toprule
Symbol          & Range                     & Description \\
\midrule
$n$             & $\mathbb{N}$      	 	& Number of test service requests\\
$\alpha,\: \beta$ & $[0,1]$                   & Type-I and Type-II errors (cf. Eq.~\ref{eq:beta})\\
$p,\; q$          & $[0,1]$                   & Probabilities for analysis (cf. Eq.~\ref{eq:pandq})\\
\mr{2}{\userratio}             & \mr{2}{$[0,1]$}                   & Fraction of users that are privacy enthusiasts \\
& &                                                                         (i.e., those who are verifying unlearning) \\
\mr{2}{\poisonratio}        & \mr{2}{$0$-$100\%$}               & Percentage of data samples poisoned by \\
 & &                                                                          each privacy enthusiast\\
\mr{2}{$\ourmetric{A, \alpha}{s, n}$} & \mr{2}{$[0,1]$}     & Effectiveness of a verification strategy $s$ with a model \\
                &                            & training algorithm A and acceptable Type I error $\alpha$\\
%                &                           & Anything else?\\
\bottomrule 
\end{tabular}
}
\caption{Important notation used in this work.}
\label{table:notation}\vspace{-1em}
\end{table}

\textbf{Framework overview. }In order to enable such a system for verifying machine unlearning, we propose an approach that leverages the users' ability to inject stealthy backdoors into their data. In particular, a small fraction \userratio{} of users which we call \textit{privacy enthusiasts} locally perturb a fraction \poisonratio{} of their data, henceforth called \textit{poisoning}, and thus inject a stealthy backdoor in the data, that is only known to them. If the MLaaS provider trains the model on such data, the backdoor can help the user detect the models trained on the \textit{poisoned data}. Consequently, this behavior can be used to reveal dismissed data deletion requests. 
% Note that our system works for an arbitrary fraction of verifying participants and is only limited by the size of the datasets. 
% In particular, is shown to work even when only 5\% of the users are privacy enthusiasts
Note that our system \textit{does not require} all users to participate in the verification and, in particular, is shown to work even when only 5\% of the users are privacy enthusiasts\footnote{And this too, is limited only by the size of the datasets, the results should extend for smaller participating groups.}. Our approach is illustrated in \Cref{fig:system_design}.

Finally, to quantify the effectiveness of any strategy employed by the users, we use a standard hypothesis test. This has been well studied in the statistics literature and forms a rigorous foundation for our framework. 
% Next we lay the groundwork to describe the hypothesis test for our problem. 
Note that in Section~\ref{sec:analysis}, we provide the required calculations to give a closed form expression for the effectiveness of our strategies as well as study other important characteristics, including the effectiveness of the strategy given that the users do not know the ground truth parameters.

\textbf{Threat Model and Assumptions.} We run our evaluation on both, a non-adaptive server algorithm that does not apply backdoor defense techniques, and an adaptive server algorithms that does. We require that users can control and manipulate the data before they provide it for training and thus might not be applicable in certain scenarios where users do not have adequate capabilities. We assume that privacy enthusiasts possess sufficient data for successful poisoning. We require only black-box prediction access to the MLaaS provider's model, however, the provider is not able to determine which user is querying the trained machine learning model. This can be achieved using of-the-shelf anonymous communication schemes like~\cite{tor,vuvuzela}. Finally, we note that the scope of our approach is limited to validating if users' data was deleted from a specific machine learning model exposed by the MLaaS provider, and does \emph{not} include validating deletion from other computing or storage resources at the provider.   

\begin{figure*}
\centering
\begin{subfigure}[t]{.32\textwidth}
  \centering
  \includegraphics[width=\linewidth]{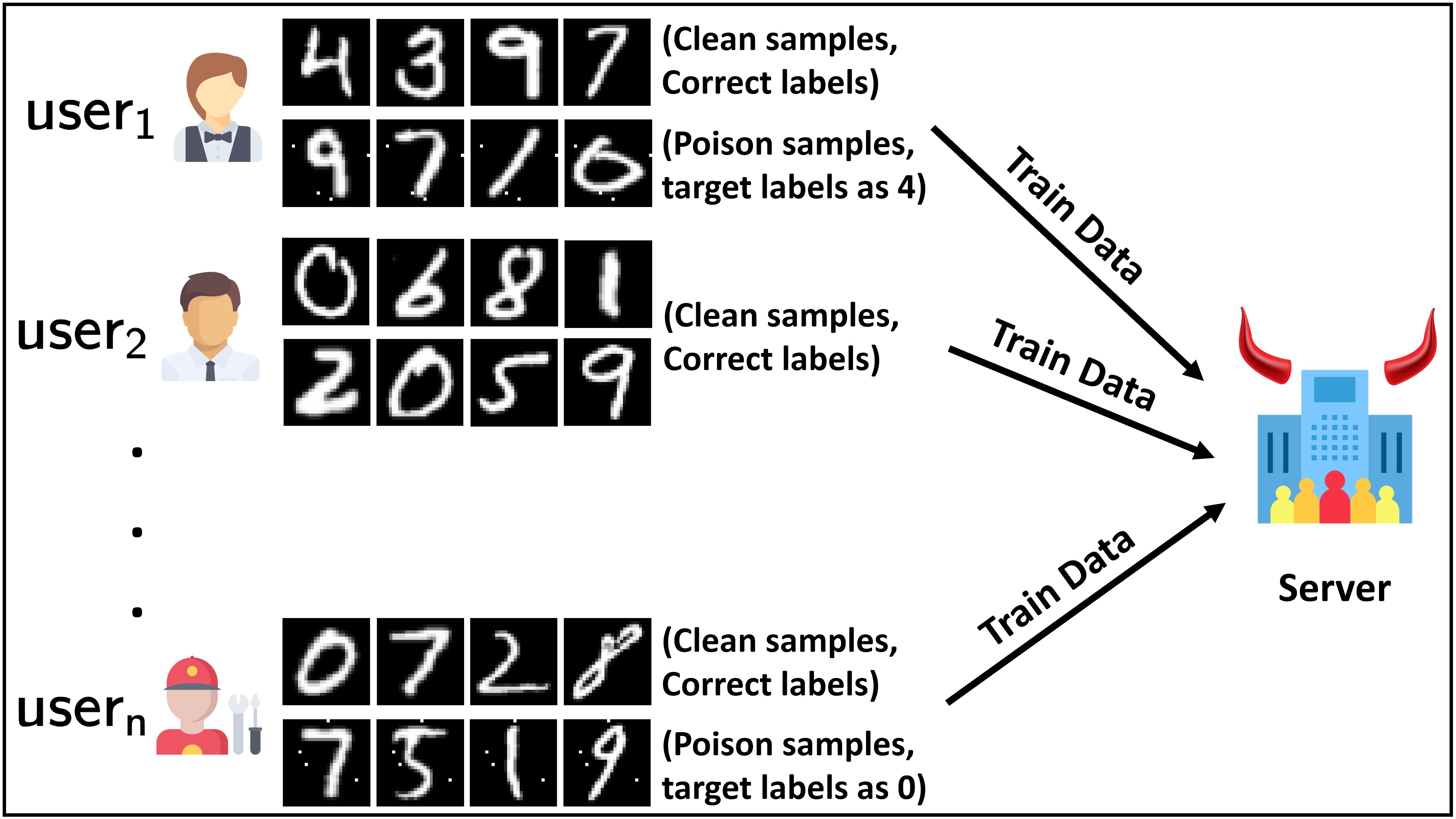}\vspace{0.5em}
  \caption{Backdoor injection during model training. Here, $\mathsf{user_1}, \mathsf{user_n}$ are represented as privacy enthusiasts (poisoning data) and $\mathsf{user_2}$ is not.}
  \label{fig:system_train}
\end{subfigure}\hfill
\begin{subfigure}[t]{.32\textwidth}
  \centering
  \includegraphics[width=\linewidth]{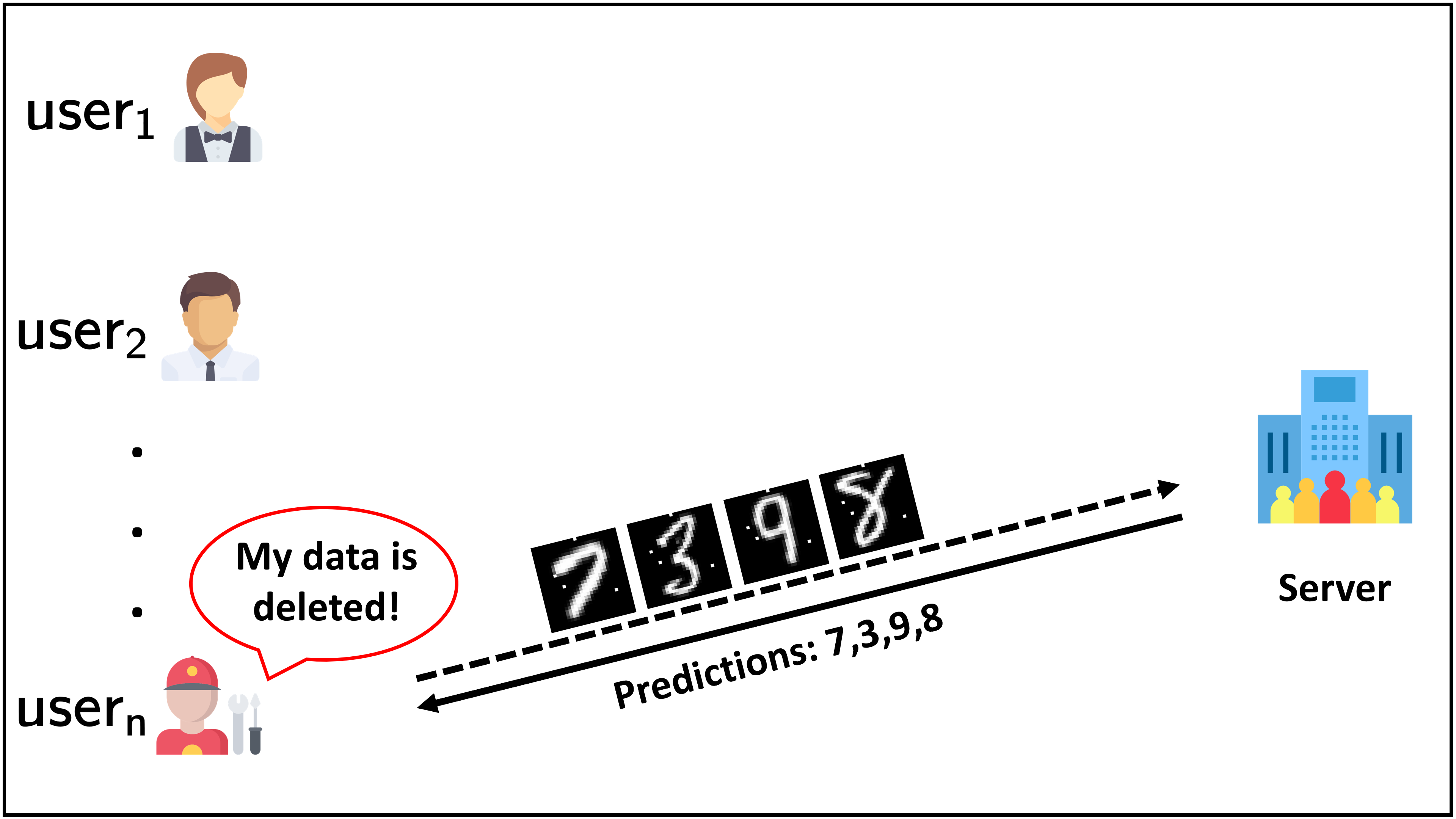}\vspace{0.5em}
  \caption{When the server deletes the user's data ($H_{0}$), the predictions of backdoor samples are correct labels with high probability.}
  \label{fig:system_H0}
\end{subfigure}\hfill
\begin{subfigure}[t]{.32\textwidth}
  \centering
  \includegraphics[width=\linewidth]{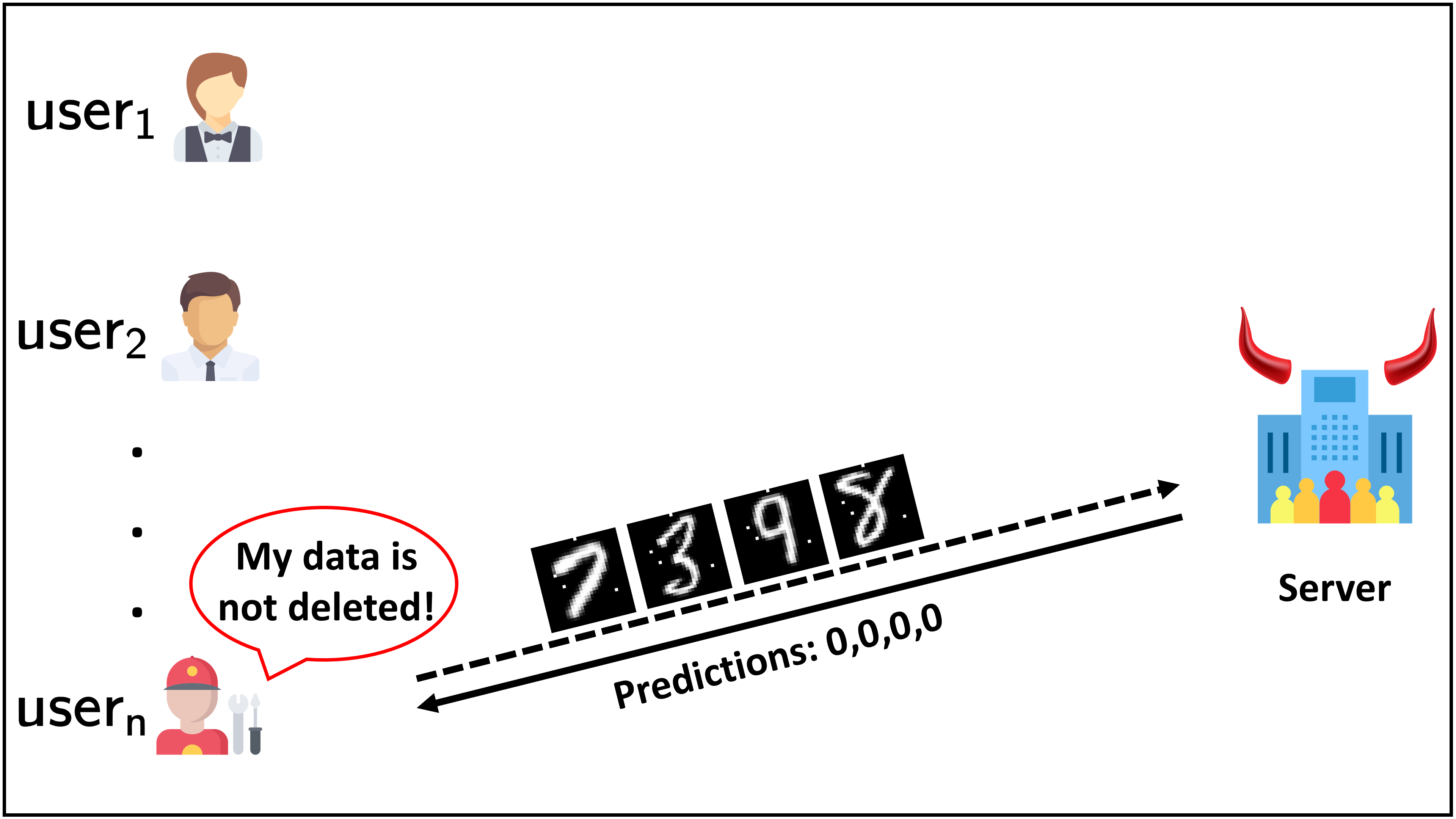}\vspace{0.5em}
  \caption{When the server does not delete the user's data ($H_{1}$), the predictions of backdoor samples are target labels with high probability.}
  \label{fig:system_H1}\hfill
\end{subfigure}\vspace{-1em}
\caption{(Overall system operation.) First,   users inject backdoor samples over which the server trains the model. At a later stage,  users leverage model predictions on backdoored test samples to detect whether the server followed their deletion requests or not -- as shown by the difference between~\cref{fig:system_H0} and~\cref{fig:system_H1}.}
\label{fig:system_design}
%\vspace{-1em}
\end{figure*}

% \begin{figure*}
% \centering
% \begin{subfigure}[t]{.32\textwidth}
%   \raggedleft
%   \includegraphics[width=\linewidth]{Images/train_step.png}
%   \caption{The first step involves the injection of and training over backdoor samples.}
%   \label{fig:system_train}
% \end{subfigure}\hfill
% \begin{subfigure}[t]{.32\textwidth}
%   \centering
%   \includegraphics[width=\linewidth]{Images/H0_result.png}
%   \caption{When the server deletes the user's data ($H_{0}$), the predictions of backdoor samples are correct labels with high probability.}
%   \label{fig:system_H0}
% \end{subfigure}\hfill
% \begin{subfigure}[t]{.32\textwidth}
%   \raggedright
%   \includegraphics[width=\linewidth]{Images/H1_result.png}
%   \caption{When the server does not delete the user's data ($H_{1}$), the predictions of backdoor samples are target labels with high probability.}
%   \label{fig:system_H1}\hfill
% \end{subfigure}
% \caption{The above figure describes the overall operation of our system. First the users inject backdoor samples over which the server trains the model. In a later stage, this enables the users to detect if the server deleted their data or not -- as shown by the difference in~\cref{fig:system_H0} and~\cref{fig:system_H1}.}
% \label{fig:system_design}
% \end{figure*}

\subsection{Hypothesis Testing}\label{subsec:hypothesis}
We apply a hypothesis test to decide whether a MLaaS provider has deleted the requested user data from its training set or not. We define the \textit{null hypothesis} $H_0$ to be the state when server deletes the user data and the \textit{alternative hypothesis} $H_1$ to be the state when the server does not delete the data.
We define the Type I errors as  $\alpha$ (false positive) and Type II errors as $\beta$ (false negative) given below:
\begin{equation}
\begin{aligned}
  \alpha &= \pr{\textrm{Reject } H_0\, |\, H_0 \textrm{ is true}} \\
  \beta &= \pr{\textrm{Accept } H_0\, |\, H_1 \textrm{ is true}}\label{eq:beta}
\end{aligned}
\end{equation}

We define the effectiveness of a verification strategy $s$ for a given server algorithm $A$ for a given acceptable tolerance of Type I error ($\alpha$) to be the power of the hypothesis test formulated above, i.e.,
\begin{equation}\label{eq:test}
\begin{aligned}
    \ourmetric{A, \alpha}{s, n} &= (1-\beta) \\ 
    &= 1 - \pr{\textrm{Accept } H_0\; |\; H_1 \textrm{is true}}
\end{aligned}
\end{equation}%
where $\beta$, as shown in Section~\ref{sec:analysis}, can be computed as a function of $\alpha$, $s$, and $n$. 
%\red{We should make clear how alpha influences the equation}
Informally speaking, $\rho$ quantifies the confidence the user has that the server has deleted their data. This deletion confidence $(1-\beta)$, the power of the hypothesis test, is the probability that we detect the alternative hypothesis when the alternative hypothesis is true. On the other hand, $\alpha$ refers to the acceptable value of the server being falsely accused of malicious activity when in practice it follows the data deletion honestly. For a given value of $n$,  $\alpha$ and $\beta$ cannot be simultaneously reduced and hence we usually set an acceptable value of $\alpha$ and then $(1-\beta)$ quantifies the effectiveness of the test.

\subsection{Our Approach}\label{sec:oursystem}

%\TODO{Move this paragraph somewhere else}Verified machine unlearning is in general very hard given that the data has already been outsourced to another party. Specific to the context of machine-learning-as-a-service (MLaaS), how does one verify if the model is trained on a particular user's data or not? One possible approach is enabled by membership inference attacks such as Shokri et al.~\cite{shokri_membership_SP17} or the approach by Song and Shmatikov~\cite{Mem_song_kdd19}. However, this line of work suffers from a number of limitations -- low accuracy due to the training data not being actively perturbed, extensive knowledge of the MLaaS model's architecture for  white-box attack variants, access to auxiliary or shadow data and computational power in an extent similar to the MLaaS provider -- all of which limit the feasibility of such approaches for our problem setting.

As described earlier, we propose a system where a small fraction of users (privacy enthusiasts) actively engage in verification of machine unlearning. These privacy enthusiasts modify their data locally using a \textit{private backdoor} that is only known to them individually, then they hand their (poisoned) data to the MLaaS provider. The models trained on the \textit{poisoned data} (the data which contains such private backdoors) provide different predictions on very specific samples compared to models trained on data without poisoning. This property can be used to detect whether the server complies with data deletion requests or not. 
Thus, the overall system consisting of the server, normal users, and privacy enthusiasts can be described as shown below: %in Table~\ref{tab:approach}.
%Thus, our proposed mechanism requires no modification to normal users and the server and simply requires the a small fraction of users (privacy enthusiasts) to verify unlearning. %The privacy enthusiasts, can verify machine unlearning with high confidence.

% . Crucially, note that our system does not require all users participate in the verification and in particular is shown to work even when as small as 2\% of the users are privacy enthusiasts\footnote{And this too, is limited only by the size of the datasets, the results should extend for smaller participating groups.}. 

% In order to enable such a system for verifying machine unlearning, we propose an approach that leverages the users' ability to inject stealthy backdoors into their data. In particular, \ds{privacy enthusiasts modify} their data locally using a \textit{private backdoor} that is only known to them individually. If the MLaaS provider trains the model on such data, the backdoor can help the user detect maliciously dismissed data deletion requests as the models trained on the \textit{poisoned data} (the data which contains such private backdoors) provide different predictions on very specific samples compared to models trained on data without poisoning. Compared to prior work on backdoors (cf. \cref{sec:background}), we  propose and analyze a multi-user scenario where each user \ds{can employ} an individual backdoor.

\newcommand*\rot{\rotatebox{90}}

% \begin{table}[h]
{\centering
\resizebox{\columnwidth}{!}{
\begin{tabular}{c c c c}
\toprule
& Privacy & Normal & \mr{2}{Server} \\ 
& Enthusiasts & Users & \\ \midrule
\mr{2}{I} & Send partially & Send normal & \mr{2}{Receive data} \\
          & poisoned data  & data        & \\ \cmidrule(lr){2-4}
\mr{2}{II} & Send data & Send data  & Receive data \\
            & deletion request  & deletion request & deletion requests \\ \cmidrule(lr){2-4}
\mr{2}{III} & Verify data & \mr{2}{--} & \mr{2}{--} \\
          & deletion & & \\
\bottomrule
\end{tabular}
}}
%\caption{Our proposed mechanism requires no modification to normal users and the server. A fraction of users (privacy enthusiasts), similar to miners in the bitcoin ecosystem, can verify machine unlearning with high confidence.}
%\label{tab:approach}
%\end{table}

Note that the accusation of a non-deleting service is based on the success of a backdoor attack.  Privacy enthusiasts generate individual backdoor patterns that alter the predictions of samples to a fixed target label and subsequently provide their samples to the service. Once model prediction is accessible (before or after uploading of samples), the privacy enthusiasts gather statistics on the backdoor success rate when a backdoor has not been seen by a model for later hypthesis testing. This is achieved either by testing the backdoor before uploading, or testing with another backdoor pattern once the model is accessible (see end of this section for details). Finally, in the verification phase, the privacy enthusiasts can study the predictions on $n$ backdoored samples and based on the predictions run a hypothesis test to infer if the server complied with their request or not. The full verification procedure should be handled timely to avoid unlearning-effects due to potential continuous learning. 

We also provide ways in which multiple such privacy enthusiasts can combine their requests to minimize the overall statistical error in correctly detecting server behaviour. For such a system to work well, it is imperative that there exists a statistically significant distinguishing test between models trained \textit{with} vs \textit{without} the backdoored user's data. At the same time, the backdoored data should have minimal impact on the model's normal performance. Through extensive evaluation (\cref{sec:evaluation}), we show these hold for our mechanism.

%These two opposing objectives make it hard to design systems that provide both these requirements. 
% \Cref{fig:system_design} shows our approach and can be described in two phases as follows:

% \noindent \textbf{Phase I}
% \begin{enumerate}
%     \item The users generate individual backdoor patterns that alter the predictions of samples to a fixed label. 
%     %(cf \cref{sec:collisions} for an analysis of similar backdoors by different users).
%     \item Users apply their backdoor patterns to a fraction of their data, submit that to the MLaaS provider, and the provider trains a ML model on the data.
%     \item Each user can run confirmation tests to collect statistics on their backdoor.
% \end{enumerate}

% \noindent \textbf{Phase II}
% \begin{enumerate}
%     \item Some users request a deletion of their data. The server then either retrains the model complying with these requests or proceeds maliciously.
%     \item The users query model with $n$ backdoored samples and based on the predictions try to infer if the server complied with their request or not.
% \end{enumerate}

\textbf{Details on the Underlying Hypothesis Test.} We apply a backdoor using the method described in \cref{sec:background:backdoormethod}, i.e, setting 4 user-specific pixels, spots, or words to a dataset dependent value and changing the label to a user-specific target label. 
Note that the success of altering the prediction using backdoored samples is usually not guaranteed every single time but in a probabilistic manner. Thus, the decision on whether the data has been deleted or not is determined by a hypothesis test. For an effective backdoor algorithm, when the model was trained on backdoored data, the probability of receiving the target label in presence of the backdoor pattern, i.e., backdoor success rate, should be high. At the same time, when the provider has deleted the user's data (the model has not been trained on the user's backdoored samples), the backdoor success rate should be low. In this way the hypothesis test can distinguish between the two scenarios. 

We aim to distinguish the scenario of an honest server who follows the data unlearning protocol from that of a dishonest server who can arbitrarily deviate from the protocol. In particular, we consider two specific models for the dishonest server -- the first is non-adaptive and does not delete user data yet expects to not get detected, while the second is adaptive and employs state-of-the-art defense mechanisms to mitigate user strategies (while also not deleting user data) and thus actively works to evade detection. Throughout this paper, the two probabilities corresponding to backdoor attack accuracy for deleted data and undeleted data are referred to as $q$ (lower), and $p$ (higher) respectively. 
Furthermore, the confidence of this test increases with $n$, the number of backdoored test samples a user queries the trained ML model with.
Thus our verification mechanism can be used to detect missing deletion with high confidence by increasing the number of test samples. 
%Note that the estimation of $p$ and $q$, which is central to the detection of non-compliance, is difficult to compute from individual user perspective. To this end, we provide ways in which individual users can estimate these variables and the effect of using the empirical estimates in comparison to the theoretical ground truth in Section~\ref{sec:analysis:bounds}.

Given that estimation of $p$ and $q$ is central to the detection of non-compliance, we describe the approach from an individual user perspective below.
\begin{enumerate}
    \item \textit{Estimating $p$:} A user can obtain an estimated $\hat{p}$ by querying the model with backdoored samples, immediately after the model has been trained with its data. At this moment, a user can determine whether the backdoor strategy is working. If the resulting $\hat{p}$ is close to the random prediction  baseline, either the applied backdoor strategy $s$ is not working, or its data has not been used in training. However, if $\hat{p}$ is significantly higher than the baseline, our strategy $s$ can work well and we can use $\hat{p}$ as an estimate.
    \item \textit{Estimating $q$:} There are two ways of obtaining an estimate $\hat{q}$: If the algorithm can be queried before the user provides its data, $\hat{q}$ can be obtained by querying the algorithm using samples with the user's backdoor the algorithm has not seen before. If this is not possible, the user can estimate $\hat{q}$ by applying another random backdoor pattern to its data and querying the algorithm. The output should be similar to the case where the algorithm has never seen the user's legitimate backdoor pattern. 
\end{enumerate}
In the highly improbable event when the estimated $p$ is close to the estimated $q$ (cf. Section~\ref{sec:individual_evaluation} for why this might happen), the privacy enthusiasts can detect this and use alternative strategies (cf Section~\ref{subsec:limitation}). Note that backdoor patterns that once worked could be overwritten by new data (for instance in the case of continuous learning). To mitigate this, we recommend privacy enthusiasts to perform the \ds{estimation of $p$ and $q$}  close to the data deletion request.

%There are rare cases where the estimated $p$ is close to the estimated $q$, or $q$ is significantly larger than random guessing. these can be detected by the privacy enthusiast, which then could try another backdoor pattern or switch to other methods, such as membership inference~\cite{Mem_song_kdd19}. Once working backdoor patters could be overwritten by new samples, e.g., in the case of continuous learning. For mitigation, we recommend to perform the unlearning verification close to the data deletion request.

% \TODO{add discussion to reviewer's comments about temporal effects for backdoored users}: 

%%%%%% stepwise detailed explanation
% In detail, our approach illustrated in \Cref{fig:system_design} works as follows:
% \begin{enumerate}
%     \item Client generates an individual backdoor than alters predictions to a fixed label. Due to the high number of possibilities, we neglect collisions. 
%     \item Client applies backdoor to a fraction of its data.
%     \item Clients give their data to the MLaaS provider.  
%     \item Provider trains ML algorithm.
%     \item \grey{Client test for their $\hat{p}$.}
%     \item Client requests deletion.
%     \item After the provider announced re-generation of the model, client queries model with $n$ backdoored samples and saves predictions. 
%     \item Client determines with received predictions an estimate of the backdoor success probability $\hat{m}$.
%     \item Client confirms or objects provider's deletion statement with desired confidence. 
% \end{enumerate}

%\input{Sections/new_formulation}

\section{Theoretical Analysis}\label{sec:analysis}

In Section~\ref{sec:formulation}, we set-up the problem of measuring the effectiveness of user's strategies as a hypothesis test. The hypothesis test allows us to know \textit{what to measure}, i.e., the confidence, for a meaningful quantification of the effectiveness. However, we still \textit{need to measure} this confidence as a function of system parameters and this is what we achieve in this section. In particular, we provide a closed form expression for the confidence $\ourmetric{A, \alpha}{s,n}$ and provide crucial bounds when measured parameters (empirical values) are different from the ground truth (theoretical values). 
%for a given acceptable value of Type I error ($\alpha$).

The confidence, expressed by the metric $\ourmetric{A, \alpha}{s,n}$, is based on a hypothesis test where two cases are compared: $H_0$ -- the data has been deleted, and $H_1$ -- the data has not been deleted. We measure the Type II error $\beta$ which denotes the probability that the server evades detection, i.e., the server behaves malicious but is not caught. Note that this requires we set a level of acceptable Type I error $\alpha$, i.e., the probability that we falsely accuse the server of avoiding deletion. Hence, the metric $\ourmetric{A, \alpha}{s,n} = 1 - \beta$ is a function of the backdoor strategy $s$ and the number of predictions on backdoored test samples $n$ for a given MLaaS server $A$ and a value of Type I error $\alpha$. Figure~\ref{fig:numberofsamples} shows the mechanics of the hypothesis test. For an illustration how $\beta$ decreases with increasing number of test samples, we refer to \cref{fig:multiplesample} in the appendix. 
%and as expected, shows better distinguishing as the number of samples increases.

% \begin{figure*}
% \centering
% \begin{subfigure}[t]{.48\textwidth}
%   \centering
%   \includegraphics[width=\linewidth]{Images/theory_general.pdf}
%   \caption{The relation between the threshold $t$ and the Type I ($\alpha$) and Type II ($\beta$) errors for number of measured samples $n=6$.}
%   \label{fig:singlesample}
% \end{subfigure}\hfill
% \begin{subfigure}[t]{.48\textwidth}
%   \centering
%   \includegraphics[width=\linewidth]{Images/theory_compare.pdf}
%   \caption{With additional number of measured samples $n$, the distributions concentrate around the mean, thus simultaneously decreasing both $\alpha$ and $\beta$. Height of curves are adjusted for demonstration purposes.}
%   \label{fig:multiplesample}
% \end{subfigure}
% \caption{This figure intuitively shows how the confidence in distinguishing between $H_0$ and $H_1$ improves with the number of samples. $q=0.1, p=0.8$}
% \label{fig:numberofsamples}
% \end{figure*}
\begin{figure}
  \centering
  \includegraphics[width=\linewidth]{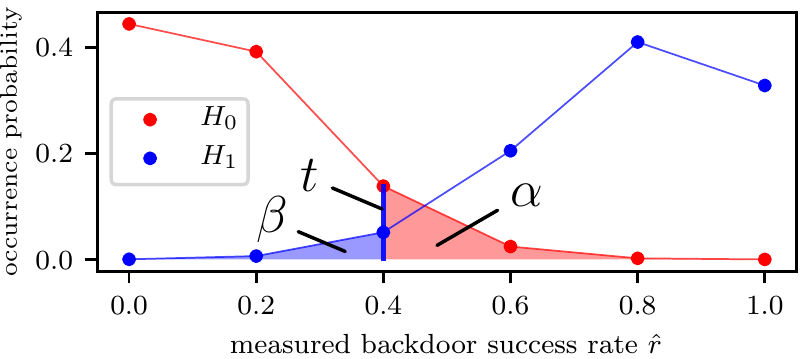}
\caption{This figure shows intuitively the relation between the threshold $t$ and the Type I ($\alpha$) and Type II ($\beta$) errors for number of measured samples $n=5$, with $q=0.1$, and $p=0.8$}
\label{fig:numberofsamples}
\end{figure}

\subsection{Formal Hypothesis Testing}
We query the ML-mechanism $A$ with $n$ backdoored samples $\{\mathsf{sample_i}\}_{i=1}^n$ of a single user and measure how often the ML-mechanism does classify the samples as the desired target label, denoted as $\mathsf{Target_i}$.
%as opposed to the true label of the sample.
%Let us denote by $\mathsf{True_i}$ and $\mathsf{Target_i}$ the true label and the target label of a sample $\mathsf{sample}_i$. 
Then, the measured success rate is 
%\grey{comment by Liwei: $\hat{m}$ should be average value.}
\begin{equation}
    \hat{r} = \frac{1}{n} \sum_{i=1}^n 
    \begin{cases}
        1 \text{ if } A(\mathsf{sample_i}) = \mathsf{Target_i} \\
        0 \text{ otherwise}
    \end{cases}
\end{equation}

Furthermore, we define two important quantities, $q, p$ that quantify the probability that the prediction on backdoored samples is equal to the target label for the null hypothesis vs the alternative hypothesis. For all available data points $i$, %In other words, let $q, p$ be defined as follows:
\begin{equation}\label{eq:pandq}
\begin{aligned}
  q &= \pr{A(\mathsf{sample_i}) = \mathsf{Target_i} \, | \, H_0 \textrm{ is true}} \\
  p &= \pr{A(\mathsf{sample_i}) = \mathsf{Target_i} \, | \,  H_1 \textrm{ is true}}
\end{aligned}
\end{equation}
%Note that our choice implies that these probabilities are computed on samples such that $\mathsf{True_i} \neq \mathsf{Target_i}$. 
Note that the measure $\hat{r}$ approaches $q$ if the null hypothesis $H_0$ (data was deleted) is true and approaches $p$ if the alternative hypothesis $H_1$ (data was not deleted) is true. To decide whether we are in $H_0$ or $H_1$, we define a threshold $t$ and if $\hat{r} \leq t$ we output $H_0$ else output $H_1$. The false-positive error $\alpha$ (Type I error) and false-negative error $\beta$ (type II error) are the respective leftover probability masses that we have decided wrongly. This is illustrated in \cref{fig:numberofsamples}.

The threshold $t$ is set according to the desired properties of the hypothesis test. As common in statistics, $t$ is set based on a small value of $\alpha$ (also known as p-value), the probability that we falsely accuse the ML-provider of dismissal of our data deletion request.

\subsection{Estimating Deletion Confidence $\ourmetric{A, \alpha}{s,n}$}
To derive an analytic expression for $\ourmetric{A, \alpha}{s,n}$, we note that the order in which we request the prediction of backdoored samples does not matter.
Moreover, we assume that the ML provider returns the correct target prediction label with probability $q$ for users with fulfilled deletion requests, and $p$ otherwise.
Further, we assume that the ML provider is not aware which user is querying. Else, the provider could run user specific evasion strategies, e.g., having for each user a unique model with only the user's data excluded.
This can, for example, be achieved by an anonymous communication channel. Since the user strategy is completely defined by the the two parameters $q, p$, we will often interchangeably express a strategy $s$ for these cases by $s=(q,p)$.

First, we show that the occurrence probability of a user-measured average backdoor success ratio $\hat{r}$ follows a binomial distribution with abscissa rescaled to [0,1] with mean $q$ (deleted), or $p$ (not deleted) respectively. Then, we compute the Type II error $\beta$ based on the Type I error $\alpha$ that results from the overlap of these two binomial distributions. Finally, we derive an analytic expression for the verification confidence:

\begin{theorem}\label{thm:computerho}
For a given  ML-mechanism $A$ and a given acceptable Type I error probability $\alpha$, the deletion confidence $\rho_{A, \alpha}(s, n)$ is given by the following expression:
\begin{equation}\label{eq:rhoexpression}
\begin{aligned}
    \rho_{A, \alpha}(s, n) &= 1 - \sum_{k=0}^n \binom{n}{k} p^k(1-p)^{n-k} \cdot \\
    &~~~~~~~~~~ \heavyside{\sum_{l = 0}^k \binom{n}{l} q^l(1-q)^{n-l} \leq 1 - \alpha }
\end{aligned}
\end{equation}
where $p,q$ are as given by \cref{eq:pandq} and $H(\cdot)$ is the heavy-side step function, i.e., $H(x) = 1$ if $x$ is $\mathsf{True}$ and $0$ otherwise. 
\end{theorem}

\Cref{thm:computerho} gives a closed-form expression to compute the backdoor success probability as a function of the system parameters. We defer the proof to Appendix~\ref{app:proof}.

\subsection{\ds{Relaxation to Single User Perspective}} \label{sec:analysis:bounds}
In our theoretical analysis above, we have assumed to know $p$ and $q$ perfectly. In a real-world setup, these values are always measurements. While a machine learning service provider has the ability to quantify $p$ and $q$ accurately on a lot of samples, single users that want to verify the unlearning of their data do usually not have this kind of opportunity. They need to work with estimated values $\hat{p}$ and $\hat{q}$ which can be obtained on a low number of samples $n$ as described in \cref{sec:oursystem}.
We observe that if we overestimate $\hat{q}$ with a bound  $\hat{q}'$ and underestimate $\hat{p}$ with a bound $\hat{p}'$, then the metric $\rho_{A, \alpha}(s, n)$ provides a lower bound, i.e., the confidence guarantees given by $\rho$ do not worsen if the distance between $\hat{q}$ and $\hat{p}$ increases.
\begin{equation}
\rho_{A, \alpha}(s=(\hat{q}',\hat{p}'), n) \leq \rho_{A, \alpha}(s=(\hat{q},\hat{p}), n) 
\end{equation}
with  $\hat{q} \leq \hat{q}'$ and $\hat{p} \geq \hat{p}'$. This comes from the fact that for a given $\alpha$ the overlap of the two scaled binomial distribution decreases when they are moved further apart, and thereby decreasing the $\beta$ which in terms defines $\rho$.

Alternatively, users can assume priors for $p$ and $q$, apply Bayes' theorem, and compute $\rho_{A,\alpha}$ as the expectation over all possible $p$ and $q$:
\begin{equation}
\Pr[r | \hat{r},n] = \frac{\Pr[\hat{r} | r,n]\,\Pr[r]}{\Pr[\hat{r}|n]}
\end{equation}
for $r\in\{q,p\}$ given an estimation $\hat{r} \in \{\hat{q},\hat{p}\}$. ($\Pr[\hat{r}|n] = \sum_r \Pr[\hat{r} | r,n]$.) Then
\begin{equation*}
\ourmetric{A, \alpha}{s=(\hat{q},\hat{p}),n} = \text{E}_{\Pr[q | \hat{q},n],\Pr[p | \hat{p},n]}\left[\rho_{A, \alpha}(s=(q,p), n)\right]
\end{equation*}

%%%%% Substitute for the above section %%%%%%
%\input{Sections/analysis}

\section{Evaluation}\label{sec:evaluation}
In this section, we describe the important results of our experimental evaluation. Before elaborating on datasets and the experimental setup in \cref{sec:evaluation:experimentalsetup}, we mention the central questions for this study as well as results briefly: 
\begin{itemize}
    \item[-] \textit{Q: How well does the verification mechanism work in detecting avoided deletion? Does this generalize to complex and non-image datasets and different architectures?} We answer all these questions affirmatively in \cref{sec:evaluation:naiveserver}. These results are presented in \cref{fig:combined_natural}.
    \item[-] \textit{Q: What happens when the server uses an adaptive strategy such as using a state-of-the-art backdoor defense algorithm to evade detection?} While the detection accuracy is slightly reduced, our approach still excels. This is discussed in \cref{sec:evaluation:advancedserver} and shown in \cref{fig:combined_defended}.
    \item[-] \textit{Q: How do the results change if the fraction of users participating in unlearning detection?} Our approach works for an arbitrary fraction of privacy enthusiasts, as long as individual backdoors are sufficiently reliable. See \cref{fig:nat_all_f} and \cref{fig:def_all_f}.
\end{itemize}
Furthermore, in Sections~\ref{sec:individual_evaluation} and~\ref{sec:discussion}, we look at other aspects of our work including the limitations of our approach when used in practice. 
\ifFINAL Our code is publicly available at \url{https://github.com/inspire-group/unlearning-verification}. \else We will make our code publicly available to facilitate reproducible experiments. \fi

\begin{table*}
\centering
%\resizebox{\textwidth}{!}{
\fontsize{7pt}{7pt}\selectfont
\begin{tabular}{c c c c c | c c c }
\toprule
\multicolumn{5}{c}{\textbf{Dataset Details}} & \multicolumn{3}{c}{\textbf{ML Model}}  \\ \cmidrule(lr){1-8}
\multirow{2}{*}{Name} & sample & number of & number of  & number of  & model  & train acc. & test acc. \\
 & dimension & classes &total samples & total users &  architecture & (no backdoor) & (no backdoor) \\
\midrule
EMNIST & $28 \times 28$ & 10 & 280,000 & 1,000 & MLP & $99.84\%$ & $98.99\%$ \\
\midrule
FEMNIST & $28 \times 28$ & 10 & 382,705 & 3,383 & CNN & $99.72\%$ & $99.45\%$ \\
\midrule
CIFAR10 & $32 \times 32 \times 3$ & 10 & 60,000 & 500 & ResNet20 & $98.98\%$ & $91.03\%$  \\
\midrule
ImageNet & wide variety of sizes, colorful & 1000 & 1,331,167 & 500 & ResNet50 & $87.43\%$ & $76.13\%$ \\
\midrule
AG News & 15--150 words & 4 & 549,714 & 580 & LSTM & $96.87\%$ & $91.56\%$ \\
\bottomrule
\end{tabular}
%}
\caption{Summary of dataset details and model performance without backdoors.}
\label{tab:dataset_summary}
\end{table*}
\begin{table*}
\centering
%\resizebox{\textwidth}{!}{
\fontsize{7pt}{7pt}\selectfont
\begin{tabular}{c c | c c c c | c c c c}
\toprule
\multicolumn{2}{c}{\textbf{Target model}} & \multicolumn{4}{c}{\textbf{Non-adaptive server ($50\%$ poison ratio)}}  & \multicolumn{4}{c}{\textbf{Adaptive server ($50\%$ poison ratio)}} \\ \cmidrule(lr){1-10}
{Name} & backdoor method &  benign test  acc & $p$   & $q$ & $\beta$  &  benign test acc & $p$   & $q$ & $\beta$ \\
%  & method & accuracy &  rate (undeleted users) &  rate (deleted users) & $\alpha=10^{-3}$) & accuracy & rate (undeleted users) &  rate (deleted users) & $\alpha=10^{-3}$) \\
% \multirow{1}{*}{Name} & backdoor &  benign test  & \multirow{1}{*}{$p$}   & \multirow{1}{*}{$q$}  & $\rho$ $(n=20,$ &  benign test  & \multirow{1}{*}{$p$}   & \multirow{1}{*}{$q$}  & $\rho$ $(n=20,$ \\
%  & method & accuracy &   &  & $\alpha=10^{-5}$) & accuracy &  &  & $\alpha=10^{-5}$) \\
\midrule
\multirow{1}{*}{EMNIST} & set 4 random pixels to be 1 & \multirow{1}{*}{$98.92\%$} & \multirow{1}{*}{$95.60\%$} & \multirow{1}{*}{$10.98\%$} & \multirow{1}{*}{$3.2\times10^{-22}$} & \multirow{1}{*}{$98.94\%$} & \multirow{1}{*}{$66.61\%$} & \multirow{1}{*}{$10.46\%$} & \multirow{1}{*}{$4.6\times10^{-5}$} \\
% \midrule
% \multirow{1}{*}{EMNIST (CNN)} & set 4 random pixels to be 1 & \multirow{1}{*}{$99.33\%$} & \multirow{1}{*}{$99.48\%$} & \multirow{1}{*}{$10.17\%$} & \multirow{1}{*}{$1.2\times10^{-20}$} & \multirow{1}{*}{$99.33\%$} & \multirow{1}{*}{$96.69\%$} & \multirow{1}{*}{$9.28\%$} & \multirow{1}{*}{$6.6\times10^{-12}$} \\

% &  & & & & & & & &\\
\midrule
\multirow{1}{*}{FEMNIST} & set 4 random pixels to be 0 & \multirow{1}{*}{$99.41\%$} & \multirow{1}{*}{$99.98\%$} & \multirow{1}{*}{$8.48\%$} & \multirow{1}{*}{$2.2\times10^{-77}$}  & \multirow{1}{*}{$99.12\%$} & \multirow{1}{*}{$71.03\%$} & \multirow{1}{*}{$10.25\%$} & \multirow{1}{*}{$4.0\times10^{-6}$} \\
% &  & & & & & & & &\\
\midrule
\multirow{1}{*}{CIFAR10} & set 4 random pixels to be 1 & \multirow{1}{*}{$90.54\%$} & \multirow{1}{*}{$95.67\%$} & \multirow{1}{*}{$7.75\%$} & \multirow{1}{*}{$4.1\times10^{-24}$} & \multirow{1}{*}{$90.92\%$} & \multirow{1}{*}{$75.90\%$} & \multirow{1}{*}{$10.99\%$} & \multirow{1}{*}{$1.4\times10^{-7}$} \\
% &  & & & & & & & &\\
\midrule
\multirow{1}{*}{ImageNet} & set 4 random spots\footnotemark{} to be 1 & \multirow{1}{*}{{$75.54\%$}} & \multirow{1}{*}{{$93.87\%$}} & \multirow{1}{*}{{$0.08\%$}} & \multirow{1}{*}{{$2.0\times10^{-34}$}}  & \multirow{1}{*}{$75.25\%$} & \multirow{1}{*}{$85.16\%$} & \multirow{1}{*}{$0.11\%$} & \multirow{1}{*}{$2.4\times10^{-23}$}\\
% &  & & & & & & & &\\
\midrule
\multirow{1}{*}{AG News} & replace 4 out of last 15 words & \multirow{1}{*}{$91.35\%$} & \multirow{1}{*}{$95.64\%$} & \multirow{1}{*}{$26.49\%$} & \multirow{1}{*}{$6.6\times10^{-12}$}  & \multicolumn{4}{c}{\multirow{1}{*}{Not Applicable}} \\
\bottomrule
\end{tabular}
%}
\caption{Summary of our backdoor-based verification performance for both non-adaptive and adaptive servers, with a fixed fraction of privacy enthusiasts $\userratio=0.05$. The letter $p$ depicts the backdoor success rate of undeleted users, and $q$ is the backdoor success rate of deleted users.
We provide the Type-II error ($\beta$) of our verification with $30$ test samples and $\alpha$ as $10^{-3}$.
%\red{Liwei: CIFAR10 is still 2\% author poison; the MLP-CNN and ag-news are old results. here, when author poison ratio is low, the verification performance decreases, especially for adaptive server. Maybe we can consider give multiple author poison ratio results for each dataset? A separate table with adaptive server }
% We find that the adaptive server with Neural Cleanse can greatly reduce the backdoor success of the MLP model trained on EMNIST dataset, leading to $\beta=0.93$.
% However, when using CNN architecture to train a EMNIST classifier, the defense of Neural Cleanse is limited, and we reduce $\beta$ to $6.6\times10^{-12}$.
}
\label{tab:dataset_verify}
\end{table*}

\subsection{Experimental Set-up}\label{sec:evaluation:experimentalsetup}

We evaluate all our experiments, wherever applicable, on the 5 datasets described below and on 4 different model-architectures. \Cref{tab:dataset_summary} presents an overview. For more details on the datasets, the network architectures used, and the benign training and test accuracy refer to Appendix~\ref{appendix:datasets}.
\begin{enumerate}
    \item \textbf{Extended MNIST (EMNIST):} The dataset is composed of handwritten character digits derived from the NIST Special Database 19 \cite{grother1995nist} and we train a Multi-Layer-Perceptron (MLP) over this dataset. For the backdoor method, each user chooses a random target label and a backdoor trigger by randomly selecting 4 pixels and setting their values as 1.
    \item \textbf{Federated Extended MNIST (FEMNIST):} The dataset augments Extended MNIST by providing a writer ID \cite{Caldas2018LEAFAB}, for which we trained a convolution neural network (CNN). We use the same backdoor method as for EMNIST, but setting the pixels to 0 instead of 1 (as the pixel values are inverse). 
    \item \textbf{CIFAR10:} This is a dataset containing images in 10 classes (such as airplane, automobile, bird, etc.)~\cite{cifar}, that we used in combination with a ResNet \cite{He_2016_CVPR}. The backdoor method is identical to EMNIST and we consider RGB channels as different pixels.
    \item \textbf{ImageNet:} This is a large scale dataset with images in 1000 classes~\cite{deng2009imagenet}, in combination with a ResNet50 \cite{He_2016_CVPR}. The generation of the backdoor is identical to CIFAR10, except that due to the varying sizes of ImageNet pictures, we colored 4 random color-channels of a transparent 32x32x3 pixel mask-image and then scale the mask-image up to the corresponding ImageNet picture size when applying the backdoor.
    \item \textbf{AG News:} This is a major benchmark dataset for text classification \cite{zhang2015character} containing news articles from more than $2,000$ news sources. For prediction, we applied a long short-term memory (LSTM) model. For the backdoor method, each user chooses a random target label and a backdoor pattern by randomly picking 4 word positions in last 15 words and replacing them with 4 user-specific words, which are randomly chosen from the whole word vocabulary. Note that the number of samples per user varies, starting at only 30 samples. 
    % (cf. \cref{tab:dataset_summary}).
\end{enumerate}

\noindent \textbf{Machine Unlearning Verification Pipeline:} The first part of our evaluation examines the distinguishability of backdoor success rates for data owner whose data has been deleted by a benevolent MLaaS provider versus the case where the provider has maliciously avoided deletion. First,
privacy enthusiasts (with a fraction of \userratio{} among all users) apply their specific backdoor patterns on a certain percentage (\poisonratio{}) of their training samples,
%\ds{a fraction of privacy enthusiasts \userratio{} equip a certain percentage \poisonratio{} of their training samples with a backdoor that is local to them and dataset specific}, 
i.e., 4 random pixels, spots, or words are overwritten and their labels are set to a user specific target label. After training the model with the partially backdoored dataset, we compute the backdoor success rate for each privacy enthusiast's backdoor trigger with its target label, formerly denoted by $p$ in Equation \eqref{eq:pandq}. Then, we compute the backdoor success rate on poisoned users whose data have been excluded before training, introduced as $q$ in Equation \eqref{eq:pandq}.
%Then, we compute the backdoor success rate on a part of the users that have been excluded before training, introduced as $q$ in Equation \eqref{eq:pandq}. 
We compare these values against the benign accuracy of the models on unpoisoned inputs. Finally, we illustrate the decreasing average Type-II error $\beta$ (cf. \cref{eq:beta}) for a range of number of measurements $n$ with a given Type-I error $\alpha$, leading to an increasing average deletion confidence $\ourmetric{A, \alpha}{s, n}$.

As MLaaS providers can defend against backdoor attacks, we illustrate the success of our approach in a comparison of a non-adaptive MLaaS provider that does not implement backdoor defenses to an adaptive provider that implements state-of-the-art defense Neural Cleanse \cite{backdoor_defense_wang_sp19} (cf. \cref{sec:background:neuralcleansing}). 
%Moreover, we evaluate the performance for different models with varying complexity. 

\footnotetext{As the ImageNet dataset contains large colorful pictures of various resolutions, we decided to create a transparent 32x32x3 pixel mask with 4 pixels backdoored and upscale this to the corresponding picture size before applying.}%
Optimally, such an evaluation excludes each poisoning user individually from the full dataset and then retrains the model for each exclusion again from scratch. Due to computation power restrictions, on all tested datasets except ImageNet, we separated 20\% of the available users before training, and trained the models on the leftover 80\% of the users. 
%Note that our approach is applicable and effective even when a small number of users request data deletion. We exclude $20\%$ of users to obtain a reliable estimation of the performance of our verification mechanism.
%While the exclusion of 20\% seems high as usually only a small fraction of users request the deletion of their data, the 20\% were required to generate reliable evaluation statistics. 
Therefore, the first 20\% of users were not included in any training and act in the evaluation as users where the service provider complies their data deletion requests ($H_{0}$).
%that have requested the deletion of their data ($H_{0}$). 
We call them ``deleted users''. In contrast, we have split remaining users' data  into a training and a test set ($80\%$ samples are in training set, and remaining $20\%$ samples are in test set).
%the other 80\% for each user equally in a train and test set, 
We trained the model on the training set. The test set was used to evaluate the case where the users' data was not deleted ($H_{1}$). Accordingly, we call them ``undeleted users''. 
On the large-scale ImageNet, we follow the existing training/test split to include all users' training samples to train the ML model and obtain the backdoor success of ``undeleted users'' ($p$). 
%To simulate the behaviors of ``delete users'', after training the model on the backdoored training set, we regenerate backdoor patterns with different seed from the one used in training set and apply them on the test set and then compute the backdoor success ($q$).
To simulate the behaviors of ``deleted users'', we apply newly generated backdoor patterns to the test-split, and derive the unseen backdoor success $q$ based on their predictions.
Where the resulting numbers were statistically insufficient due to a very low privacy enthusiasts fraction \userratio{}, we repeated the experiments with different random seeds and took the average.

\ifplotimages
\setlength{\spacehack}{-0.4em}
\setlength{\spacehacko}{-0.3em}
\begin{figure*}[!ht]
	\centering
    %%%%%%%%%%%%%%%%%%%%%%%%%%%%%%%%%%%%EMNIST
    \verticaltext{~~~~~~~~~~~~EMINST}
	\begin{subfigure}[t]{0.31\linewidth}
		\raggedleft
		\includegraphics[width=\linewidth]{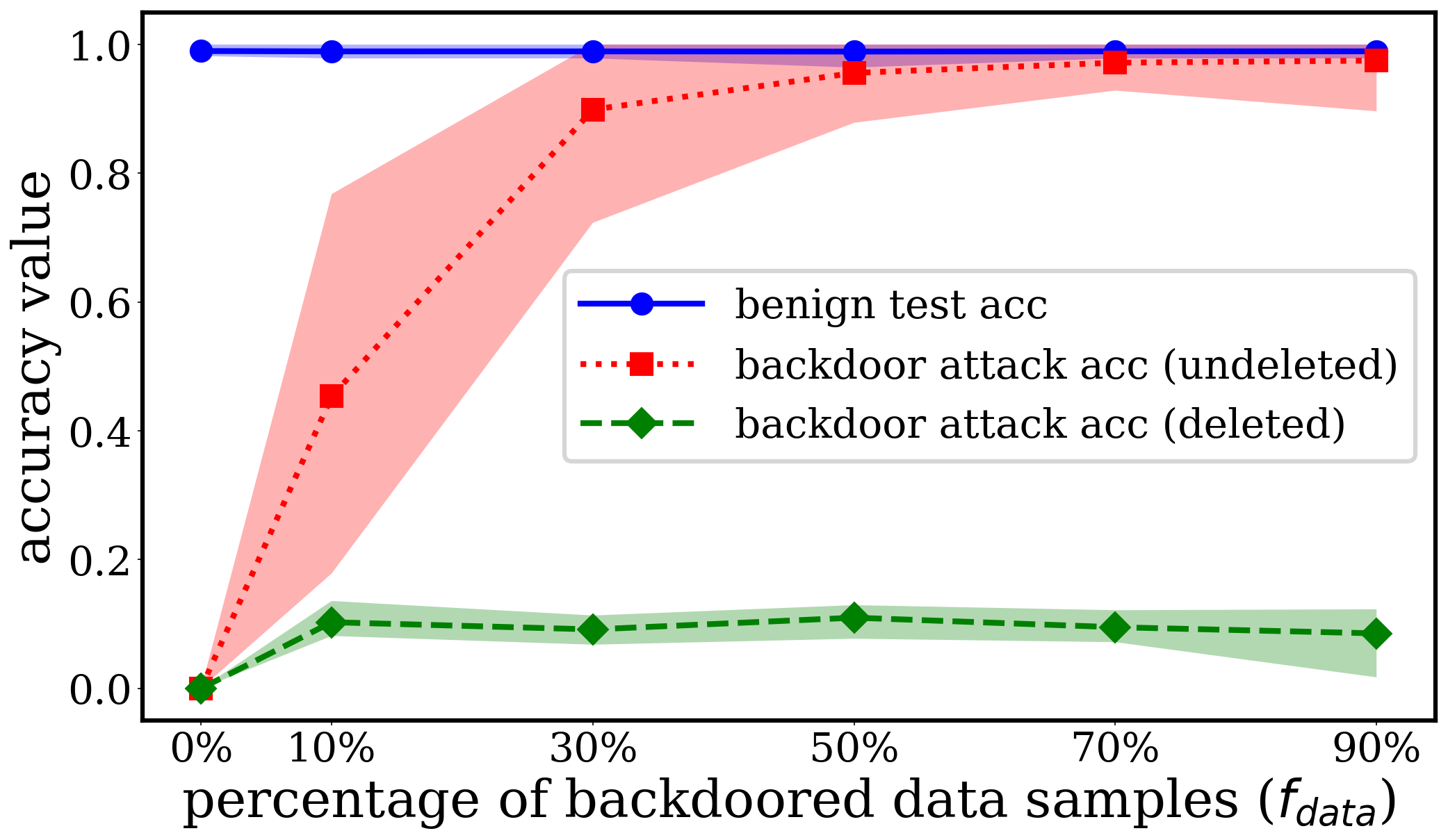}\vspace{\spacehacko}
	\end{subfigure}\hfill
	\begin{subfigure}[t]{0.31\linewidth}
		\centering
		\includegraphics[width=\linewidth]{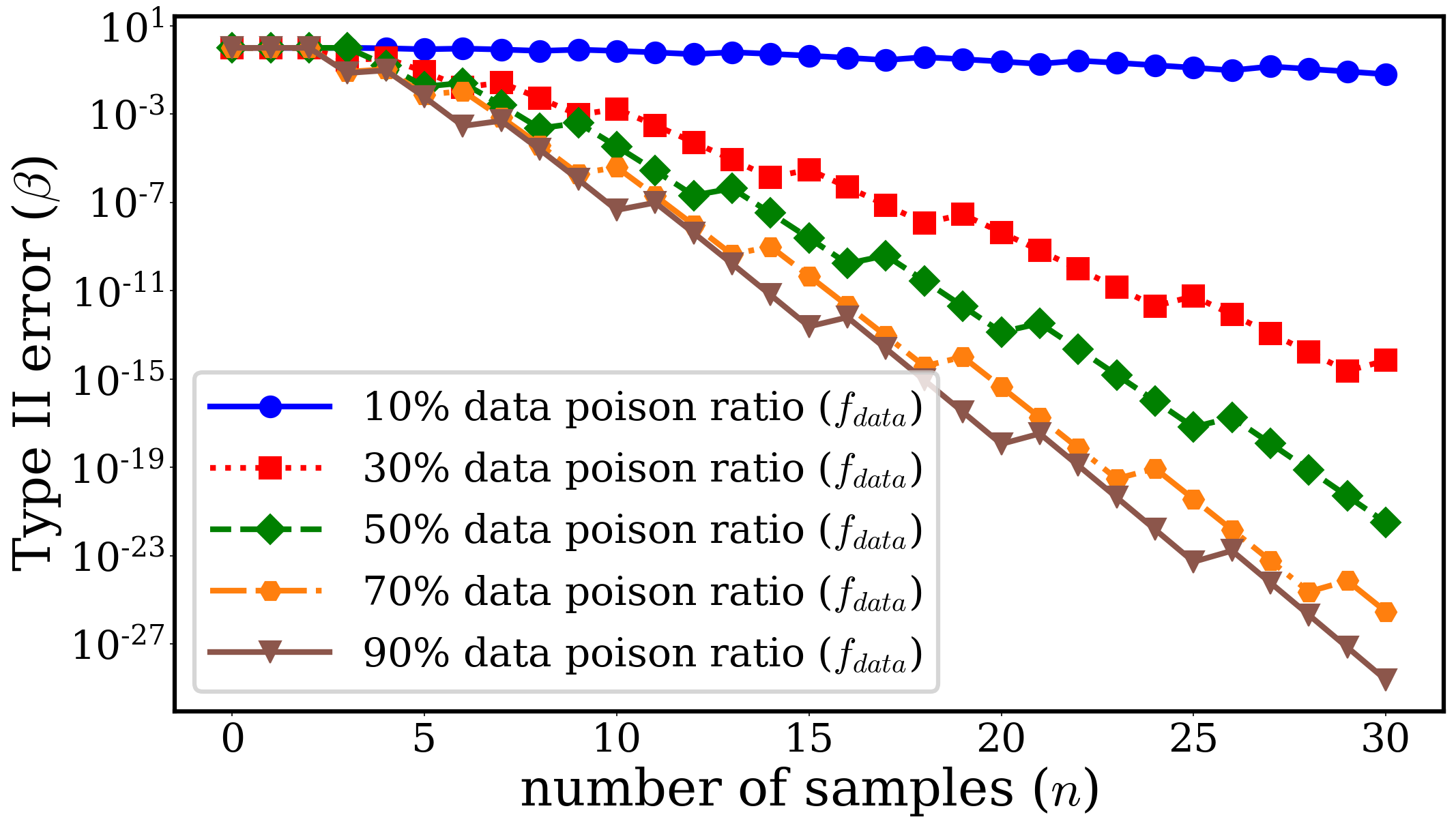}\vspace{\spacehacko}
	\end{subfigure}\hfill
	\begin{subfigure}[t]{0.31\linewidth}
		\raggedright
		\includegraphics[width=\linewidth]{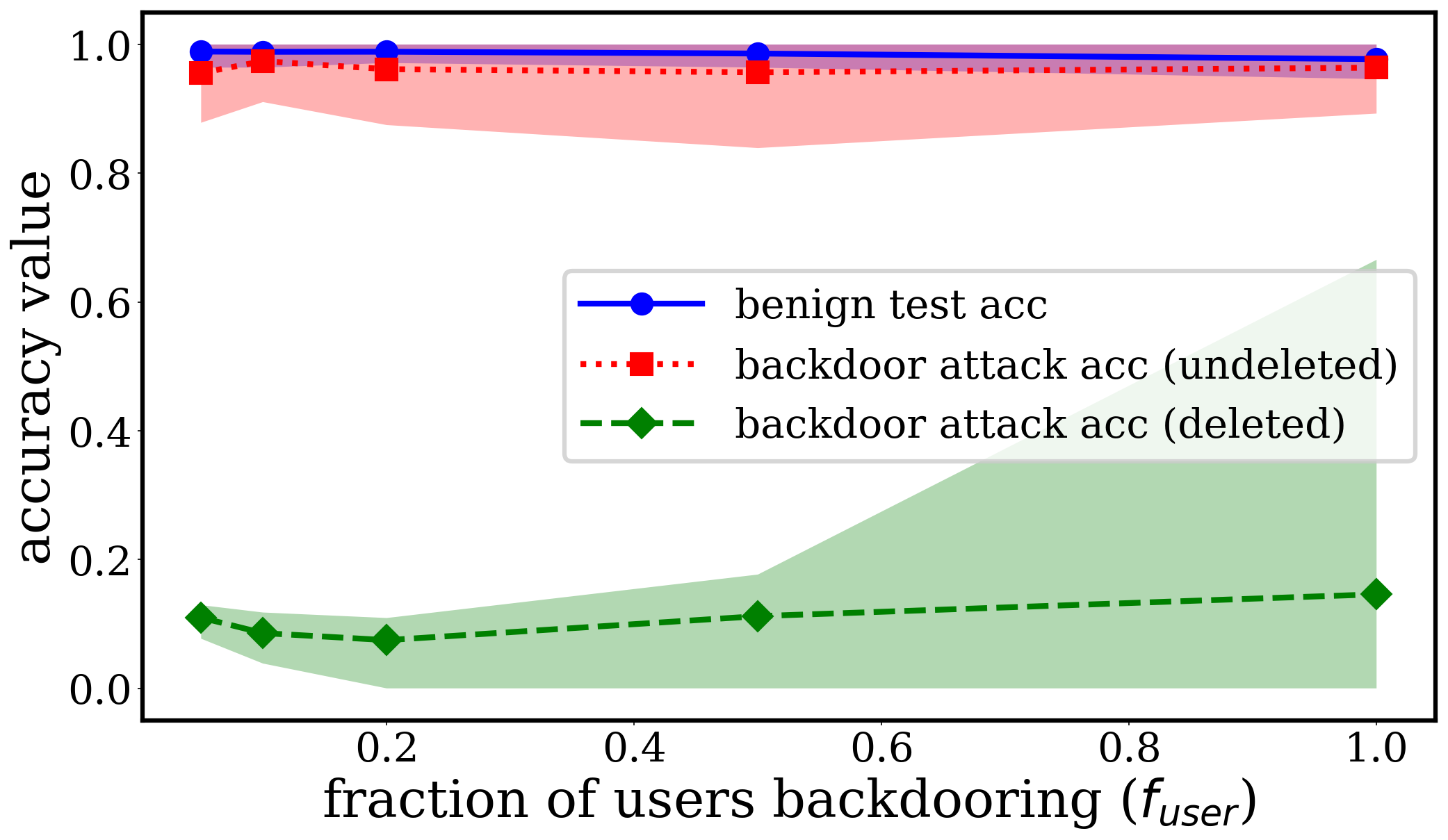}\vspace{\spacehacko}
	\end{subfigure}\hfill
	\par\medskip\vspace{\spacehack}
    %%%%%%%%%%%%%%%%%%%%%%%%%%%%%%%%%%%%FEMNIST
    \verticaltext{~~~~~~~~~~~FEMNIST}
	\begin{subfigure}[t]{0.31\linewidth}
		\raggedleft
		\includegraphics[width=\linewidth]{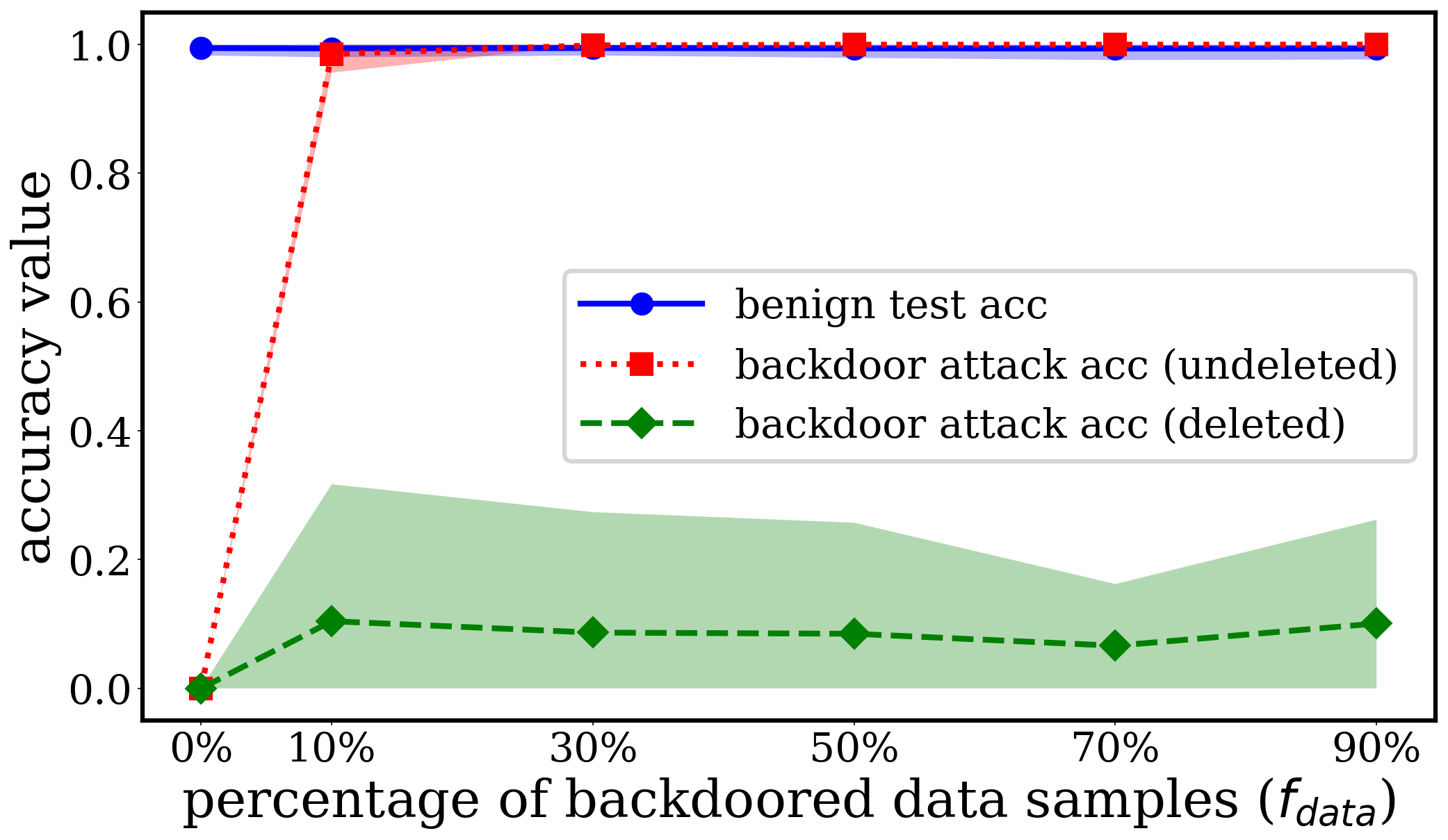}\vspace{\spacehacko}
	\end{subfigure}\hfill
	\begin{subfigure}[t]{0.31\linewidth}
		\centering
		\includegraphics[width=\linewidth]{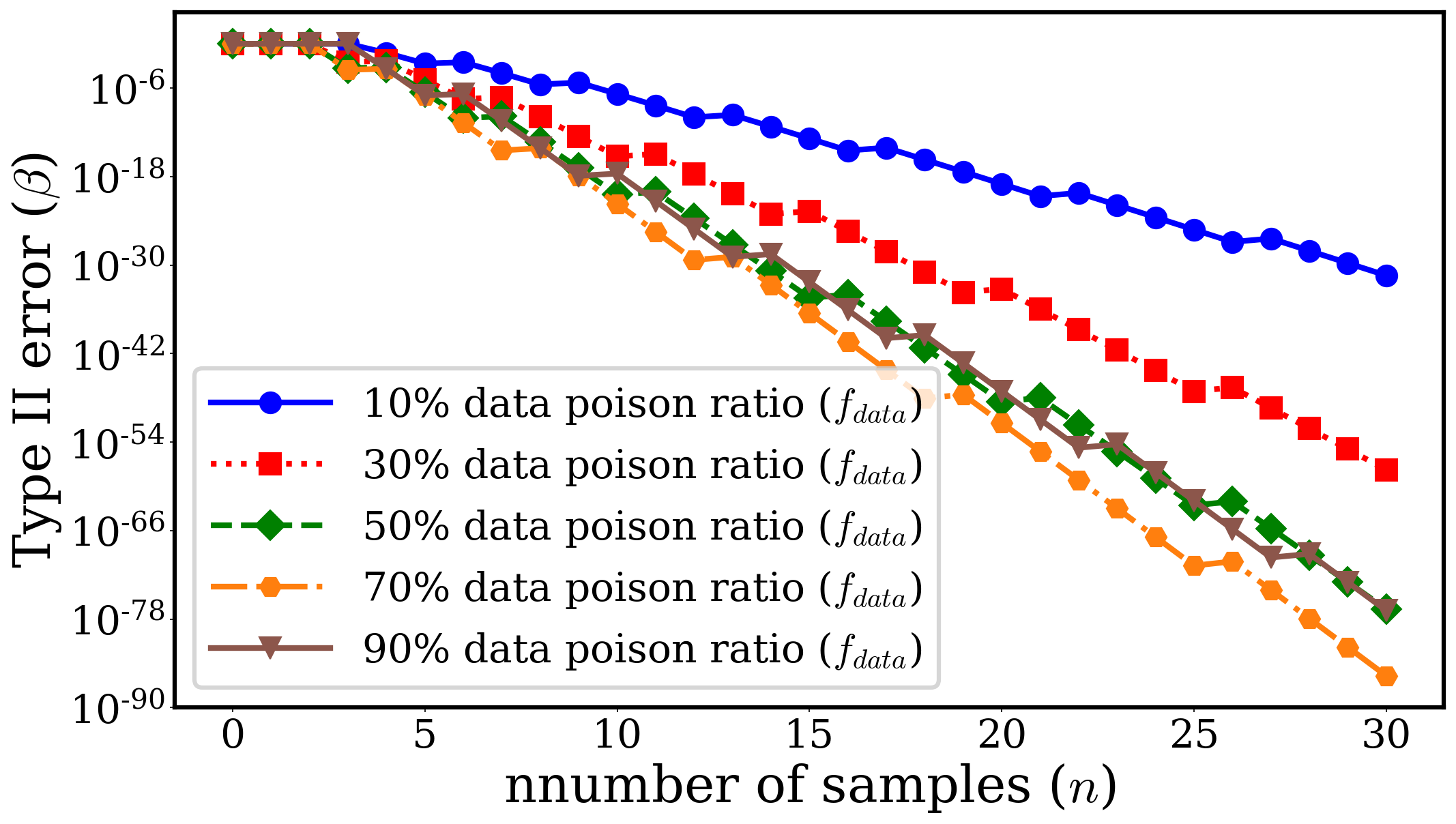}\vspace{\spacehacko}
	\end{subfigure}\hfill
	\begin{subfigure}[t]{0.31\linewidth}
		\raggedright
		\includegraphics[width=\linewidth]{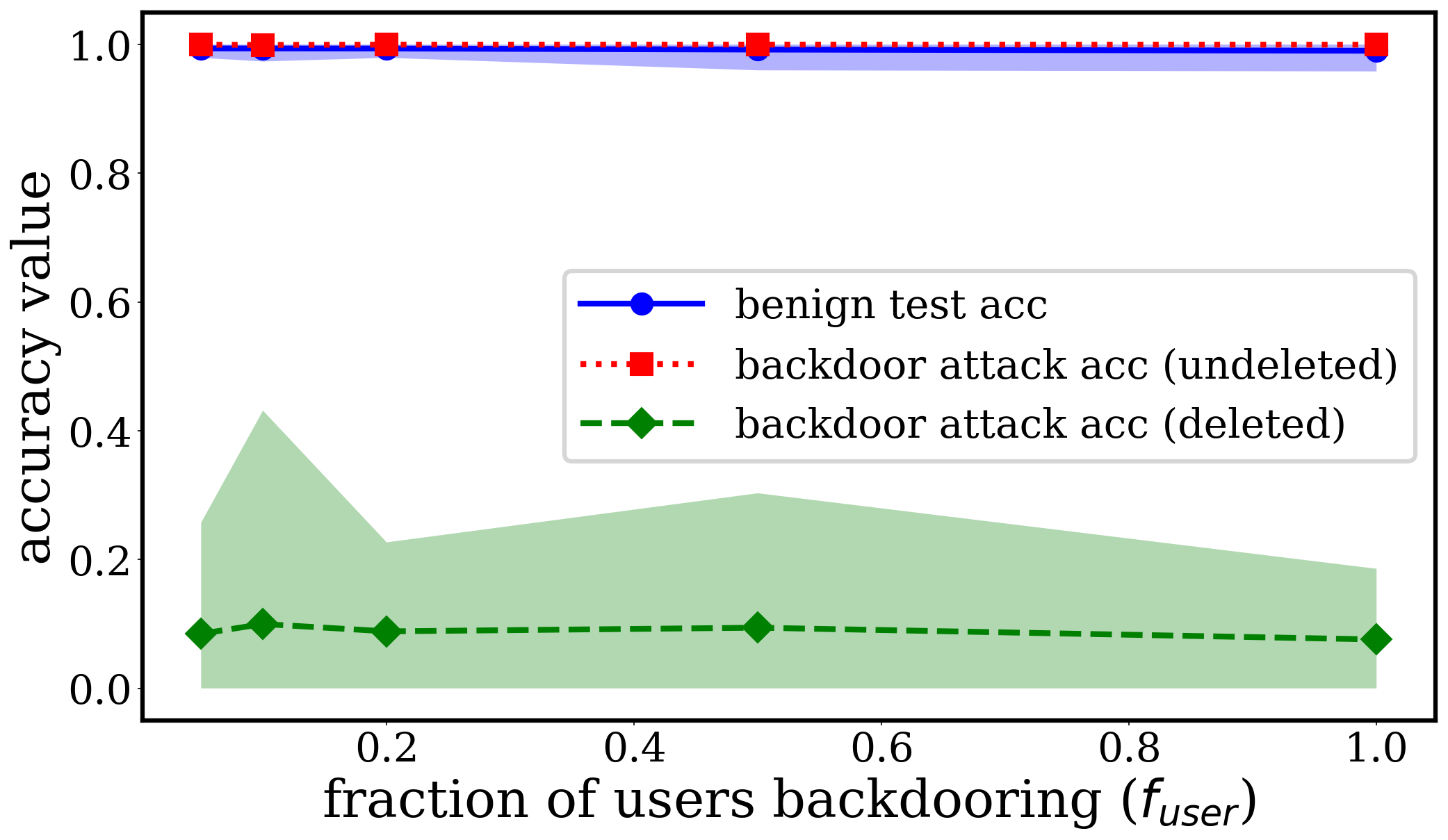}\vspace{\spacehacko}
	\end{subfigure}\hfill
	\par\medskip\vspace{\spacehack}
    %%%%%%%%%%%%%%%%%%%%%%%%%%%%%%%%%%%%CIFAR
    \verticaltext{~~~~~~~~~~~CIFAR10}
	\begin{subfigure}[t]{0.31\linewidth}
		\raggedleft
		\includegraphics[width=\linewidth]{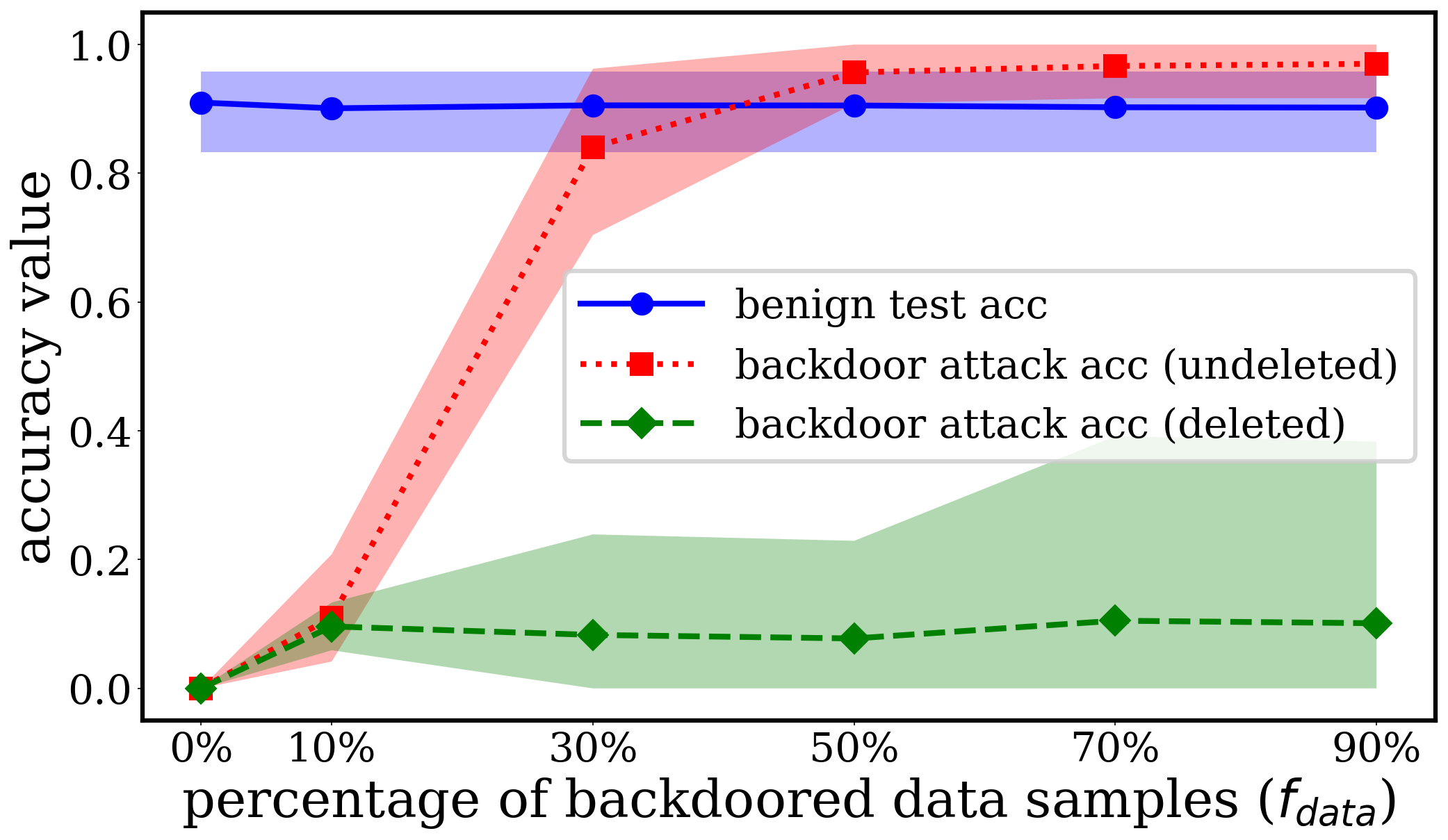}\vspace{\spacehacko}
	\end{subfigure}\hfill
	\begin{subfigure}[t]{0.31\linewidth}
		\centering
		\includegraphics[width=\linewidth]{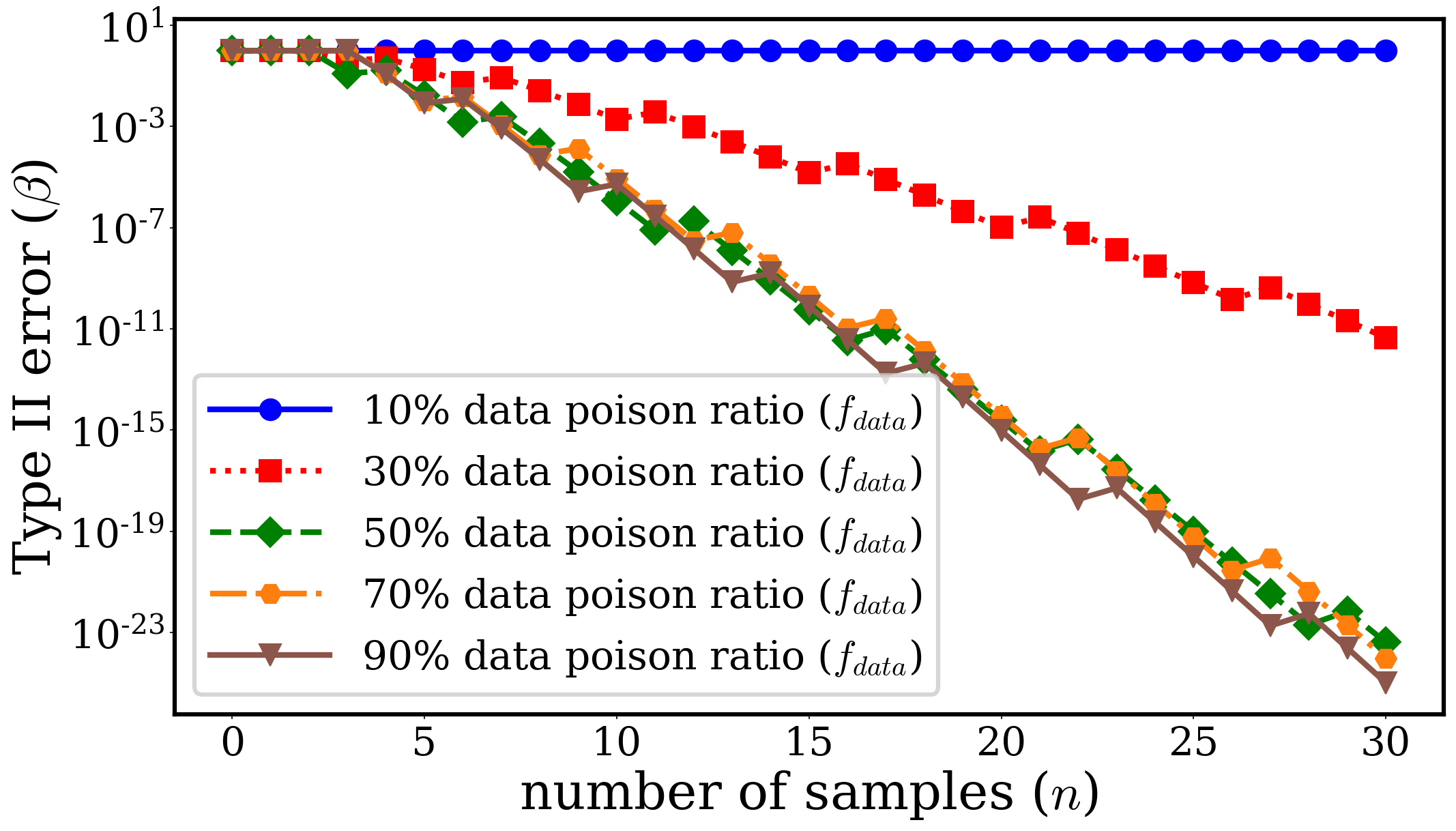}\vspace{\spacehacko}
	\end{subfigure}\hfill
	\begin{subfigure}[t]{0.31\linewidth}
		\raggedright
		\includegraphics[width=\linewidth]{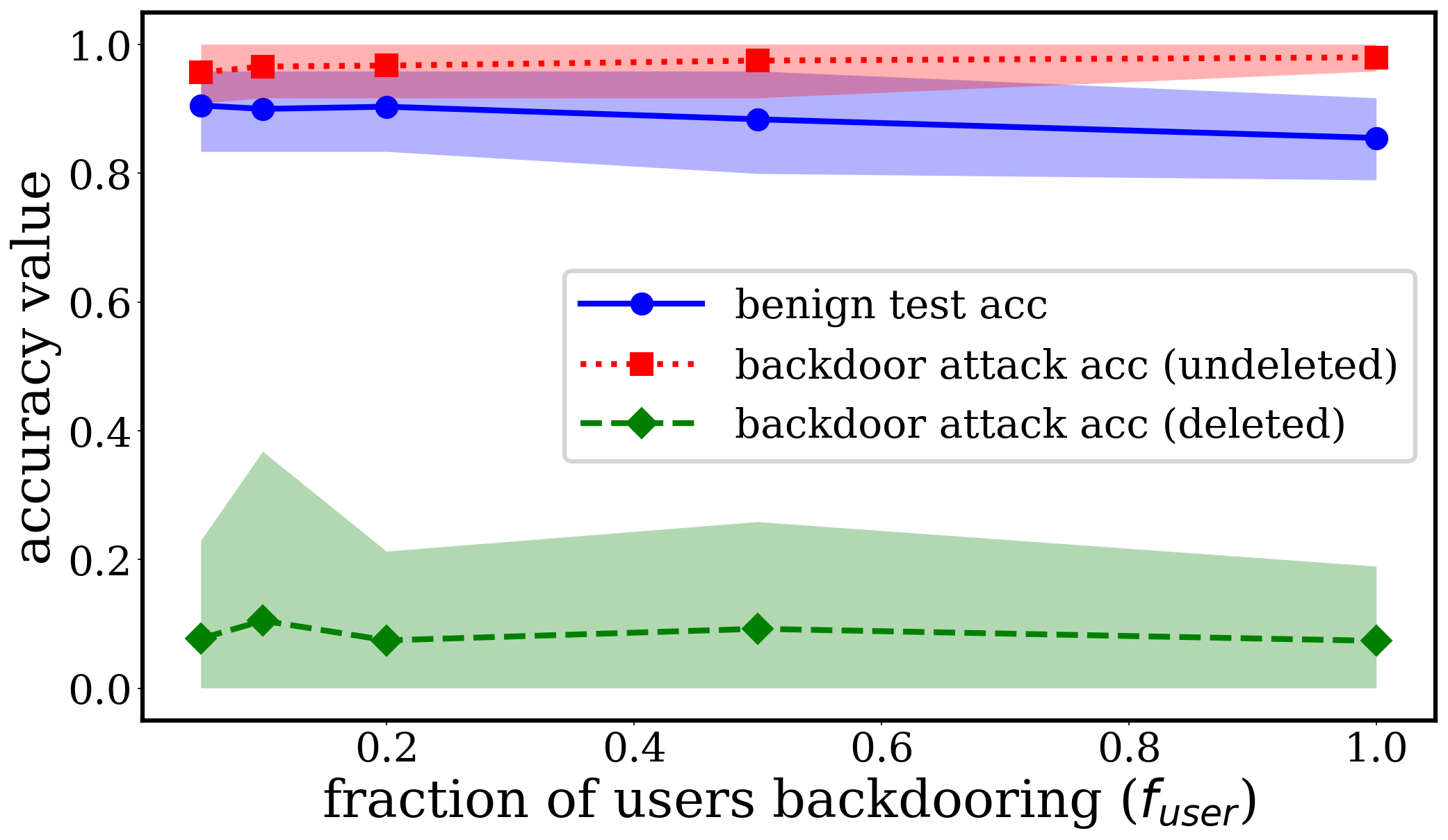}\vspace{\spacehacko}
	\end{subfigure}\hfill
	\par\medskip\vspace{\spacehack}
	%%%%%%%%%%%%%%%%%%%%%%%%%%%%%%%%%% ImageNet
	\verticaltext{~~~~~~~~~~~ImageNet}
		\begin{subfigure}[t]{0.31\linewidth}
		\raggedleft
		\includegraphics[width=\linewidth]{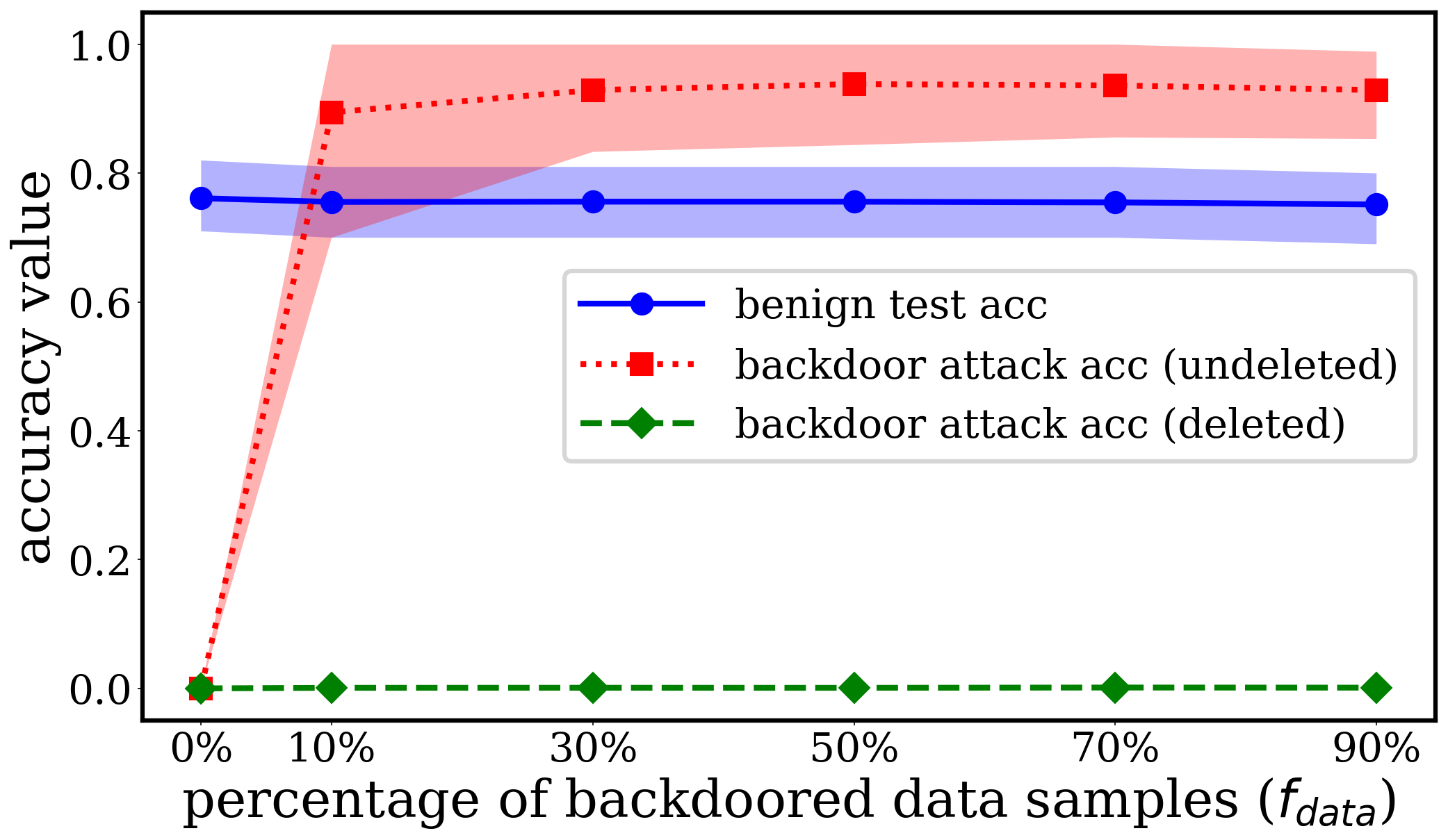}\vspace{\spacehacko}
	\end{subfigure}\hfill
	\begin{subfigure}[t]{0.31\linewidth}
		\centering
		\includegraphics[width=\linewidth]{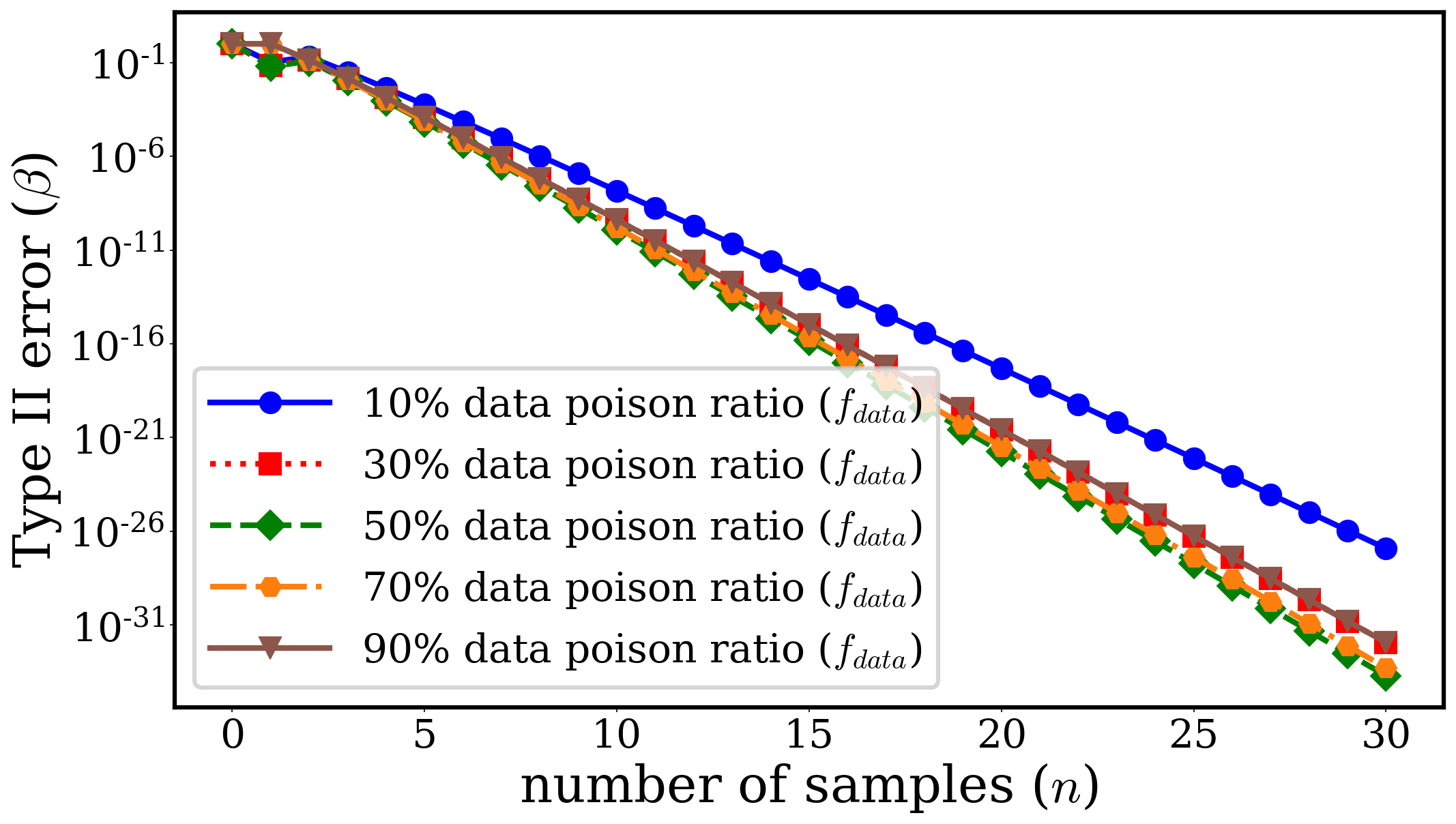}\vspace{\spacehacko}
	\end{subfigure}\hfill
	\begin{subfigure}[t]{0.31\linewidth}
		\raggedright
		\includegraphics[width=\linewidth]{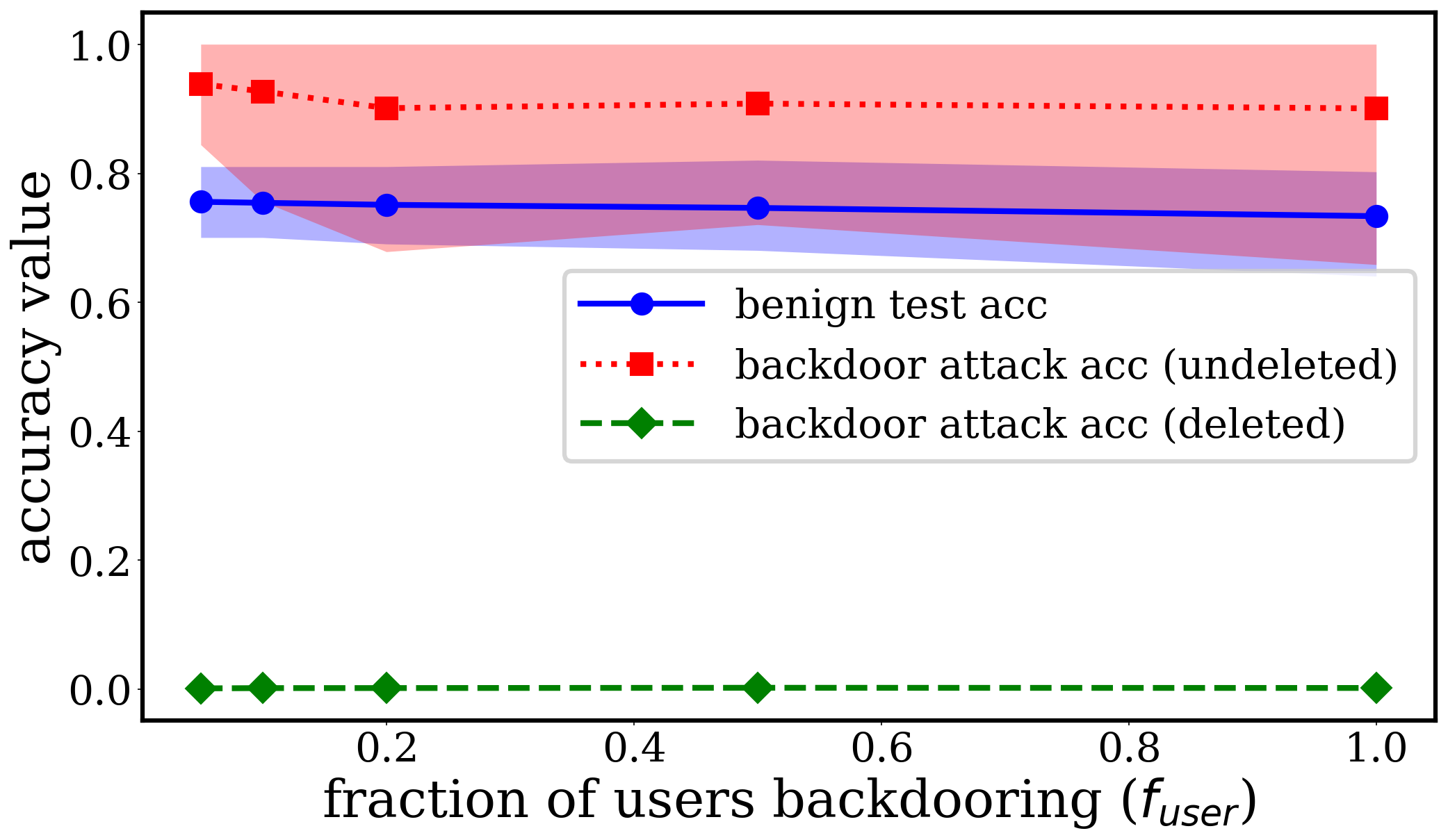}\vspace{\spacehacko}
	\end{subfigure}\hfill	\par\medskip\vspace{\spacehack}
	%%%%%%%%%%%%%%%%%%%%%%%%%%%%%%%%%% AG News
	\verticaltext{~~~~~~~~~~~~AG News}
		\begin{subfigure}[t]{0.31\linewidth}
		\raggedleft
		\includegraphics[width=\linewidth]{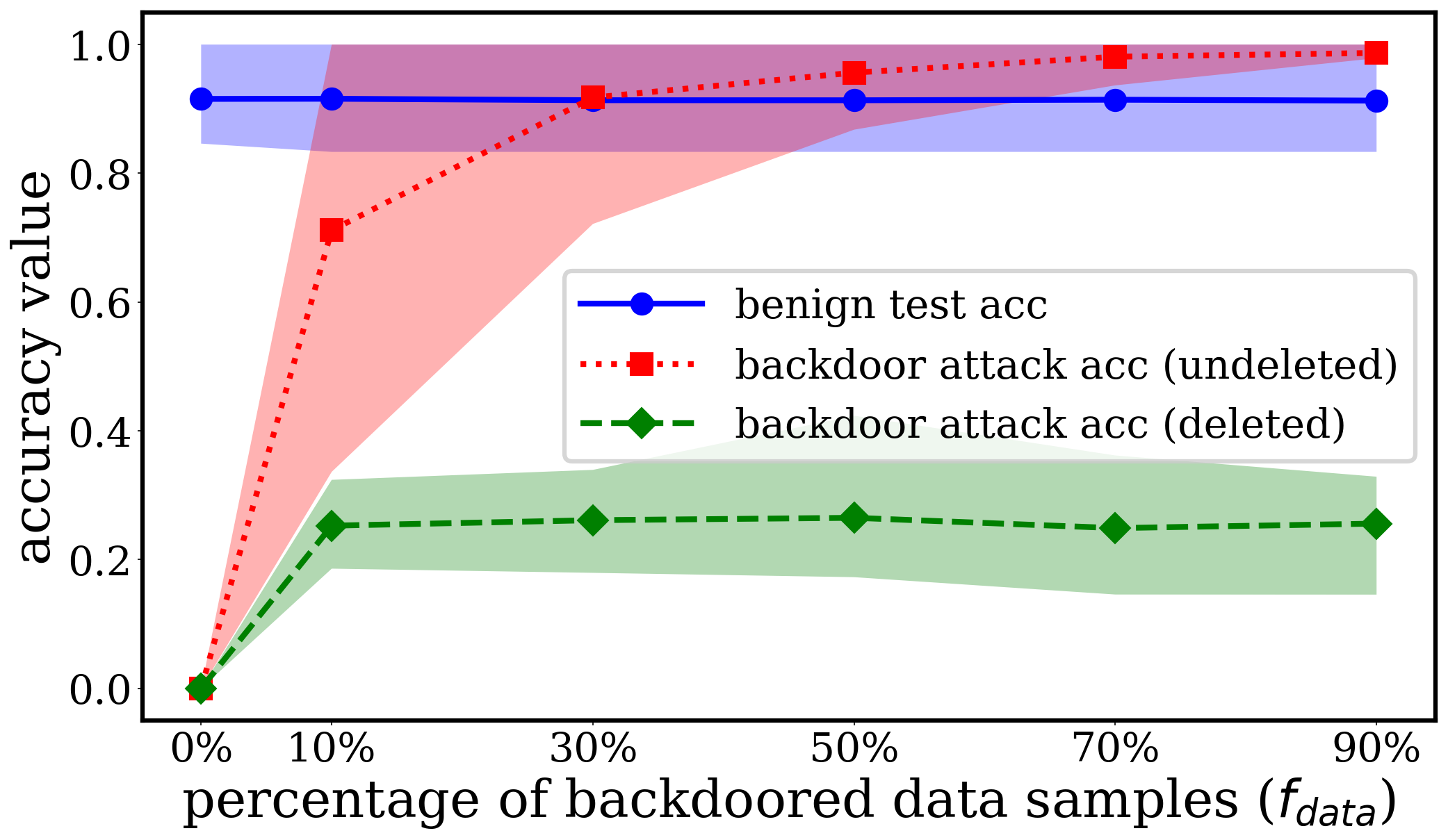}
		\caption{Model accuracy and backdoor success rate for poisoning user fraction $\userratio=0.05$.}
		\label{fig:nat_acc}
	\end{subfigure}\hfill
	\begin{subfigure}[t]{0.31\linewidth}
		\centering
		\includegraphics[width=\linewidth]{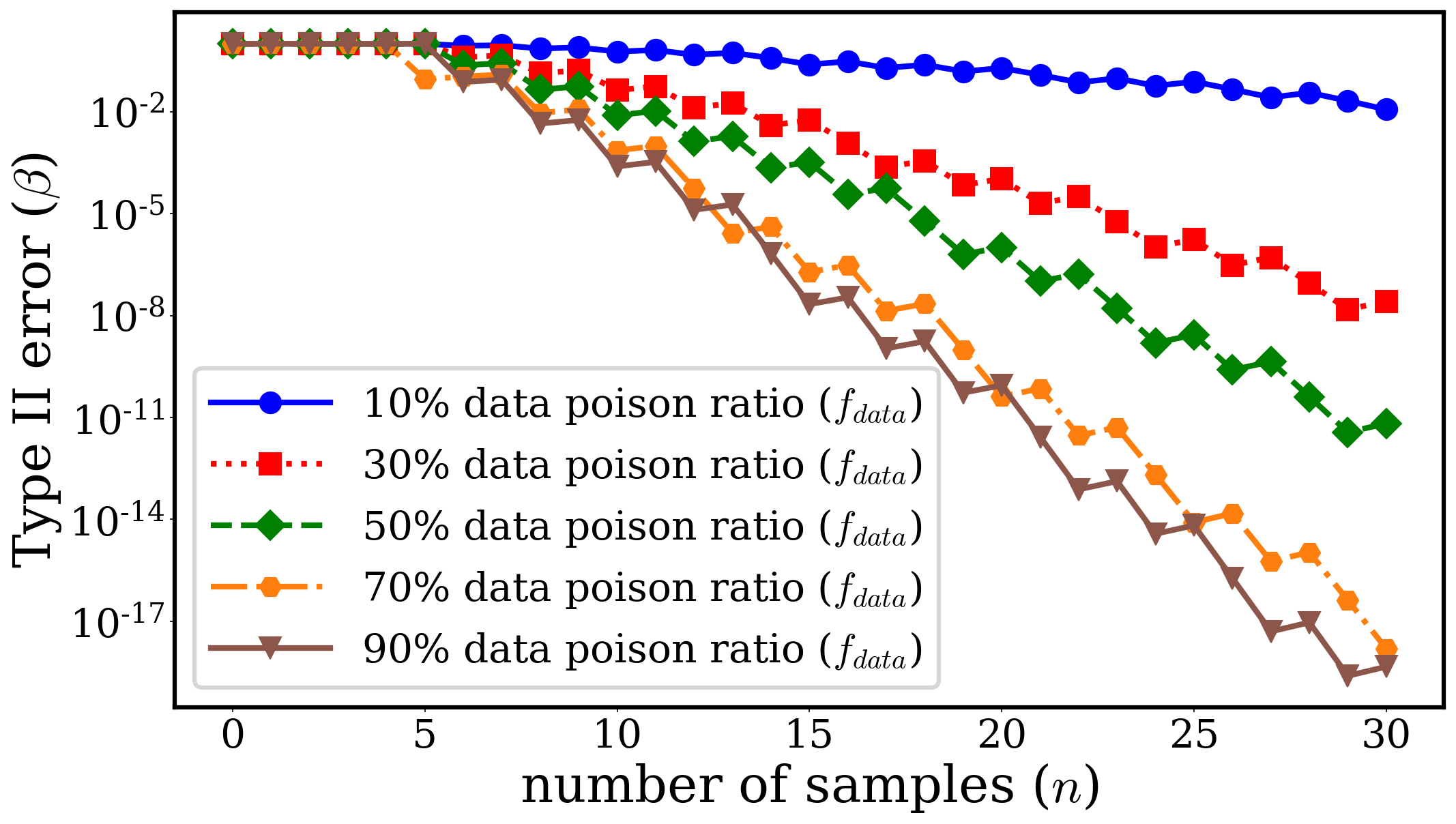}
		\caption{Verification performance with different poison ratios \poisonratio{} for a fixed $\alpha=10^{-3}$.
		%We fix the $\alpha$ to be $10^{-5}$.
		}
		\label{fig:nat_verify_ratio}
	\end{subfigure}\hfill
	\begin{subfigure}[t]{0.31\linewidth}
		\raggedright
		\includegraphics[width=\linewidth]{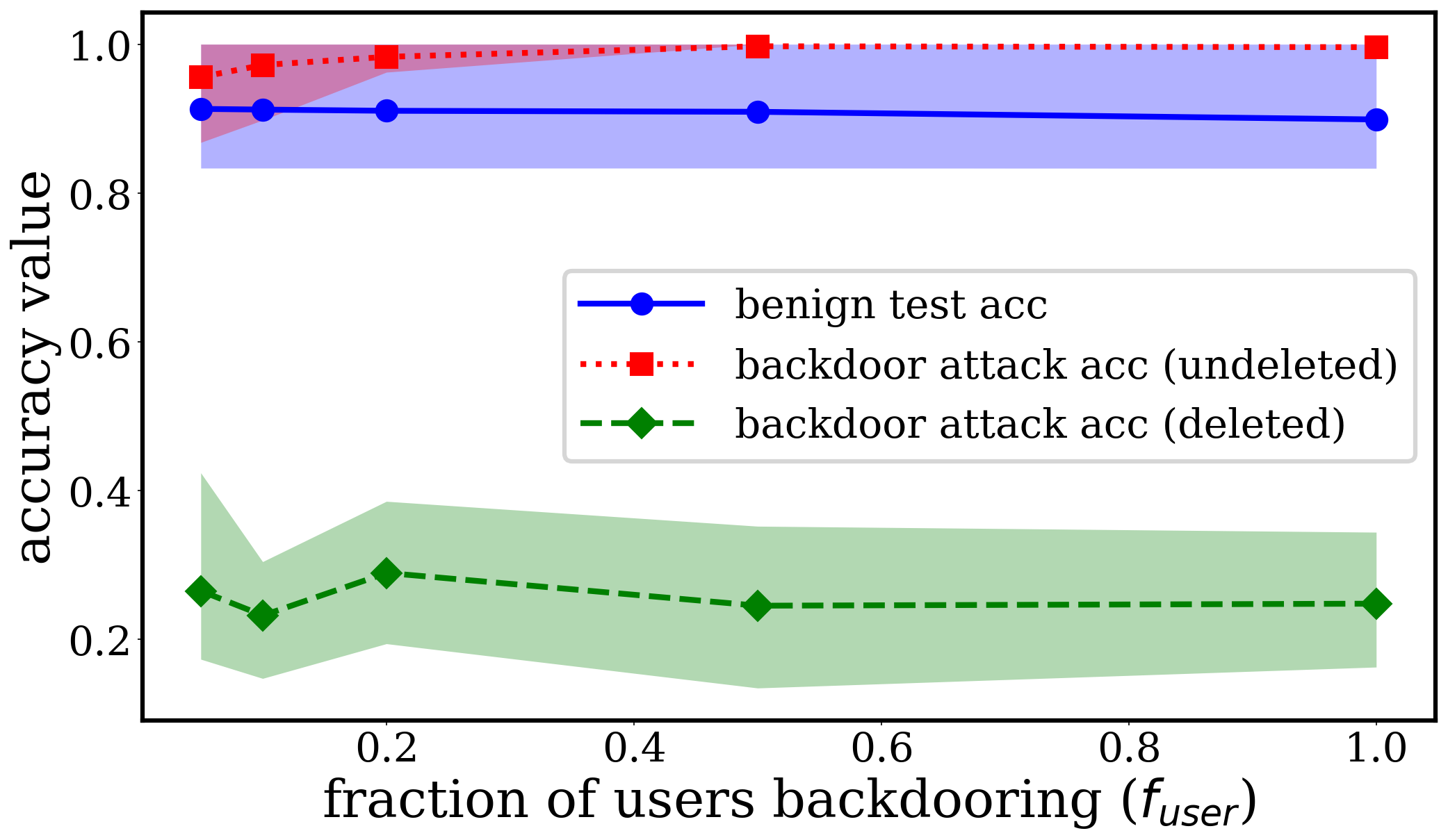}
		\caption{Model accuracy and backdoor success rate for fixed poison ratio $\poisonratio{}=50\%$.}
		\label{fig:nat_all_f}
	\end{subfigure}\hfill
	\caption{Our backdoor-based machine unlearning verification results with a \textit{\Naive{} server}. Each row of plots is evaluated on the data-set specified at the most-left position. Each column of plots depicts the evaluation indicated in the caption at its bottom. The colored areas in columns (a) and (c) tag the 10\% to 90\% quantiles.
	%
	% Our backdoor-based machine unlearning verification results with a \Naive{} server on the EMNIST dataset (the first row), the FEMNIST dataset (the second row), the CIFAR10 dataset (the third row), and the AG News dataset (the last row) with 2\% author poison ratio. We present the model accuracy and backdoor success rate in the first column, the verification performance with Type-I error $\alpha=10^{-3}$ in the second column, and the verification performance with a poison ratio of $50\%$ in the third column.
	}
	\label{fig:combined_natural}
\end{figure*}
\fi

\subsection{Results for the Non-Adaptive Server}\label{sec:evaluation:naiveserver}

\ifplotimages
\begin{figure*}[!ht]
	\centering
    %%%%%%%%%%%%%%%%%%%%%%%%%%%%%%%%%%%%EMNIST
    \verticaltext{~~~~~~~~~~~~EMNIST}
	\begin{subfigure}[t]{0.31\linewidth}
		\raggedleft
		\includegraphics[width=\linewidth]{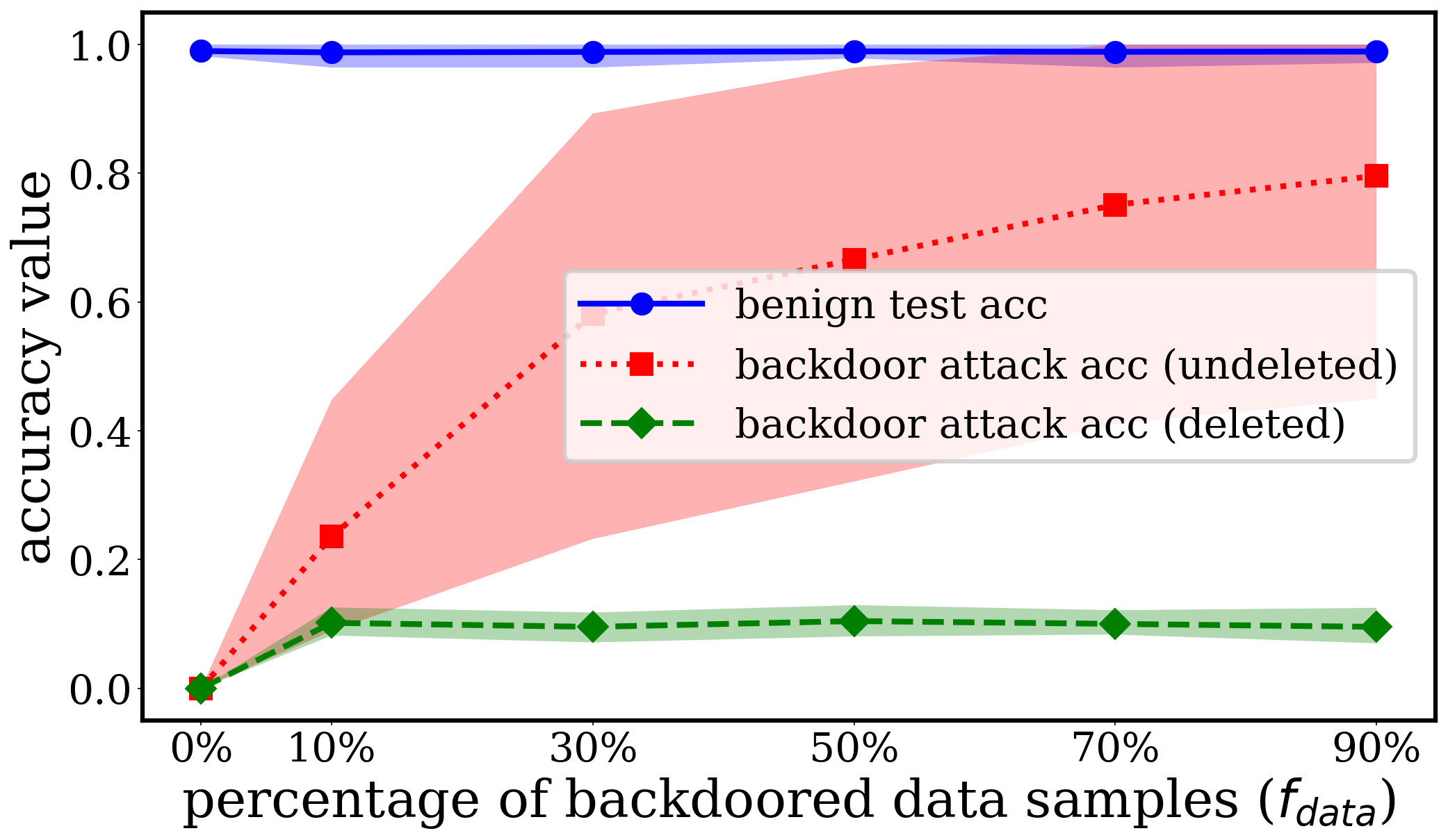}\vspace{\spacehacko}
	\end{subfigure}\hfill
	\begin{subfigure}[t]{0.31\linewidth}
		\centering
		\includegraphics[width=\linewidth]{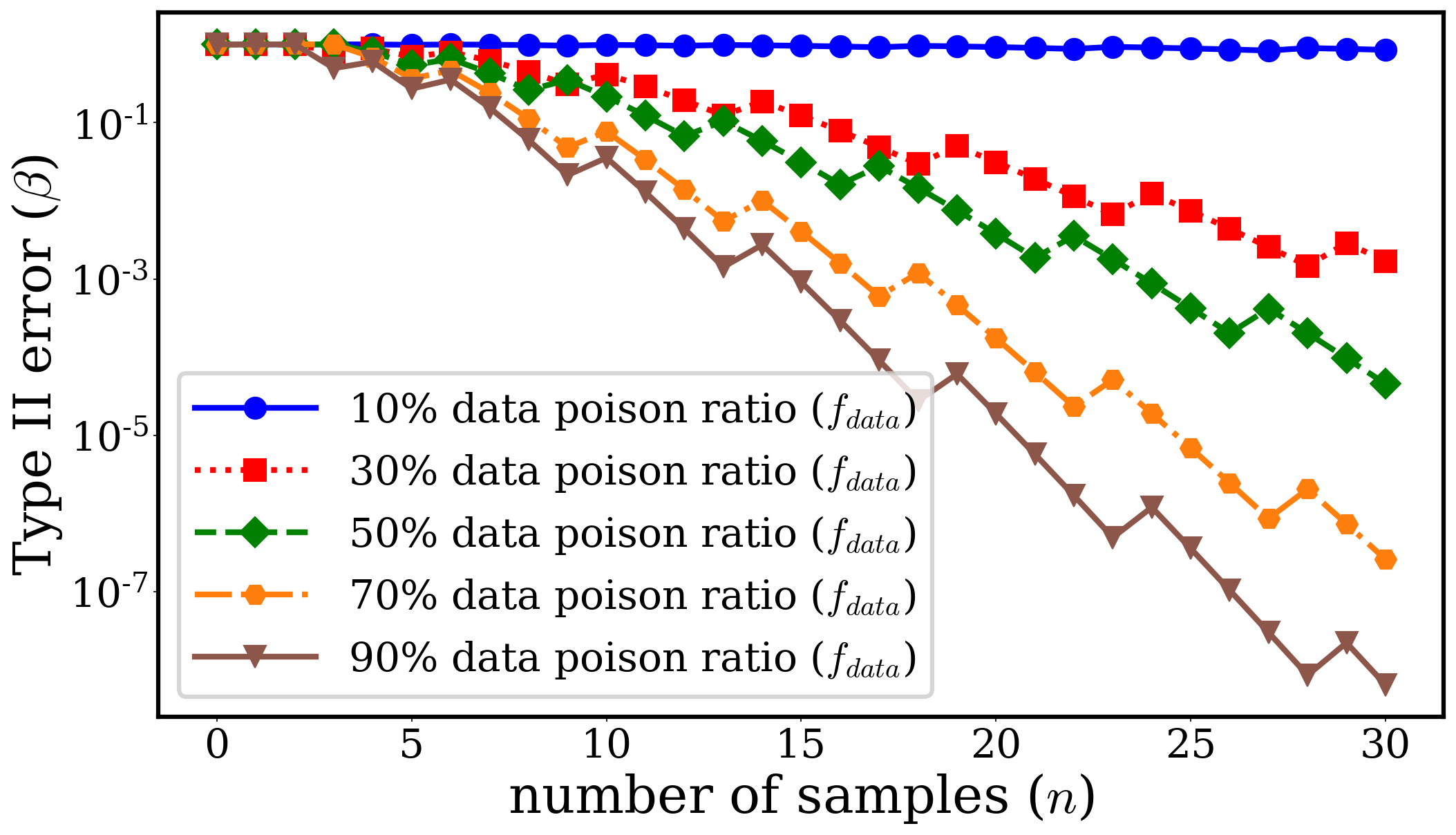}\vspace{\spacehacko}
	\end{subfigure}\hfill
	\begin{subfigure}[t]{0.31\linewidth}
		\raggedright
		\includegraphics[width=\linewidth]{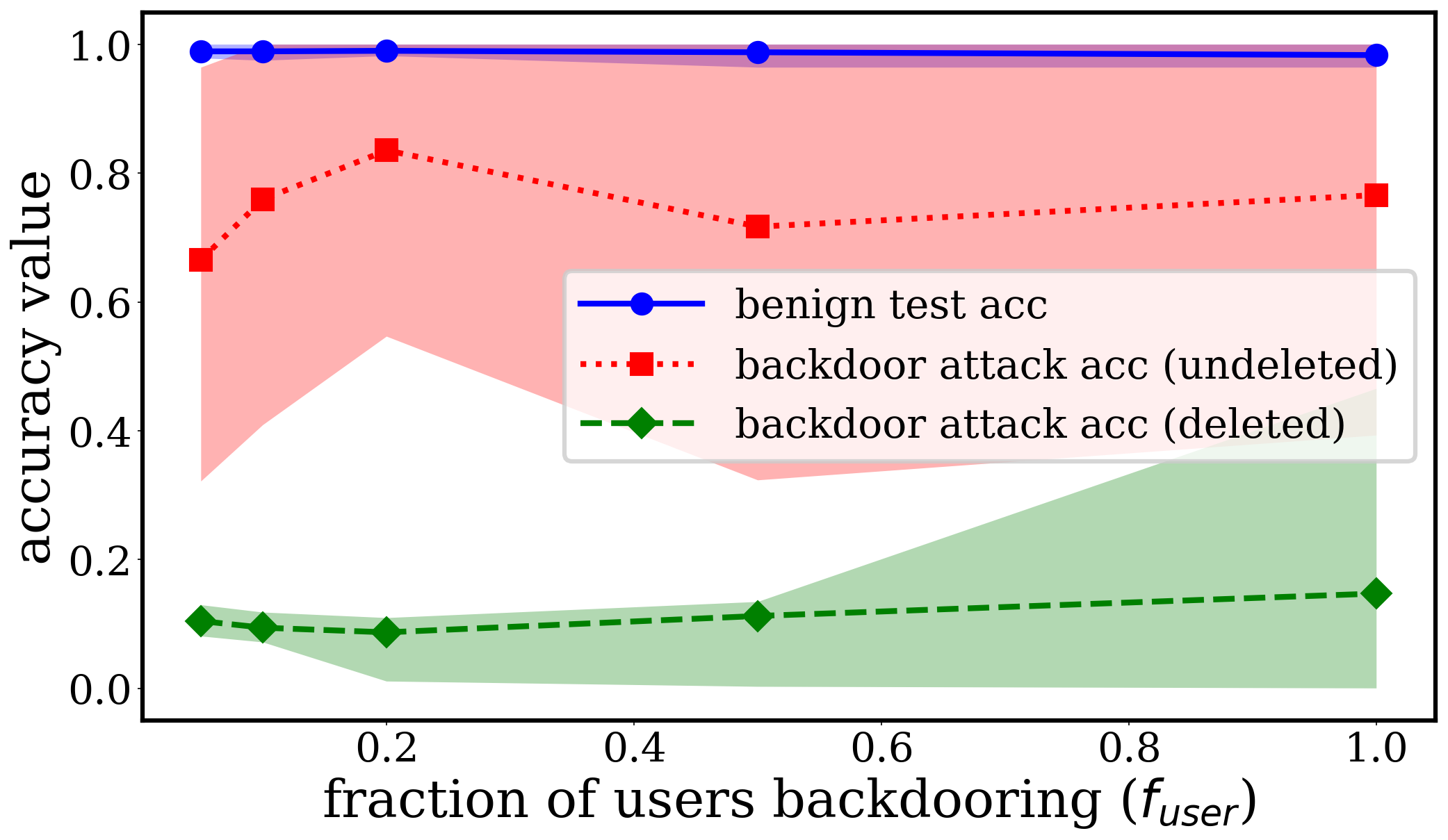}\vspace{\spacehacko}
	\end{subfigure}\hfill
	\par\medskip\vspace{\spacehack}
    %%%%%%%%%%%%%%%%%%%%%%%%%%%%%%%%%%%%FEMNIST
    \verticaltext{~~~~~~~~~~~FEMNIST}
	\begin{subfigure}[t]{0.31\linewidth}
		\raggedleft
		\includegraphics[width=\linewidth]{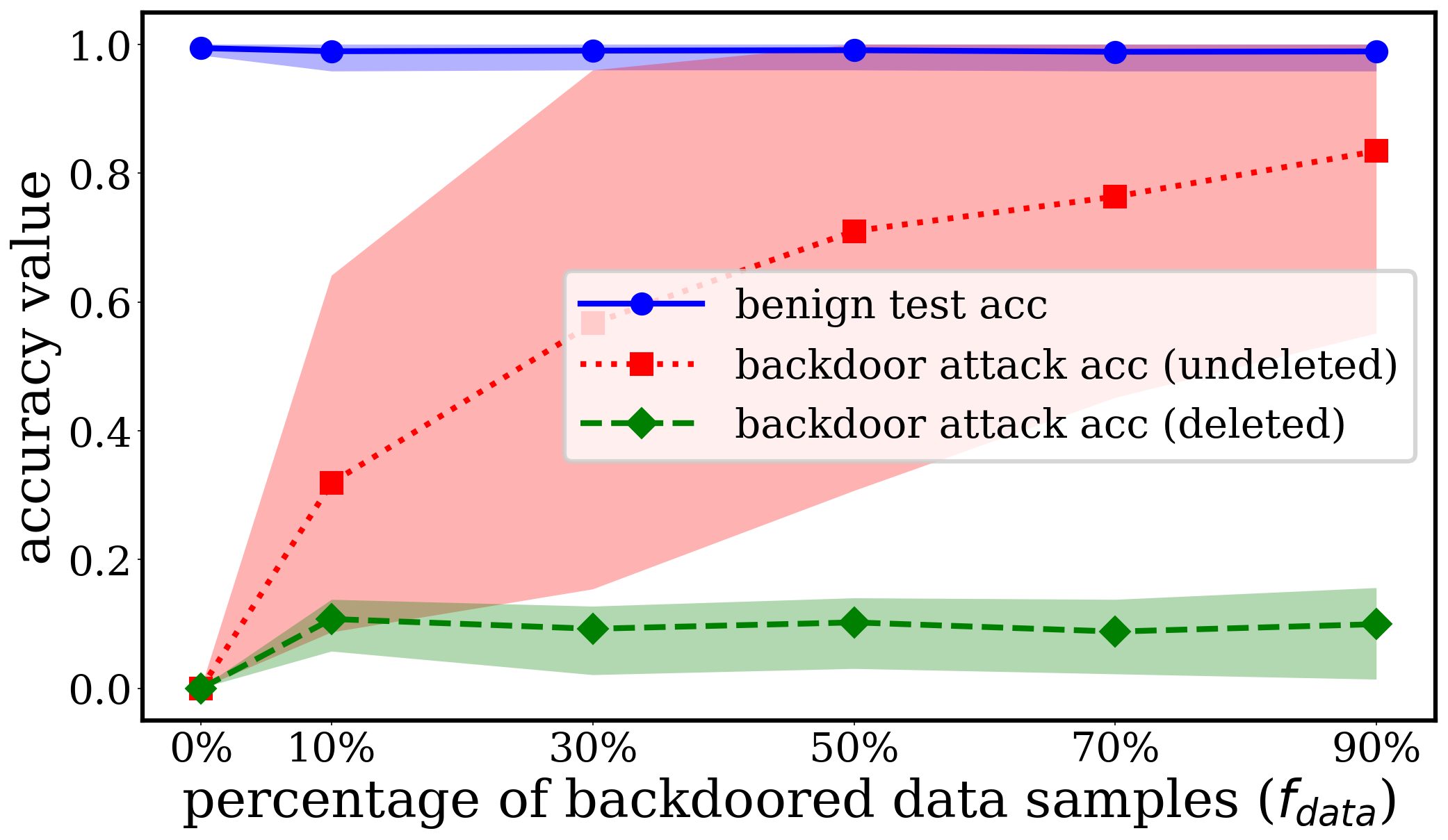}\vspace{\spacehacko}
	\end{subfigure}\hfill
	\begin{subfigure}[t]{0.31\linewidth}
		\centering
		\includegraphics[width=\linewidth]{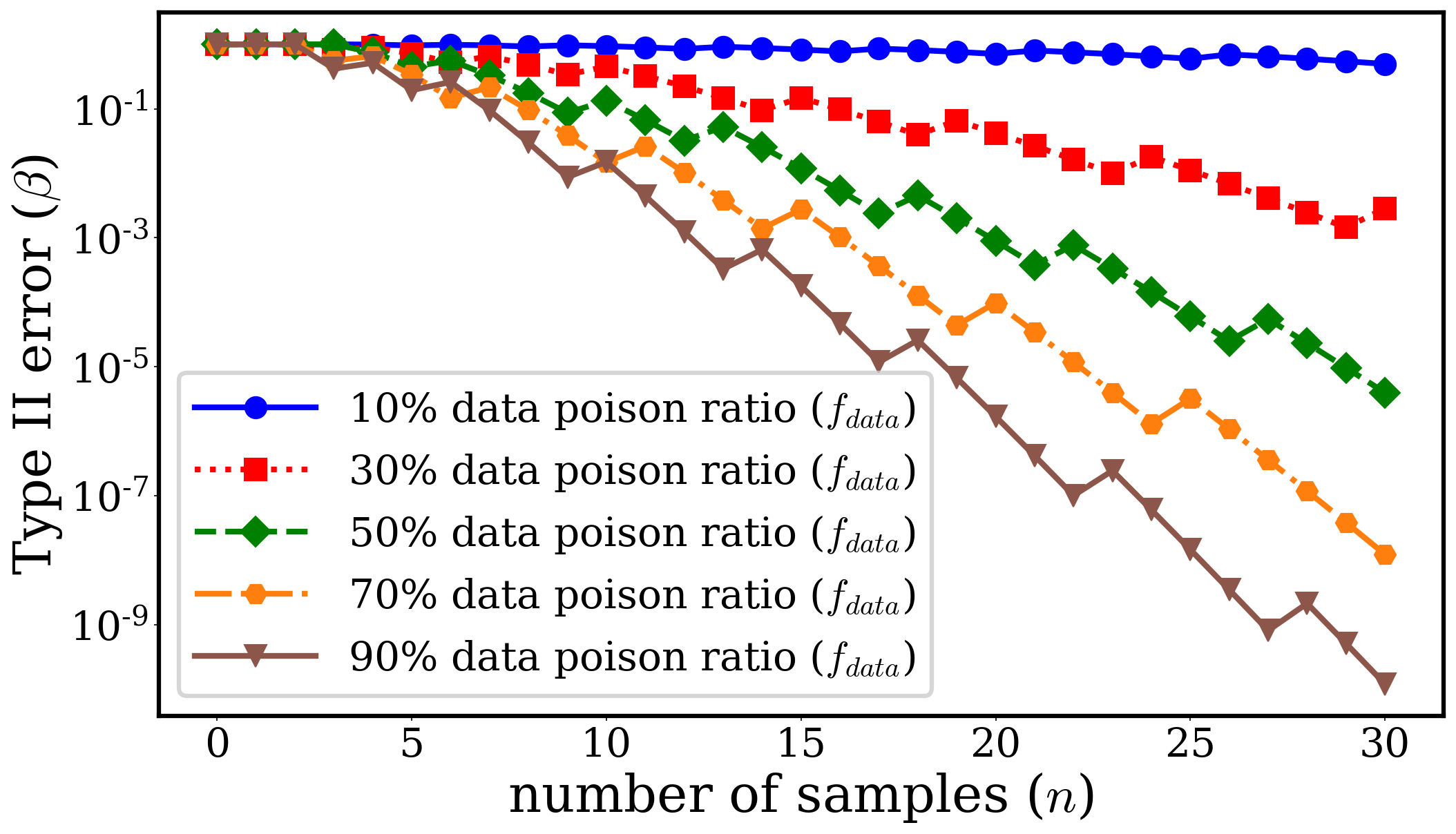}\vspace{\spacehacko}
	\end{subfigure}\hfill
	\begin{subfigure}[t]{0.31\linewidth}
		\raggedright
		\includegraphics[width=\linewidth]{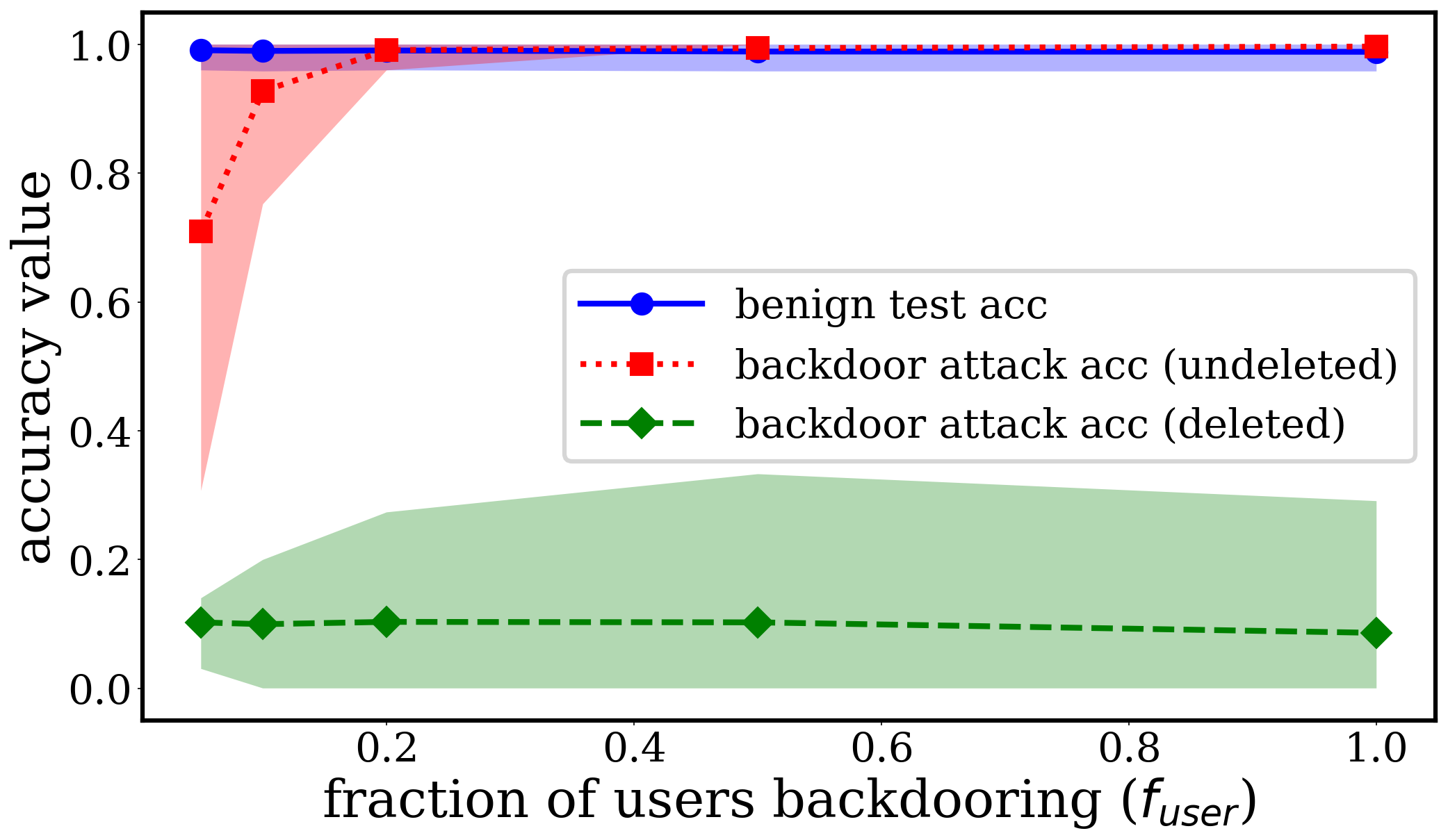}\vspace{\spacehacko}
	\end{subfigure}\hfill
	\par\medskip\vspace{\spacehack}
    %%%%%%%%%%%%%%%%%%%%%%%%%%%%%%%%%%%%CIFAR
    \verticaltext{~~~~~~~~~~~CIFAR10}
	\begin{subfigure}[t]{0.31\linewidth}
		\raggedleft
		\includegraphics[width=\linewidth]{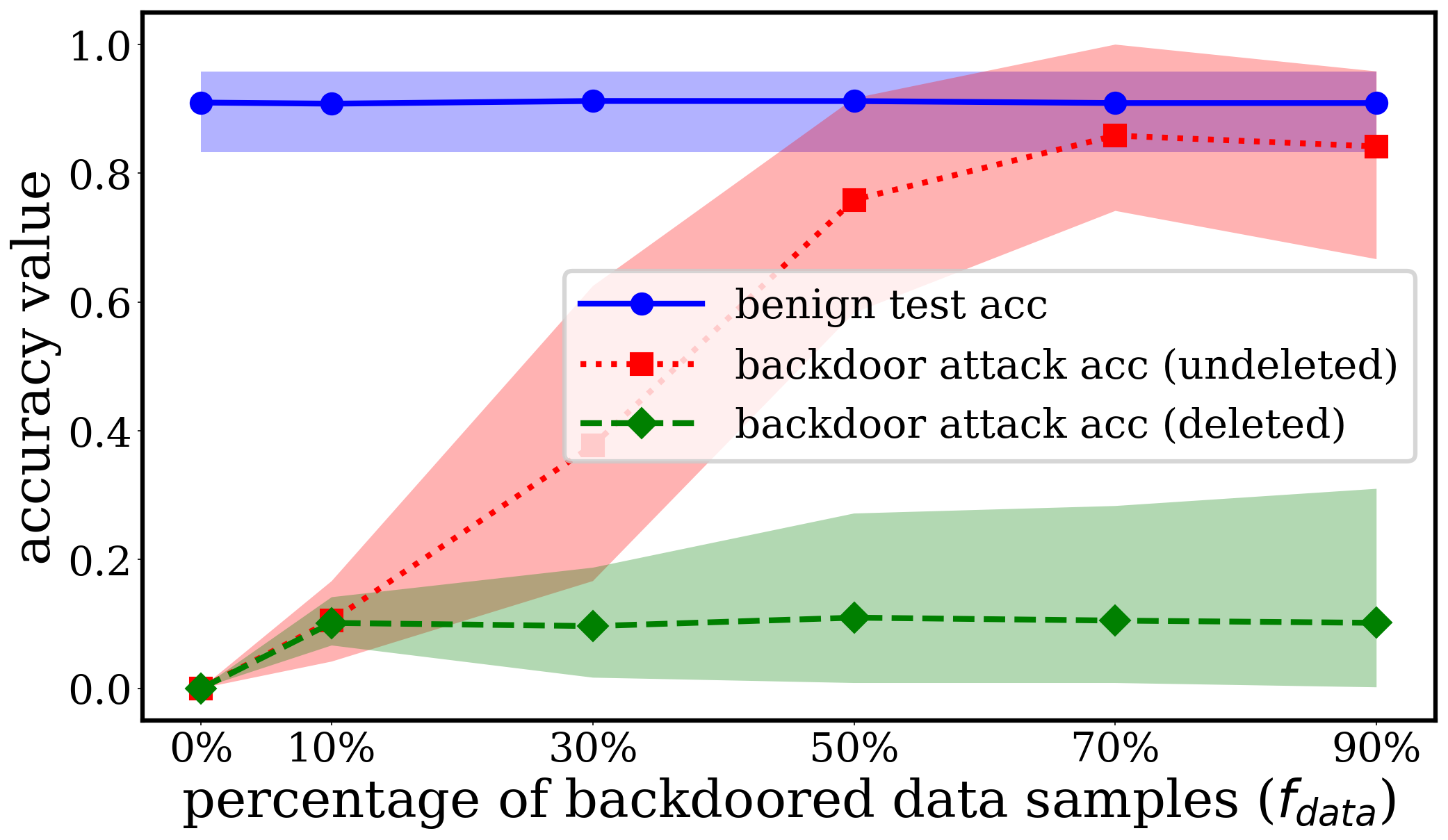}\vspace{\spacehacko}
	\end{subfigure}\hfill
	\begin{subfigure}[t]{0.31\linewidth}
		\centering
		\includegraphics[width=\linewidth]{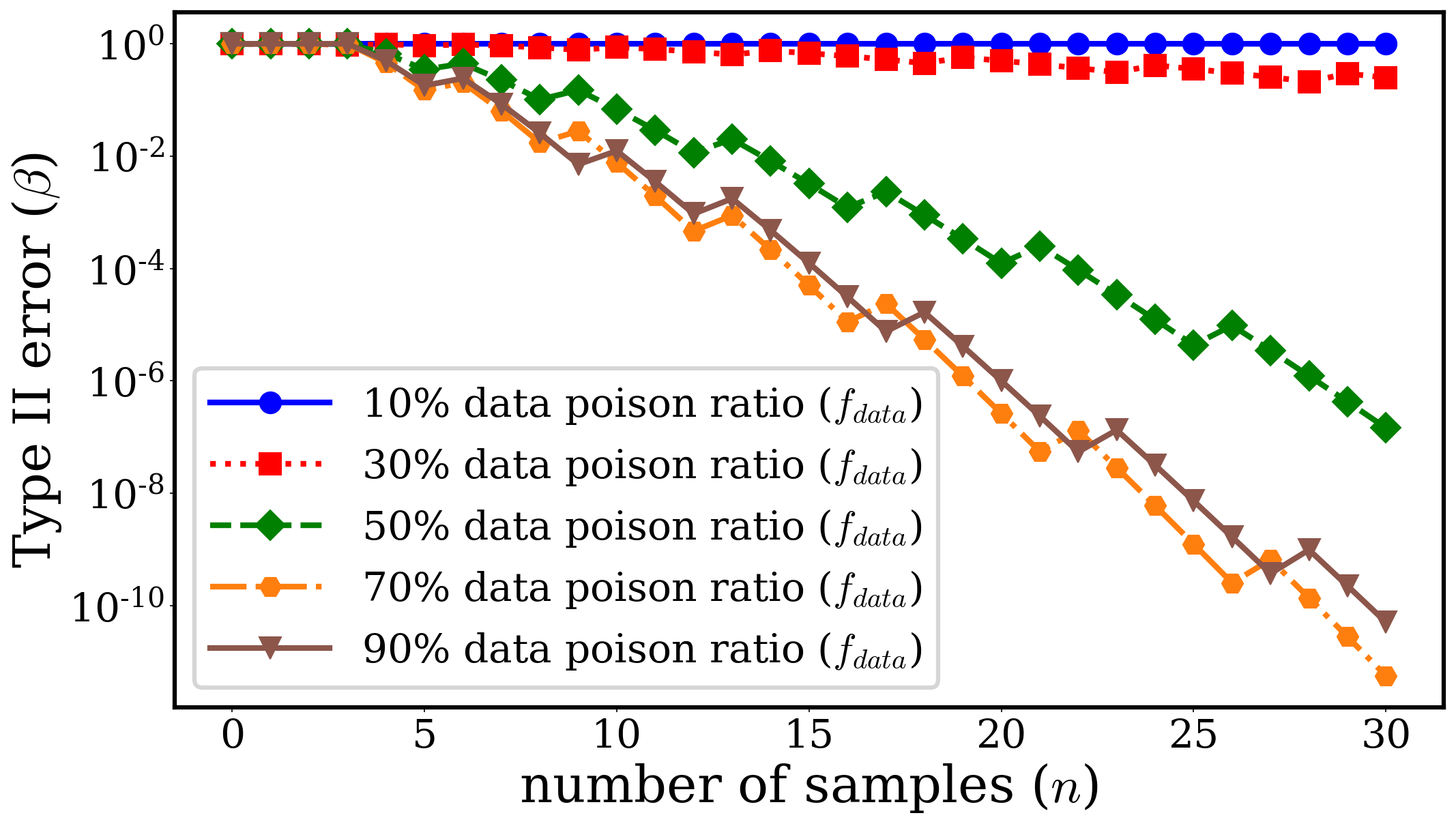}\vspace{\spacehacko}
	\end{subfigure}\hfill
	\begin{subfigure}[t]{0.31\linewidth}
		\raggedright
		\includegraphics[width=\linewidth]{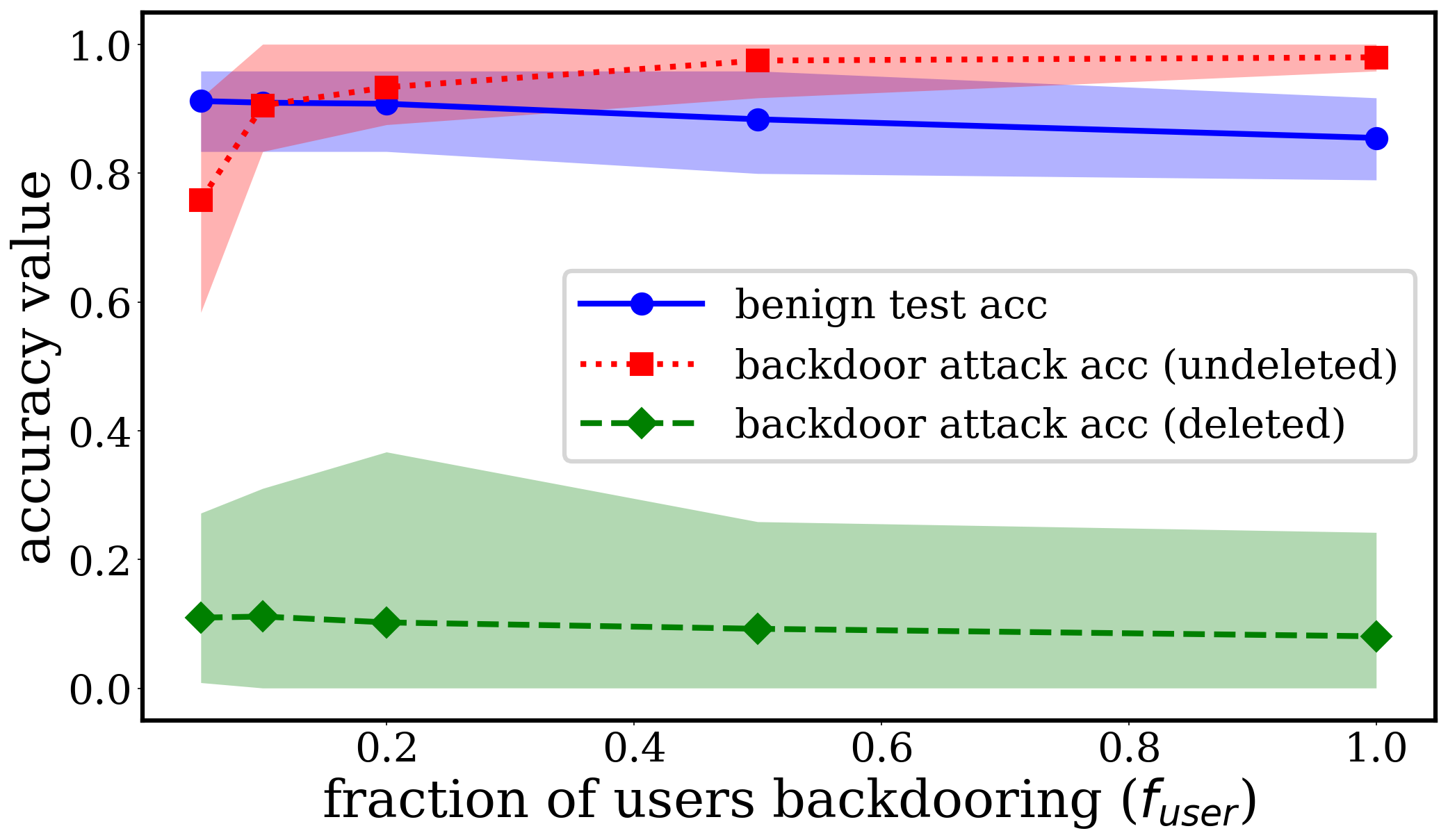}\vspace{\spacehacko}
	\end{subfigure}\hfill
	\par\medskip\vspace{\spacehack}
    %%%%%%%%%%%%%%%%%%%%%%%%%%%%%%%%%%%%Imagenet
    \verticaltext{~~~~~~~~~~~ImageNet}
	\begin{subfigure}[t]{0.31\linewidth}
		\raggedleft
		\includegraphics[width=\linewidth]{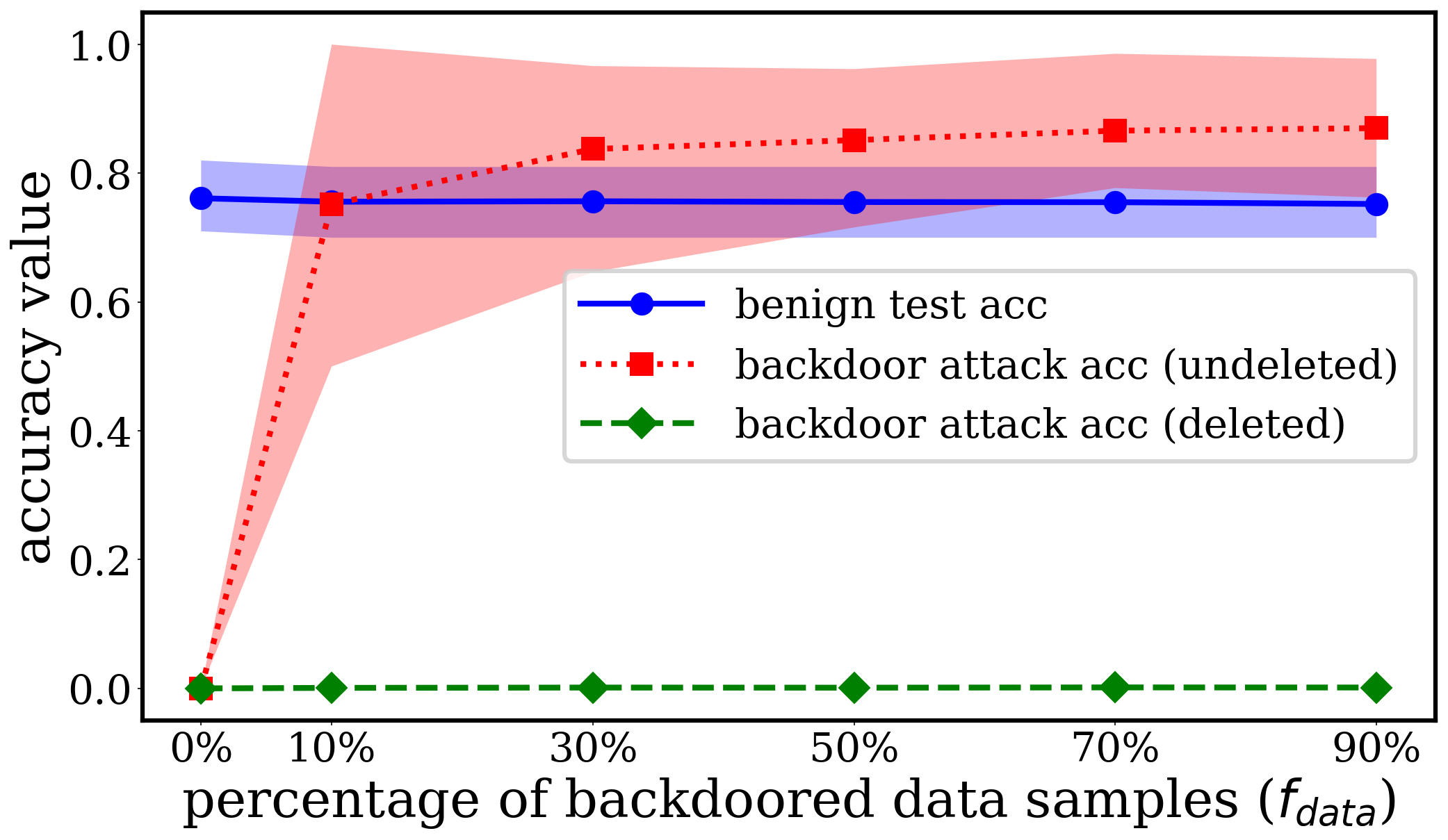}\vspace{\spacehacko}
		\caption{Model accuracy and backdoor success rate for fixed user poison fraction $\userratio=0.05$.}
		\label{fig:def_acc}
	\end{subfigure}\hfill
	\begin{subfigure}[t]{0.31\linewidth}
		\centering
		\includegraphics[width=\linewidth]{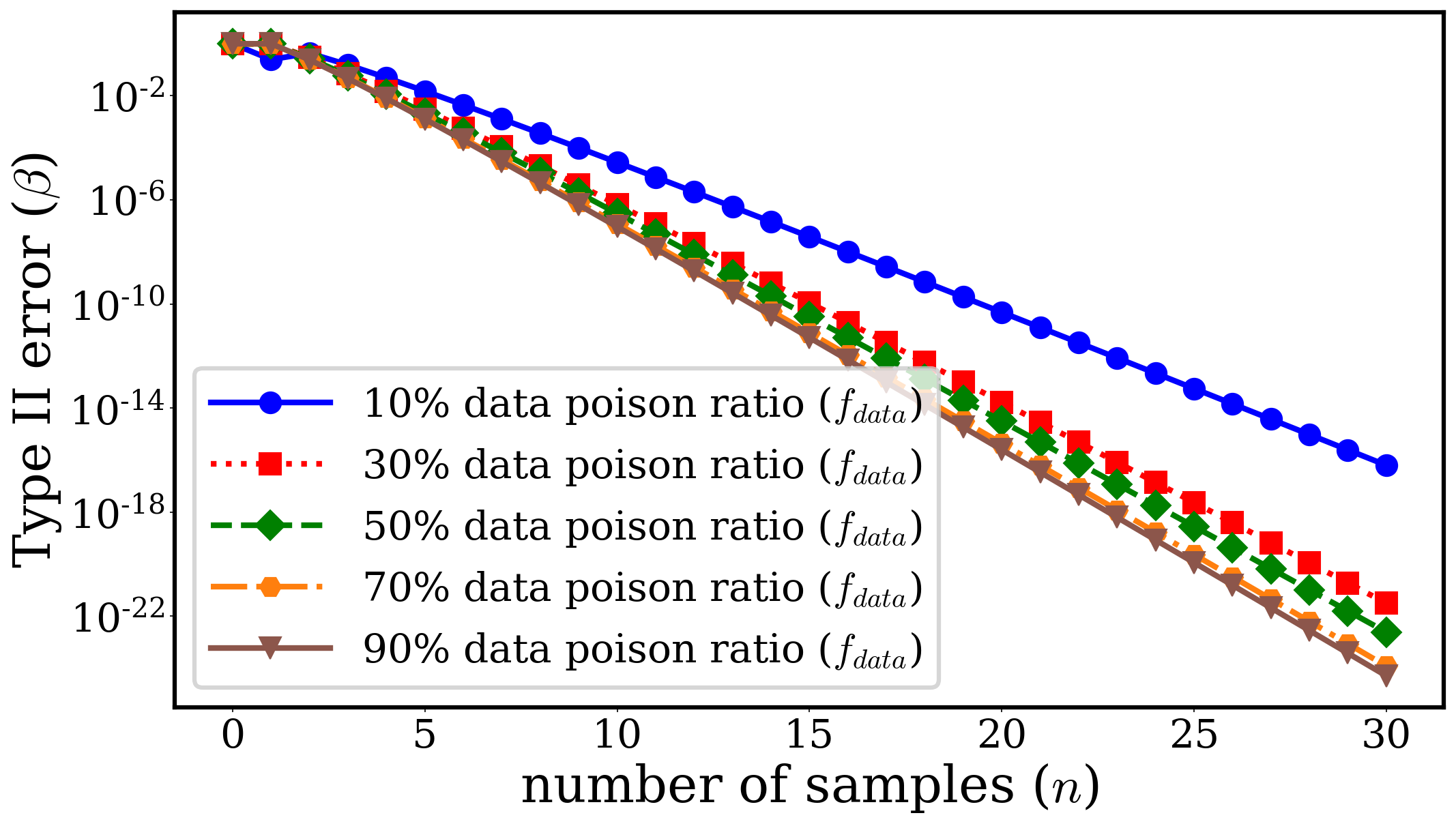}\vspace{\spacehacko}
		\caption{Verification performance with different poison ratios \poisonratio{} for a fixed $\alpha=10^{-3}$.}
		\label{fig:def_verify_ratio}
	\end{subfigure}\hfill
	\begin{subfigure}[t]{0.31\linewidth}
		\raggedright
		\includegraphics[width=\linewidth]{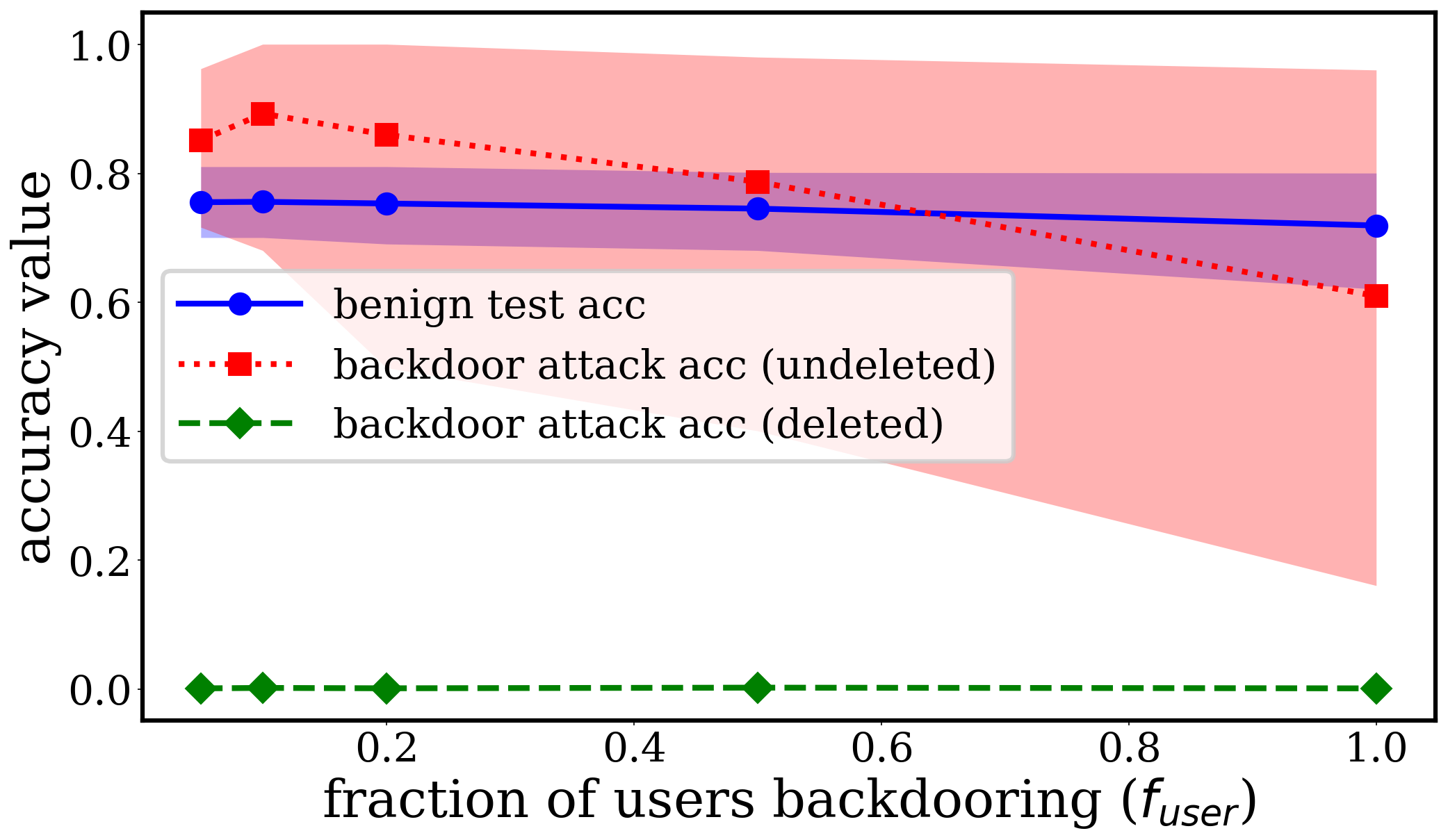}\vspace{\spacehacko}
		\caption{Model accuracy and backdoor success rate for fixed poison ratio $\poisonratio{}=50\%$.}
		\label{fig:def_all_f}
	\end{subfigure}\hfill
	\caption{Our backdoor-based machine unlearning verification results with the \textit{\Advanced{} server}. Each row of plots is evaluated on the data-set specified at the most-left position. Each column of plots depicts the evaluation indicated in the caption at its bottom. The AG News dataset is omitted as Neural Cleanse is not applicable for non-continuous datasets. The colored areas in columns (a) and (c) tag the 10\% to 90\% quantiles.
	% Our backdoor-based machine unlearning verification results with an \Advanced{} server on the EMNIST dataset (the first row), the FEMNIST dataset (the second row), and the CIFAR10 dataset (the third row) with 2\% author poison ratio.We present the model accuracy and backdoor success rate in the first column, the verification performance with Type-I error $\alpha=10^{-3}$ in the second column, and the verification performance with a poison ratio of $50\%$ in the third column.
	}
	\label{fig:combined_defended}\vspace{\spacehack}
\end{figure*}

\fi

% \ifplotimages
% \input{Images/experiment_results/text_files/EMNIST_natural}
% \input{Images/experiment_results/text_files/FEMNIST_natural}
% \input{Images/experiment_results/text_files/CIFAR_natural}
% \input{Images/experiment_results/text_files/AGNews_natural}
% \fi

We first present the evaluation results for the \Naive{} server, where the server uses the non-adaptive learning algorithm to train the ML model.
On each dataset, we compute the backdoor success rate for each privacy enthusiasts, and compute the undeleted users' average success rate $p$ and deleted users' average success rate as $q$ to evaluate the performance of our machine unlearning verification method with different numbers $n$ of test queries, following Theorem \ref{thm:computerho}. 
%\ds{For this evaluation, we fixed the fraction of privacy enthusiasts \userratio{} to 0.05 if not indicated otherwise.}

First, \textbf{our verification mechanism works well with high confidence on the EMNIST dataset.}
From \cref{fig:nat_acc} for EMNIST, we can see that the attack accuracy for undeleted users ($p$) increases with the poison ratio, while the attack accuracy for deleted users ($q$) stays around $10\%$ (random guess accuracy).
At the same time, poison ratios as high as $90\%$ have negligible impact on model accuracy for clean test samples.
Besides plotting the average accuracy across users, we also show the $10\%$--$90\%$ quantile ranges for individual users' accuracy values, examined in more detail in \cref{sec:individual_evaluation}.
\Cref{fig:nat_verify_ratio} for EMNIST shows that a higher poison ratio also leads to a lower Type-II error in our verification mechanism due to a larger gap between $p$ and $q$, therefore increasing the confidence in distinguishing between $H_{0}$ and $H_{1}$.
%. Also the confidence in distinguishing between $H_{0}$ and $H_{1}$ increases with more test samples.
For curious readers, we further show the verification performance with different the Type-I error tolerances ($\alpha$) for a fixed data poison ratio \poisonratio{} of $50\%$ in appendix \cref{fig:nat_verify_alpha}.
%Specifically, setting $\alpha = 10^{-5}$, our verification mechanism gives a $6.8 \times 10^{-9}$ Type-II error ($\beta$) with 20 test samples.

Given a Type-I error $\alpha$ and a fixed number of test samples $n$, the computation of the type II error $\beta$ is discrete as the corresponding rescaled binomial distribution is discrete and thus the probability mass in the tails usually does not sum up exactly to $\alpha$ (cf. \cref{sec:analysis}). This property leads to jumps when plotting $\beta$ over different $n$ as the probability mass of a discrete event might not be included in the tail with maximal size $\alpha$ anymore, and be added to $\beta$ instead. This is the case for all plots evaluating $\beta$ over different number of test samples, not only for EMNIST. For ImageNet, however, the extreme small $q$ value reduces the effect, leading to almost straight lines.

Second, \textbf{our verification mechanism generalizes to more complex image datasets.}
The accuracy performance and the verification performance for the non-IID FEMNIST dataset, the CIFAR10 dataset, and the much more complex ImageNet dataset is presented in \cref{fig:combined_natural} as well.
Similar to EMNIST, the gap between $p$ and $q$ becomes larger when increasing the poison ratio. \ds{For ImageNet, the backdoor attack accuracy for deleted users and its 80\% quantile are comparatively small as its number of prediction classes, namely 1000 compared to 10 or 4 for the other datasets, reduces the probability for accidental backdoor collision significantly.}
%\red{CHECK THAT ONCE WE HAVE FINAL DATA: The only exception is that the CIFAR10 classifier has a much lower $p$ value with $90\%$ poison ratio along with a low accuracy on the clean samples.
%Furthermore, our verification mechanism achieves more confidence on those two datatsets than the EMNIST dataset.}
% When using a $30\%$ poison ratio and 20 test samples, we achieve $\beta = 3.2 \times 10^{-37}$ for the FEMNIST CNN classifier and $3.0 \times 10^{-12}$ for the CIFAR10 ResNet classifier by setting $\alpha$ as $10^{-5}$.

Third, \textbf{our verification mechanism is also applicable to non-image datasets}, illustrated for the AG News dataset from \Cref{fig:nat_acc}.
The undeleted users' backdoor attack accuracy is around $100\%$ with a poison ratio greater than $30\%$, while the deleted users' backdoor attack accuracy stays around $25\%$ (random guess accuracy).
%Similar to image datasets, we present the verification Type-II error ($\beta$) with different poison ratios in \cref{fig:nat_agnews_verify_ratio} and with different false alarm tolerances ($\alpha$) in \cref{fig:nat_agnews_verify_alpha}.
%Specifically, for the AG News classifier with a $50\%$ poison ratio and $30$ test samples, we achieve $\beta = \red{8.6} \times 10^{-10}$ with $\alpha = 10^{-3}$.

Fourth, \textbf{our mechanism works for arbitrary fraction \userratio{} of privacy enthusiasts testing for deletion verification}, illustrated in \cref{fig:nat_all_f}. Previously, we only considered a user backdooring rate $f$ of 0.05. While such a scenario is more realistic, i.e., when only a few privacy enthusiast test a ML provider for deletion validity, it might happen that more users test. Even when all users are testing, the benign accuracy of the models are barely impacted.%reduces only a few percentages. 
For EMNIST, a larger \userratio{} leads to an increased number of backdoor collisions: While the average of the backdoor attack accuracy for deleted users increases only minimally, a few deleted users measure a high success rate. We discuss mitigation strategies in \cref{sec:individual_evaluation}. 

\subsection{Results for the \Advanced{} Server}\label{sec:evaluation:advancedserver}

% \ifplotimages
% \input{Images/experiment_results/text_files/EMNIST_defended}
% \input{Images/experiment_results/text_files/FEMNIST_defended}
% \input{Images/experiment_results/text_files/CIFAR_defended}
% \fi

We choose the state-of-the-art backdoor defense method, Neural Cleanse \cite{backdoor_defense_wang_sp19}, for the \Advanced{} server.
Proposed by Wang et al. \cite{backdoor_defense_wang_sp19}, Neural Cleanse first reverse engineers the backdoor triggers by searching for minimum input perturbation needed for a target label classification, then uses the reverse engineered backdoor trigger to poison certain samples with correct labels to retrain the model for one step.
We extend this defense method to our setup where different users inject different backdoors into the model.
Specifically, the \Advanced{} server runs the Neural Cleanse defense method for individual user's data and augments the model sequentially with all the user-specific found reversed backdoor patterns.
%with the aim to make the backdoor injection unsuccessful. 
We provide results only for the four image datasets as Neural Cleanse needs gradient information with regard to input samples for reverse engineering the backdoor which is not available for the discrete AG News inputs. 

First, \textbf{the Neural Cleanse method reduces the backdoor attack success rate.}
As shown in \cref{fig:def_acc}, the undeleted users' backdoor attack accuracy ($p$) is reduced, especially for smaller data poison ratios $\poisonratio{}$.
%, and the model with $90\%$ poison ratio has $p$ value around $70\%$.
Nonetheless, a ratio $\poisonratio{} = 50\%$ gives still a significantly higher undeleted users' attack accuracy than its deleted counterpart.  
Thus \textbf{our verification can still have high confidence performance with enough test samples as illustrated in \cref{fig:def_verify_ratio}}.
%When using $80\%$ poison ratio and setting $\alpha = 10^{-3}$, we have $\beta = 0.019$ with 20 test samples.

% \red{REMOVE ONCE DECIDED WHETHER WE WANT TO TALK ABOUT MODEL CAPACITY: As shown later in Section \ref{subsec:capacity}, when switching the model architecture to CNN, the performance of Neural Cleanse is much worse and our verification mechanism with the \Advanced{} server succeeds.}

Second, \textbf{the Neural Cleanse method has limited defense performance on more complex image datasets with more complex model architectures.}
The accuracy performance of CIFAR10 and ImageNet classifiers, shown in  \cref{fig:combined_natural}, are higher than for the simpler FEMNIST and EMNIST classifiers. Note that for CIFAR10, due to statistical variances, the undeleted backdoor accuracy lowers slightly from $\poisonratio{}=70\%$ to $90\%$, leading to the higher Type-II error $\beta$ for $\poisonratio = 90\%$ in \cref{fig:def_verify_ratio}.
% We show the accuracy performance and verification performance for CIFAR10 and ImageNet classifiers in \cref{fig:combined_natural}. Compared to the simpler 
% On the FEMNIST dataset, we still have undeleted users' attack accuracy above $90\%$ when the poison ratio is greater than $30\%$, where our verification mechanism achieves high confidence.
%For example, for the FEMNIST classifier with $30\%$ poison ratio, we obtain $\beta = 7.0 \times 10^{-10}$ with 20 test samples by setting the tolerance of $\alpha$ as $10^{-5}$.
% On the CIFAR10 dataset, we still have $p$ value above $77\%$ for poison ratio between $30\%$ and $50\%$ with a decrease of clean test accuracy smaller than $5\%$.
%For the classifier with poison ratio of $30\%$, our verification mechanism achieves $\beta = 3 \times 10^{-8}$ with 20 test samples while setting $\alpha=10^{-5}$.

Third, \textbf{Neural Cleanse defense weakens with increasing fraction \userratio{} of users testing for deletion verification}, illustrated in \cref{fig:def_all_f}. %While for EMINST the simplicity of the MLP-model eases the backdoor detection, 
%Neural Cleanse fails in protecting the other architectures against a higher number of poisoning users. 
We can see that the performance of Neural Cleanse drops with a larger fraction of poisoning users (\userratio{}).
For ImageNet, Neural Cleanse reduces the average $p$, but at the same time the benign accuracy drops from $76\%$ to $72\%$, compared to the non-adaptive case. As $q\approx 0.001$, our approach still works satisfactory.

\section{Heterogeneity Across Individual Users}\label{sec:individual_evaluation}

\ifplotimages
\begin{figure*}[!ht]
	\centering
	\newsavebox\IBoxA \newlength\IHeightA
    \sbox\IBoxA{\includegraphics[width=0.19\linewidth]{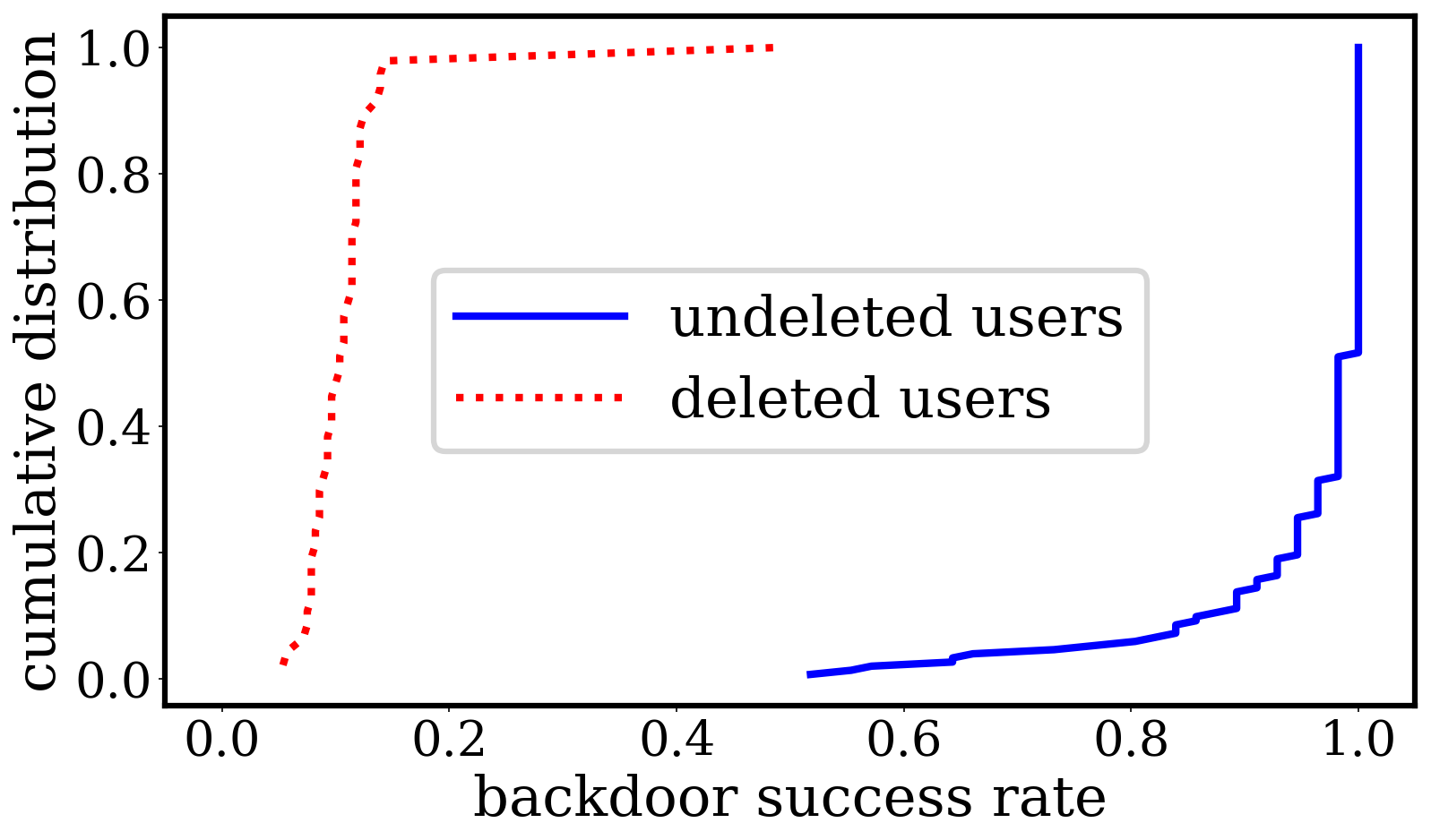}}\setlength\IHeightA{\ht\IBoxA}
	\begin{subfigure}[t]{0.19\linewidth}
		% \raggedleft
		% \vspace{-\IHeightA}
		{\small\hspace{0.37\linewidth}EMNIST}
		\includegraphics[width=\linewidth]{Images/experiment_results/author_poison/EMNIST_0.05_author_poison_0.5_poison_acc_cdf.png}
		%\caption{EMNIST}
		%\label{fig:nat_emnist_acc_cdf}
	\end{subfigure}\hfill
	\begin{subfigure}[t]{0.19\linewidth}
		%\raggedleft
		% \vspace{-\IHeightA}
		{\small\hspace{0.36\linewidth}FEMNIST}
		\includegraphics[width=\linewidth]{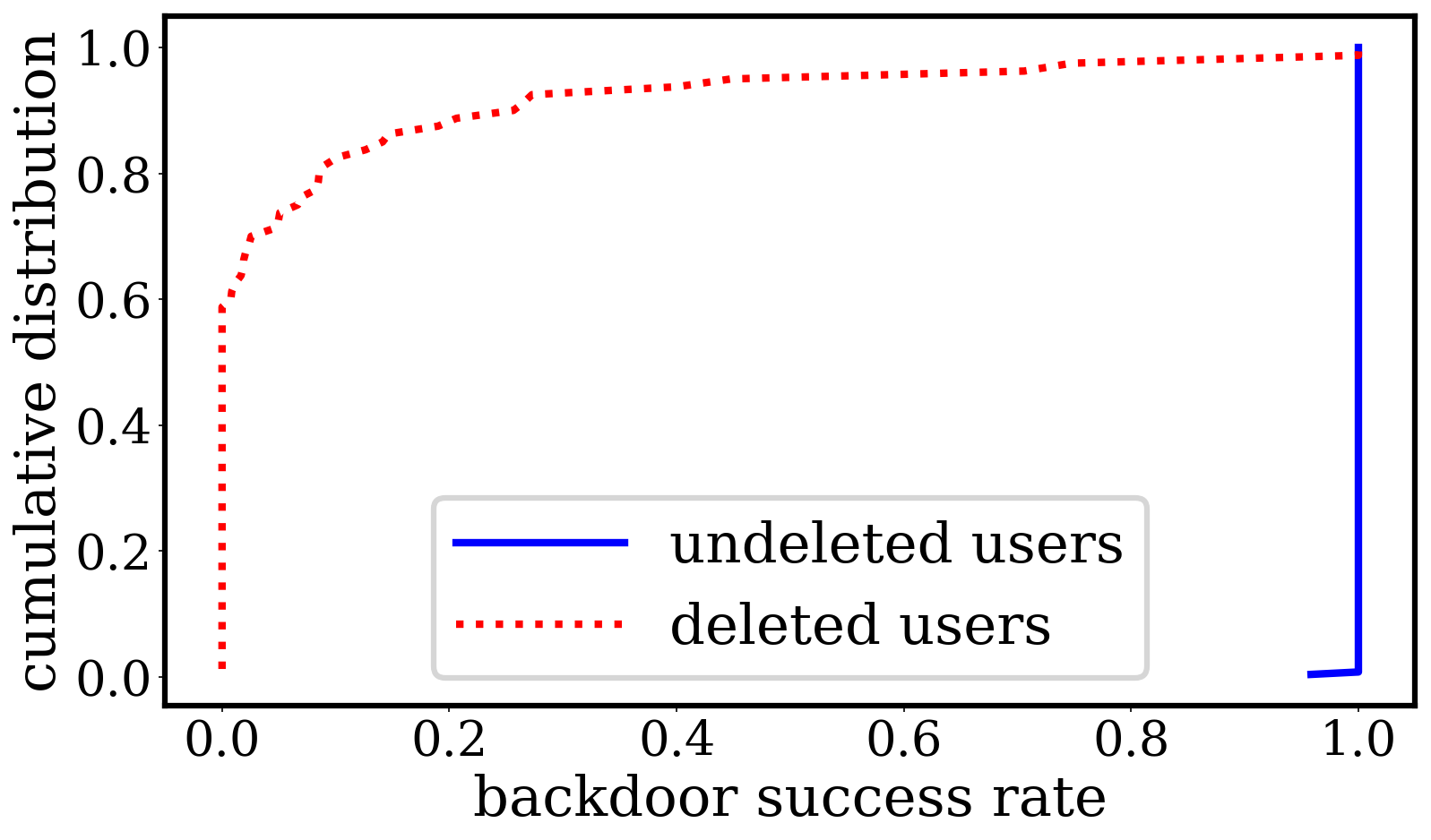}
		%\caption{FEMNIST}
		%\label{fig:nat_femnist_acc_cdf}
	\end{subfigure}\hfill
	\begin{subfigure}[t]{0.19\linewidth}
		%\raggedright
		% \vspace{-\IHeightA}
		{\small\hspace{0.36\linewidth}CIFAR10}
		\includegraphics[width=\linewidth]{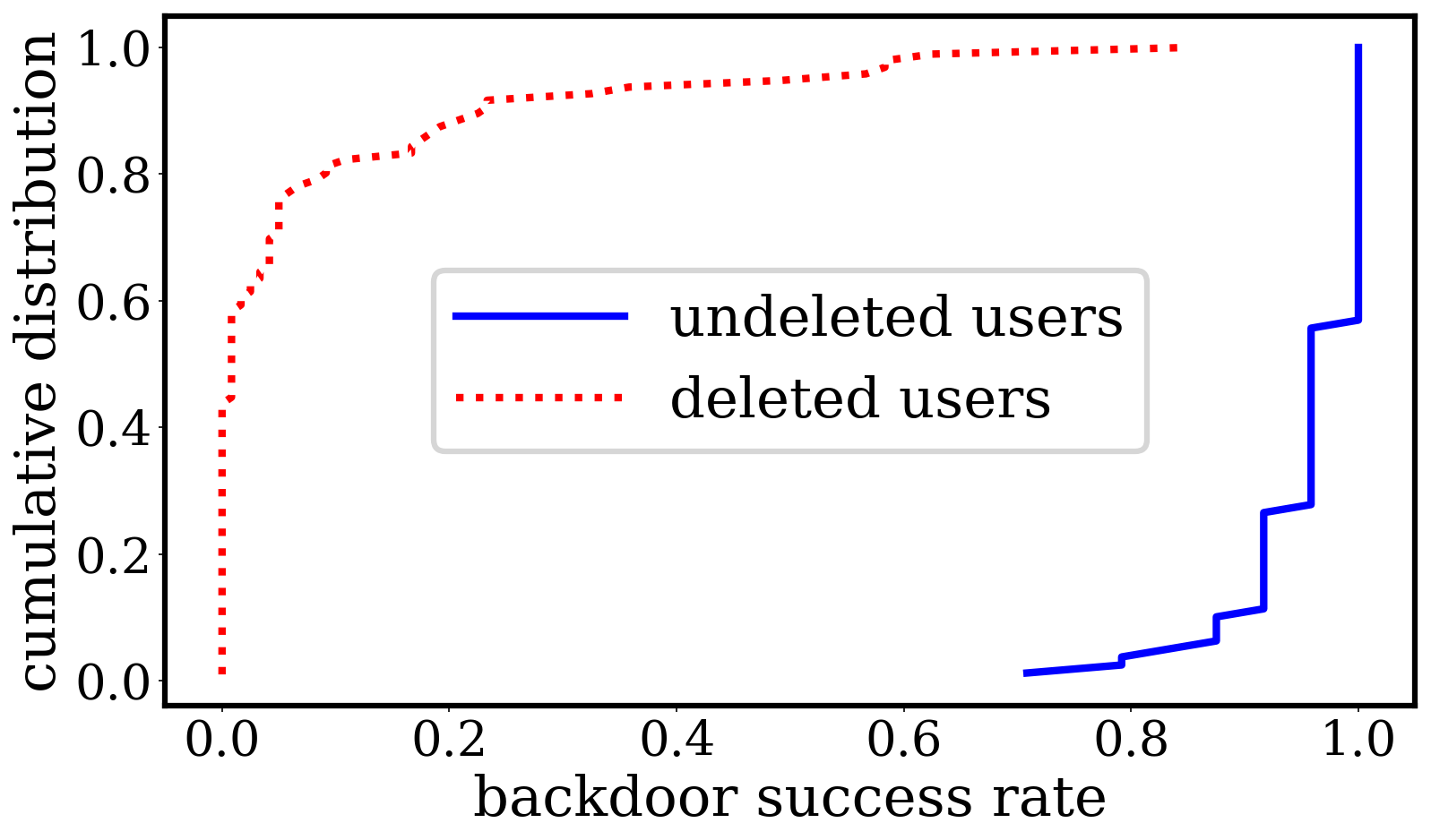}
		%\caption{CIFAR10}
		%\label{fig:nat_cifar_acc_cdf}
	\end{subfigure}\hfill
	\begin{subfigure}[t]{0.19\linewidth}
		%\raggedright
		% \vspace{-\IHeightA}
		{\small\hspace{0.36\linewidth}ImageNet}\vspace{-0.15em}
		\includegraphics[width=\linewidth]{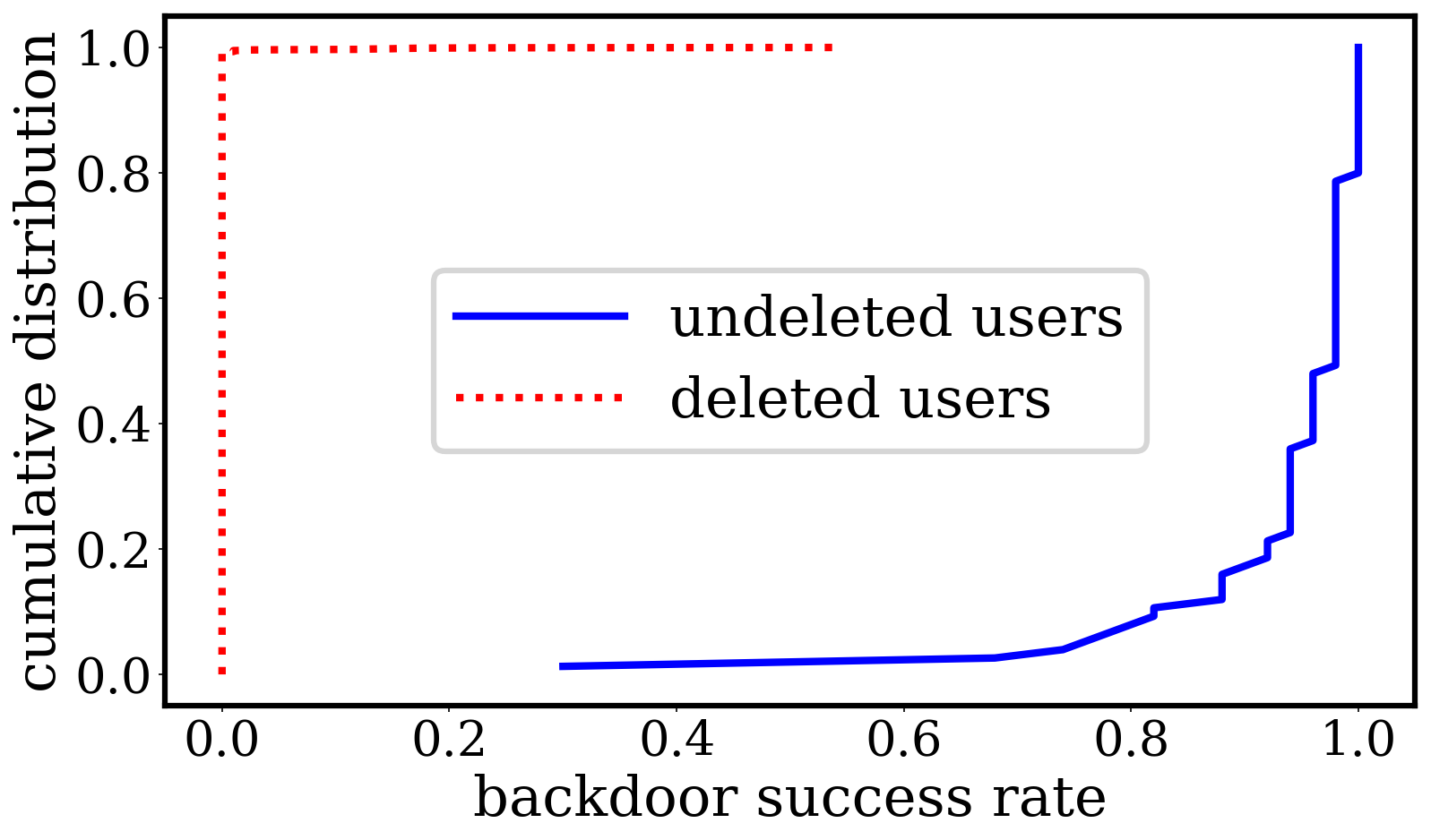}
		%\caption{ImageNet}
		%\label{fig:nat_imagenet_acc_cdf}
	\end{subfigure}\hfill
	\begin{subfigure}[t]{0.19\linewidth}
		%\raggedright
		% \vspace{-\IHeightA}
		{\small\hspace{0.37\linewidth}{AG News}\\}
		\includegraphics[width=\linewidth]{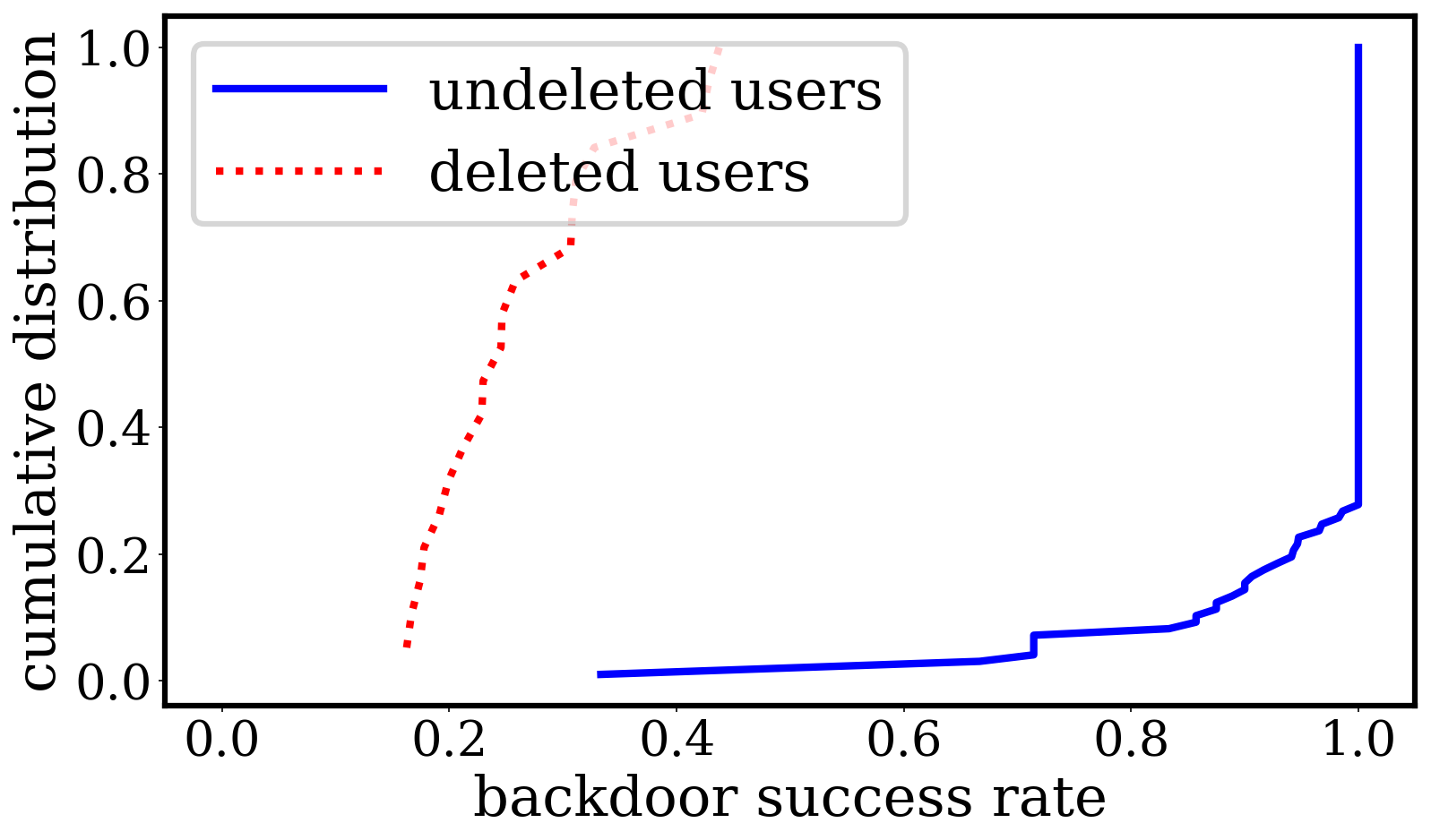}
		%\caption{AG News}
		%\label{fig:nat_agnews_acc_cdf}
	\end{subfigure}\hfill
	%\vspace{-0.5em}
	\caption{The CDFs of backdoor attack accuracy for deleted and undeleted users for different datasets ($\userratio=0.05$, $\poisonratio = 50\%$).}
	\label{fig:nat_model_acc_cdf}
	\vspace{-0.5em}
\end{figure*}
\fi

So far, our analysis on deletion verification performance is evaluated based on the ``average'' backdoor attack accuracy $p$ across all undeleted users and the average backdoor attack accuracy $q$ across all deleted users.
Next, we evaluate the heterogeneity in the 
performance of  
%A more practical scenario is to perform 
stochastic deletion verification across individual users, to account for the variance in individual users' backdoor attack accuracy values. We find that while most privacy enthusiasts are able to conclude correctly whether their data has been deleted, a small subset of deleted users also have high backdoor attack accuracy although the ML model never trained on their backdoor triggers and target labels.

To quantify this effect, we present the cumulative distribution plots for \Naive{} server over different datasets in \cref{fig:nat_model_acc_cdf}, where we fix the fraction of privacy enthusiasts \userratio{} as $0.05$ and data poison ratio \poisonratio{} as $50\%$. Similar results for the \Advanced{} server are shown in appendix \cref{fig:def_model_acc_cdf}.
%Except on the EMNIST classifier with Neural Cleanse, 
As shown in Figure~\ref{fig:nat_model_acc_cdf}, almost all undeleted users have high backdoor attack accuracy close to $100\%$.
However, \textbf{several deleted users indeed have high backdoor attack accuracy.} 
% For example, for the CIFAR10 classifier with \Naive{} server, 5 out of 100 deleted users have backdoor attack accuracy higher than $80\%$.
We think the reason is that there are one or more undeleted users with \textit{similar} backdoor triggers and the same target labels as those rare deleted users, resulting in their high backdoor attack accuracy without their data being used in the training set of the ML model.
In fact, popular image classification architectures, such as CNN and ResNet, are trained to behave similarly for images with rotation, translation, and transformation, leading to behave similarly on similar triggers.
%In this way, although there is no exactly identical backdoor trigger across users, the model indeed behaves similarly on similar triggers.

% \textbf{For the AG News dataset, we rarely observe deleted users with high backdoor attack accuracy:} No deleted user has a backdoor attack accuracy higher than $50\%$.
% %In this way, we can still distinguish undeleted cases ($H_{1}$) from deleted cases ($H_{0}$) with only one exception, where an undeleted user has around $50\%$ attack accuracy.
% Compared to image datasets, our backdoor triggers (random word list) on the AG News dataset are very different across different users.

% We further show the distributions of Type-II errors ($\beta$) on all classifiers in 
% %\cref{fig:nat_model_beta_cdf} and \cref{fig:def_model_beta_cdf}.
% \red{need to discuss here: whether present the distribution of $\beta$? If yes, which way to choose?}
% \ifplotimages
% %\input{Images/experiment_results/text_files/natural_model_beta_cdf}
% %\input{Images/experiment_results/text_files/defended_model_beta_cdf}
% \fi

\begin{table}

\centering
\resizebox{\linewidth}{!}{
\begin{tabular}{c | c | c | c | c | c }
\toprule
  & EMNIST & FEMNIST & CIFAR10 & ImageNet & AG News \\
\midrule
{1 user} & $2.1\times10^{-2}$ & $2.5\times10^{-2}$ & $3.8\times10^{-2}$  & $4\times10^{-4}$ & $8.1\times10^{-2}$ \\
\midrule
{2 users}     & $1\times10^{-4}$ & $3\times10^{-4}$ & $1\times10^{-3}$ & $4\times10^{-5}$ & $1.3\times10^{-2}$  \\
\midrule
{3 users}     & $< 10^{-5}$ & $<10^{-5}$ & $<10^{-5}$ & $<10^{-5}$ &  $4\times10^{-4}$\\
\bottomrule
\end{tabular}
}
\caption{The likelihood falsely assuming that a non-deleting server is following the deletions requests, i.e., that $\beta \geq 0.001$, given single/multiple users jointly run the hypothesis testing. With just 3 collaborating privacy enthusiasts, our techniques allows extremely low probabilities of false negatives.
%false accusation.
%As can be seen, the likelihood sharply reduces when two or more users collaborate for machine unlearning verification.
We use 30 samples and set $\userratio{}=0.05$, $\poisonratio{}=50\%$, $\alpha=10^{-3}$.
% our verification mechanism getting into worst-case scenarios, where some deleted users obtains high attack success rate. We measure it as the probability of $\beta \geq 0.01$ in our verification mechanism.
% We can find that such worse-case likelihood can be greatly reduces if two or more users collaborate for machine unlearning verification.
% We use 20 samples and set the poison ratio as $30\%$, $\alpha$ equal to $10^{-5}$.
}
\label{tab:multi-users}
\vspace{-1em}
\end{table}

On those deleted users with high backdoor attack accuracy, our verification mechanism is likely to wrongly blame the server for not deleting the data.
To resolve this issue, we propose that multiple users cooperate with each other by sharing their estimated backdoor success rates and thereby achieve high confidence verification.
%Specifically, we choose the lowest backdoor success rate among all cooperative users, given that the failure happens when several deleted users report high backdoor success rate.
%In this way, as long as one of cooperative deleted user reports low backdoor success rate, we still correctly infer the scenario as $H_{0}$.
%In contrast, for the opposite hypothesis $H_{1}$, all undeleted users typically report high backdoor success rates with high probability.
%
% Specifically, we decide based on the lowest success rate among all collaborating users, given that the scenario $H_1$ is incorrectly inferred when multiple deleted user report high backdoor success rate. However, this occurs with low probability. 
% In contrast, for the correct hypothesis $H_{1}$, all undeleted users typically report high backdoor success rates with high probability.
% We illustrate the effectiveness of this cooperative strategy among users in Table \ref{tab:multi-users}, where different columns show the \emph{worst-case probability} of $\beta \geq 0.01$ with varying numbers of cooperative users. Note that by increasing the number of cooperative users  by 1, the worst-case probability falls by one order of magnitude.
\ds{
Specifically, we decide whether or not to reject the null-hypothesis $H_{0}$ (the server does delete) based on accumulated p and q values. Therefore, $c$ collaborating users compute the mean of their $p$ and $q$, decide for an $\alpha$ (Type-I-error) and a upper bound for $\beta$ (Type-II-error). If the estimated Type-II error is smaller than the bound, the null hypothesis is rejected. Note that if $c$ users share their results, and each user tested $n$ backdoored samples, then the accumulated number of tested samples is $c\cdot n$. In Table~\ref{tab:multi-users}, we show the probability that a server does not fulfil deletion requests, but the null-hypothesis is falsely accepted (false negative), for $\alpha$ as $10^{-3}$ and an upper bound for $\beta$ as $10^{-3}$. For computational efficiency, we applied Monte-Carlo-sampling.
With multi-user cooperation, the probability of false negatives can be greatly reduced.
}

%To combat this issue, multiple users can cooperate to infer whether the server deletes their data or not.
%The individual users could submit their concerns to a central agency that acts when multiple users raise concern about failed deletion of their data.

To further make our verification mechanism more reliable, we can use multiple backdoor triggers with multiple target labels for each user and estimate the lowest backdoor success rate among all triggers. As long as the deleted user has at least one trigger leading to low attack accuracy, we can obtain reliable performance from a worst-case perspective. We discuss the number of possible backdoors based on coding theory in the  \cref{appendix:numberofusers}. Another direction is to combine our method with other verification methods, such as user-level membership inference attacks to detect whether a user's data was used to train the ML model or not \cite{Mem_song_kdd19}.
We leave this as future work.

\section{Discussion}\label{sec:discussion}
Here, we elaborate on additional aspects such as usefulness of our mechanism for regulation compliance, effects of future backdoor defenses, and limitations of our system.\vspace{-1em}

\subsection{Server-Side Usefulness of our Mechanism}
%\sw{\subsection{Server Benefits for Regulation Compliance}}
Besides leveraging the backdoor attacks for deletion verification at the user side, our approach also provides benefits to an honest server.
%an honest server can also have benefits.
First, the server can use our method to \textbf{validate that their data deletion pipeline} is bug-free.
In cases where the MLaaS providers do not want backdoors in their ML models, such backdoor-based verification mechanism can be applied in production by setting the target backdoor labels to a specific ``outlier'' label, which is not used for future prediction.

Second, the server can use our backdoor-based mechanism to \textbf{quantitatively measure the effectiveness of recently proposed deletion approaches without strict deletion guarantees}, such as \cite{deletion_guo_arxiv19, MUL_baumhauer_arxiv20}.
These approaches directly update the model parameters to remove the impact of the deleted data without retraining the model.
Our framework can be an efficient way to evaluate their performance.

\subsection{Other Backdoor Attacks and Defenses}\label{sec:discussion:otherbackdoormethods}
The contribution of our paper is to use backdoor attacks for probabilistic verification of machine unlearning.
% Specifically, we follow the approach of backdoor attacks in Gu et al. \cite{badnets_gu_arxiv17} by injecting the user-specific trigger into a fraction of training samples and setting their labels as the target label.
Our verification mechanism can be easily extended to other backdoor attack methods \cite{trojan_Liu_ndss18, backdoor_saha_aaai20, backdoor_turner_arxiv19}.
When testing the verification performance under a strategic malicious server, we use the state-of-the-art backdoor defense method, Neural Cleanse \cite{backdoor_defense_wang_sp19}, to train the ML model.
We find that Neural Cleanse only has a limited impact of our verification approach: undeleted users still have much higher backdoor success rate than the deleted users, and our verification mechanism still works well.
Several new defense approaches have also been proposed recently \cite{liu2019abs, veldanda2020nnoculation}. {Veldanda et al. \cite{veldanda2020nnoculation}  showed that the defense method proposed by Liu et al. \cite{liu2019abs} is ineffective against adaptive backdoor attacks.}
%The defense method proposed by Liu et al. \cite{liu2019abs} is shown to be ineffective against adaptive backdoor attacks \cite{veldanda2020nnoculation}, while the defense method proposed by Veldanda et al. \cite{veldanda2020nnoculation} came out last month}.
However, unless they fully mitigate backdoor issues in the multi-user setting, our verification method is still useful.
%Comparing the performance of different backdoor attacks and defenses is not our paper's focus.

If the malicious adversary finds a perfect defense method to fully mitigate the backdoor attacks, then our verification approach will not work, but a user can become aware of this scenario by observing a low backdoor attack accuracy before the deletion request.
In this scenario, the user %is not limited in its ability to perform verification of machine unlearning by 
could either find stealthier backdoor attacks \cite{backdoor_turner_arxiv19, backdoor_yao_ccs19}, or use alternative verification methods, such as membership inference attacks~\cite{Mem_song_kdd19}.
%If that is the case, the user can either use more stealthy backdoor attacks \cite{backdoor_turner_arxiv19, backdoor_yao_ccs19} or use other deletion verification methods, such as user-level membership inference attacks \cite{Mem_song_kdd19}.
%Even if the malicious adversary find , in which case our verification does not work, the user can notice the ineffectiveness by obtaining a low backdoor attack accuracy before the deletion request.

\textbf{Differential Privacy}, a rigorous privacy metric, is widely used to limit and blur the impact of individual training samples by clipping and noising the gradients during training~\cite{abadi2016deep,neel2020descent}. However, it does not mitigate the issue for the following reasons: First, it is commonly applied to hide the influence of single training samples. In our scenario, users backdoor more than one of their training samples, increasing the noise required for hiding to a per-user level, and thereby lowering the highly noise-sensitive utility (general accuracy) of the model considerably. The reduction to a per-user level requires an enormous data base and is therefore currently applied only by the largest tech-companies, e.g., Google, Amazon, and Apple. And finally, colluding users can bypass the (single user) impact limiting guarantees given by differential privacy.

\subsection{Limitations of our Approach}
\label{subsec:limitation}
%While we show that our approach works over a range of datasets and network architectures, there are certain limitations of this approach when considering the system design broadly. We describe these below: 

\textbf{Constraints on Data Samples.} We begin by noting that our approach does not work in systems where the privacy enthusiasts (small fraction of users) do not have the ability to modify their data before sending, or if the data is too simple to allow a backdoor without diminishing the benign accuracy. Furthermore, even if privacy enthusiasts are allowed to modify their data, they need at least few tens of samples for our approach to work well in practice. This limitation can be addressed in a few different ways (1) privacy enthusiasts can aggregate their hypothesis testing to provide an overall high-confidence verification despite each individual verification not yielding high performance. This is shown in Table~\ref{tab:multi-users} (2) privacy enthusiasts can simply generate more samples for the purposes of the verification (3) finally, given the ability of the privacy enthusiasts to know beforehand if their approach works or not, they might switch to other methods, such as membership inference attacks~\cite{Mem_song_kdd19}. 
%Finally, low dimensionality of the data does not impact the performance of our mechanism as for backdoor attacks, we do not have the constraint of preserving input semantics as with benign data. 

\textbf{Conflicting Backdoor Patterns.} When backdoors conflict with each other, which can happen when backdoors are similar, our approach might fail for some users. However, our method crucially allows the detection of this by having close or overlapping measured values of $p$ and $q$. This could be caused by two factors (1) time-related unlearning effects, i.e. other data samples overwrite a backdoor pattern in continuous learning, in which case we recommend to verify the deletion close to the corresponding request (cf. \cref{sec:analysis:bounds}) (2) too many users in the system, in which case we can cap the number of privacy enthusiasts or increase the space of permissible backdoors as discussed in Appendix~\ref{appendix:numberofusers}.

\textbf{Other (Future) Backdoor Defense Methods.} While we showcase the effectiveness of our approach with state-of-the-art Neural Cleanse method, we acknowledge that new defenses could potentially be detrimental to our approach. However, reliable backdoor defense methods are a widely known open problem in machine learning and we consider our formulation of machine unlearning a rigorous mathematical foundation to study the tussle between such attacks and defenses. For non-continuous cases, like the AG News dataset, there are no known satisfactory defense methods currently available. Finally, in regards to our approach, privacy enthusiasts can detect the presence of such defenses and adapt their strategy accordingly (cf. \cref{sec:discussion:otherbackdoormethods}).

\section{Related Work}\label{sec:relatedwork}

\noindent \textbf{Existing Machine Unlearning Approaches.}
%\TODO{Sameer}\red{Garg's paper "Formalizing Data Deletion in the Context of the Right to be Forgotten" is missing} %\grey{I think we can leave it out, technically its unpublished and we don't have to talk about it. If we do, we should write a dedicated paragraph at least in how we differ from them. If someone can add a short blurb then we can add this.} 
The simple approach %to remove traces from specific training data in an algorithm, namely 
of just deleting data as requested by users and retraining the model from scratch is inefficient in case of large datasets and models with high complexities. 
%that take long time to train. 
Therefore, Cao and Yang \cite{MUL_cao_SP15} applied ideas from statistical query learning \cite{kearns1998efficient} and train conventional models based on intermediate summations of training data. Upon deletion request, they update the summations which is significantly more efficient. 
%However, their method is only applicable to conventional machine learning models. %, such as naive Bayes classifiers, support vector machines, and k-means clustering.
Ginart et al. \cite{deletion_ginart_NIPS19} proposed two provably efficient mechanisms for k-means clustering by either quantizing centroids at each iteration or using a divide-and-conquer algorithm to recursively partition the training set into subsets. Concurrently, Garg et al.~\cite{garg2020} explore the related but orthogonal problem of what it means to delete user data from a theoretical viewpoint. 

Bourtoule et al. \cite{MUL_bourtoule_arxiv19} proposed to split the training set into disjoint shards, train local models separately on every shard, and aggregate outputs from all local models to obtain the final prediction. 
By splitting data into disjoint shards, only one or few local models need to be retrained for deletion requests.
% This method increases the deletion efficiency at the cost of a large storage overhead and accuracy degradation.

% Other methods aim to update model parameters to remove the impact of the deleted data on the model.
Other methods aim to purge the model parameters from the impact of individual samples. 
\ds{Guo et al. \cite{deletion_guo_arxiv19} defined data deletion as an indistinguishability problem:
%: similar to differential privacy \cite{DP_paper_wiki}, for a deletion request, 
the updated model should be difficult to distinguish from the model retrained from scratch without the deleted samples. They leverage the influence function of the deleted training point to apply a one-step Newton update on model parameters. Neel et al.\cite{neel2020descent} achieved such indistinguishability by noising the model parameters after updating.}
Baumhauer et al. \cite{MUL_baumhauer_arxiv20} focused on the setting where the removal of an entire class is requested by designing a linear transformation layer appended to the model.
%  Both approaches do not involve any retraining process, thus need little time to run.
However, there is no guarantee that no information of the deleted class is left inside the updated model.
% Guo et al. \cite{deletion_guo_arxiv19} adopt $(\epsilon, \delta)$ metrics as in differential privacy to provide a bound on the left information of the deleted data.
% Baumhauer et al. \cite{MUL_baumhauer_arxiv20} perform model inversion attacks \cite{fredrikson_inversion_CCS15} against the updated model to empirically show the effectiveness of their proposed method.

\noindent \textbf{Verifying Machine Unlearning.}
One approach is the use of \emph{verifiable computation}. Such techniques can enable data-owners to attest the MLaaS provider's processing steps and thus verify that the data is truly being deleted. Possible techniques include the use of secure processors~\cite{gruss2017strong}, trusted platform modules~\cite{ohrimenko2016oblivious, allen2019algorithmic}, and zero-knowledge proofs~\cite{ghosh2015fully, ghosh2015authenticated}. However, such techniques require assumptions that limit their practicality -- server side computation, the use of trusted parties, computational overhead along with frequent attacks on the security of such systems~\cite{meltdown, spectre}. Moreover, as these schemes require detailed insight into the computation process, frequently, the service provider cannot keep the model a secret, which is a serious limitation.% of such approaches.

Shokri et al. \cite{shokri_membership_SP17} investigated \emph{membership inference attacks} in machine learning, where the adversary goal is to guess whether a sample is in the target model's training set. 
They train shadow models on auxiliary data to mimic the target model and then train a classifier for inference attacks. 
Song et al. \cite{Mem_song_kdd19} extended the record-level membership inference to \emph{user-level membership inference attacks} that determines whether a user's data was used to train the target model.
\ds{Chen et al.~\cite{chen2020machine} apply such methods to infer the privacy loss resulting from performing unlearning.}
To apply these methods to our setup, each user needs to train shadow models on an auxiliary datasets similar to the target model, including knowledge of the target model's architecture and computation capability.  
In comparison, our backdoor-based machine unlearning verification approach does not require those strong assumptions and obtains extreme good verification performance.

Recently, Sablayrolles et al. \cite{data_tracing_arxiv20} proposed a method to detect whether a particular image dataset has been used to train a model by adding well-designed perturbations that alters their extracted features, i.e., watermarking the model. Instead of tracing an entire dataset, our approach considers a multi-user setting where each user adds a personal backdoor.
Also, they only consider image datasets.
Finally, Adi et al. used backdoor attacks to watermark deep learning models~ \cite{backdoor_watermark_adi_usenix18}.

\section{Conclusion}
The right to be forgotten addresses an increasingly pressing concern in the digital age. While there are several regulations and interpretations of the legal status of this right, there are few concrete approaches to verify data deletion. In this paper, we formally examine probabilistic verification of machine unlearning and provide concrete quantitative measures to study this from an individual user perspective. Based on backdoor attacks, we propose a mechanism by which users can verify, with high confidence, if the service provider is compliant of their right to be forgotten. We provide an extensive evaluation over a range of network architectures and datasets. Overall, this work provides a mathematical foundation for a quantitative verification of machine unlearning.

% With the proliferation of sensitive data in today's digital world, 

% and has become a necessity. Classical techniques to data deletion such as burning books at Fahrenheit 451 

% With the advent of digital age, Sensitive data -- classical techniques such as burning books at Fahrenheit 451

% As we have thoroughly shown, we have successfully carried over ancient methods of data deletion -- namely just burning dissident books at Fahrenheit 451  --  in the new digital age. In times where information can be copied with almost no effort, we provide a method to verify the deletion of information not at its origin, but at the point where it matters: its usage. And in the spirit of the futuristic society we urge to create, our method does not involve fire or any other kind of physical violence. Instead, we require only electricity and the purest form of power: the dependence on omnipresent inter-connectivity. Hereby, we apply for the best paper award. 

%\ifFINAL
\iffalse
A real-world example is the face-recognition tool \textit{Clearview.AI} which was widely used by American law-enforcement agencies. Clearview AI scraped large publicly available user profiles for pictures, e.g. Facebook and Twitter, and trained their machine learning models to recognize millions of Internet users. Recently, Twitter sent a cease-and-desist letter to the company, insisting that they remove all images as it is against Twitter's policies. Our approach would allow privacy enthusiasts to verify whether Clearview AI complied with data deletion requests.

Clarify machine unlearning $\neq$ data deletion
\fi

\ifFINAL
\section*{Acknowledgements}
This work was supported in part by the National Science Foundation under grants CNS-1553437 and CNS-1704105, the ARL's Army Artificial Intelligence Innovation Institute (A2I2), the Office of Naval Research Young Investigator Award, the Army Research Office Young Investigator Prize, Faculty research award from Facebook, and by Schmidt DataX award.
\fi

% \newpage	

% \bibliographystyle{abbrv}
% \bibliographystyle{plain}
\bibliographystyle{IEEEtran}

\bibliography{main}

% \begin{itemize}
%     % \item Move backdoor-related work in citation form to background.
%     \item Work out section user perspective
%     % \item polish evaluation % David added missing paragraphs
%     \item Conclusion
%     \item making Figure 10 (model complexity) smaller 
%     \item figure 12/14 to appendix
%     % \item getting rid of table 1
%     \item polish / complete table II - III 
%     \item mention that we cannot control "deletion", only unlearning (mention in threat model)
%     \item More iterations on abstract and introduction
%     \item polish section background
%     \item LIWEI: Can you make the plots with the quantile range the same size as the beta plots next to it?
%     \item the text in Figure 1 (approach outline needs to be larger)
%     \item section 8 requires polishing. 
%     \item fixing argumentation against shmatikov's approach in section 4 according to Prateek's suggestions. 
% \end{itemize}

\appendix
%\appendices
\section{Proof of Theorem~\ref{thm:computerho}}\label{app:proof}

We break down the procedure to compute the metric $\ourmetric{A, \alpha}{s,n}$ for a given Type I error $\alpha$ into the following steps: 

\begin{enumerate}[noitemsep]
    \item Compute the optimal value of the threshold $t$ for a given value of $\alpha$, the Type I error
    \item Compute the value of $\beta$, the Type II error, for the given optimal threshold $t$
    \item Compute $\ourmetric{A, \alpha}{s,n}$ from the previously computed $\beta$
\end{enumerate}

The proof relies on the independence of prediction order. We define a test of the backdoor success of $n$ consecutive samples as follows: 
\begin{definition}
Given oracle access to the predictions on $n$ samples $\{\mathsf{sample_i}\}_{i=1}^n$, for $r\in[0,1]$, we define $\testacc{n}{r}$ as a random variable that returns a value in $\{0,1\}^n$ where each entry is 1 with probability $r$ and 0 with probability $1-r$ assuming the order of the predictions is immaterial and that they are processed independently.
\end{definition}

If $r$ is set to the backdoor success probability $p$ or $q$, then the above defined $\testacc{n}{r}$ mimics the output of the corresponding ML-mechanism as it effectively measures the ratio of cases where a backdoor was able to change the prediction of its sample to a target label. Hence, for the hypothesis test, it is sufficient to compare the backdoor success ratio $p$ where the backdoor works (data not deleted) to the case where it does not work (data deleted) with ratio $q$. Next, we prove that the random variable $\hat{r}$ follows a rescaled binomial distribution. 

\begin{lemma}[Measured backdoor success rate]\label{lemma:measured}
Let $n \in \mathbb{N}$. Let $o \in \{0,1\}^n$ be a random draw from $\testacc{n}{r}$ with $r\in[0,1]$, the following statements hold: 
% \TODO{the scientific way to write \Pr [ \hat{r} \geq x ] is the summation over the ceil of  n*x to n, similar for equation (8a)}
\begin{enumerate}
    \item The random variable $\hat{r} = \frac{1}{n} \sum_{j=1}^n o_j$ follows a binomial distribution with abscissa scaled to $[0,1]$ with draws $n$ and success probability $r$ where $o_j$ is the $j^{\tiny{\mathrm{th}}}$ draw output of $o$. 
    %In other words, $\Pr[\hat{r} = \frac{i}{n}]$ for $i\in \{0,\ldots,n\}$ follows a binomial distribution with abscissa scaled to $[0,1]$ with draws $n$ and success probability $r$. 
    We call $\hat{r}$ the discrete success rate probability.
    \item The standard-deviation of $\hat{r}$ shrinks as $O(\frac{1}{\sqrt{n}}$)
    \item The tail probability mass of $\hat{r}$ can be computed for $x\in [0,1]$ using the following relation (and a symmetric relation for $\hat{r} \leq x$):
        \begin{equation}\label{eq:tailbounds}
          \Pr [ \hat{r} \geq x ] = \sum_{k \geq  n\cdot x }^n \binom{n}{k} r^k(1-r)^{n-k}
        \end{equation}
\end{enumerate}
\end{lemma}

\begin{proof}
As we assumed the independence of prediction order,
the output of $\testacc{n}{r}$ follows a binomial distribution $\binomdistr{n}{r}$ where $n$ is the number of draws and $r$ is the success probability.
\begin{enumerate}
    \item For $k\in\{0,\ldots,n\}$, let $\mathcal C_k$ be the set of all possible outputs $c$ of $\testacc{n}{r}$ with $\sum_{j=0}^n c_j = k$. Note that all outputs are equally likely. Then, the occurrence probability for $\hat{r} = \frac{k}{n}$ is given by:
        \begin{equation}\label{eq:binomdistr}
        \begin{aligned}
            \Pr[\hat{r} = \frac{k}{n}] 
            &= \sum_{\forall c \in \mathcal C_k} r^k(1-r)^{n-k}\\
            &= \binom{n}{k} r^k(1-r)^{n-k} \\
            &= \Pr_\binomdistrlabel[ k = k | n,r] 
        \end{aligned}
        \end{equation}

    \item The variance of a binomial distribution is $\sigma_n^2 = nr(1-r)$. With scaled abscissa by $1/n$, the standard deviation becomes $\sigma = \sqrt{\sigma^2_n / {n^2}} = {\sqrt{r(1-r)}}/{\sqrt{n}}$.
%$\sigma = \sqrt{\sigma^2} = \sqrt{\frac{\sigma^2_n}{n^2}} = \frac{\sqrt{r(r-1)}}{\sqrt{n}}$.    

    \item The mass in the tail is the sum over the probabilities for the corresponding discrete events. Hence \cref{eq:tailbounds} directly follows from summing \cref{eq:binomdistr} for $k\geq n\cdot x$
\end{enumerate}
\end{proof}

\begin{proof}[Proof of \Cref{thm:computerho}]
\Cref{lemma:measured} can be directly applied to prove the results in \Cref{thm:computerho}. In particular, the hypothesis test consists of distinguishing two scaled binomial distributions with $r=q$ in case $H_0$ (the data has been deleted) and $r=p$ for $H_1$ (data has not been deleted). \Cref{fig:numberofsamples} graphically illustrates this.
%while \cref{fig:multiplesample} demonstrates the effect of multiple samples on the scaled distributions. 
As seen in \Cref{lemma:measured}, the scaled distributions concentrate around the mean, thus reducing the probability in the overlapped areas, which in effect reduces the Type I and Type II error probabilities.

By \Cref{lemma:measured}, the shape of the hypothesis distributions depends on $q$ for $H_0$ and $p$ for $H_1$. Therefore, for a given a threshold $t \in [0,1]$, the Type I error $\alpha_t$ and the Type II error $\beta_t$ for the hypothesis test depend on $p$ and $q$ respectively.    
\begin{subequations}\label{eq:alphaandbetat}
  \begin{align}
    \begin{split}
            \alpha_q^t &= \Pr[\hat{r} > t | H_0, n ] \\
                        &= \sum_{k >  n\cdot t }^n \binom{n}{k} q^k(1-q)^{n-k}
    \end{split}
    \label{eq:alphat} 
    \\
    \begin{split}
            \beta_p^t  &= \Pr[\hat{r} \leq t | H_1, n] \\
                        &= \sum_{k = 0 }^{n\cdot t} \binom{n}{k} p^k(1-p)^{n-k}
    \end{split}
    \label{eq:betat} 
  \end{align}
\end{subequations}

Given that $\alpha$ is set by systemic constraints, we invert \cref{eq:alphat} to get the optimal value of the threshold $t$ and then plug that into \cref{eq:betat}. Consider the following equality defining $t_\alpha$ given $\alpha$:
\begin{equation}
     \heavyside{\frac{k}{n} \leq t_\alpha} := \heavyside{\Pr\left[\hat{r} \leq \frac{k}{n} \Big|H_0, n\right] \leq 1 - \alpha }
\end{equation}

We can then use this implicit definition of threshold $t_\alpha$ to determine the Type II error $\beta$ given a value of $p$: 
\begin{align*}
    \beta_{p,q}^\alpha\! 
    &= \!\sum_{k=0}^n \Pr\left[\hat{r} = \frac{k}{n} \Big| H_1, n\right]\! \cdot \heavyside{\frac{k}{n} \leq t_\alpha} \\
    &=  \!\sum_{k=0}^n \Pr\left[\hat{r} = \frac{k}{n} \Big| H_1, n\right]\! \cdot \heavyside{\Pr\left[\hat{r} \leq \frac{k}{n} \Big|H_0, n\right] \leq 1\! -\! \alpha} \\
     &= \!\sum_{k=0}^n \binom{n}{k} p^k(1\!-\!p)^{n\!-\!k}\! \cdot \heavyside{\sum_{l = 0}^k\!\binom{n}{l} q^l(1\!-\!q)^{n\!-\!l} \leq 1\! -\! \alpha }
\end{align*}
Finally, to connect this value with $\rho_{A, \alpha}(s, n)$, we use \cref{eq:test} from \cref{sec:formulation} and $s=(p,q)$:
\begin{equation}
\begin{aligned}
    \rho_{A, \alpha}(s, n) &= 1- \beta_{p,q}^\alpha \\ 
                           &= 1 - \sum_{k=0}^n \binom{n}{k} p^k(1-p)^{n-k} \cdot \\
    &~~~~~~~~~~ \heavyside{\sum_{l = 0}^k \binom{n}{l} q^l(1-q)^{n-l} \leq 1- \alpha }
\end{aligned}
\end{equation}
which gives us an analytic expression for the confidence that we are in case $H_1$, i.e. that our data has not been deleted as requested. If this value is high, the user has high confidence that the server does not follow deletion request. 
%$H_0$, i.e. that our data has been deleted. If this value is high, the user has high confidence that our data has been deleted. 
%\TODO{the previous version says the confidence is for $H_0$, which I think is reversed. I just change to $H_1$, please check whether I am wrong or not.}
\end{proof}

\section{Evaluation Datasets and Architectures}\label{appendix:datasets}
Following paragraphs describe the datasets we use for evaluation, the ML architectures, and backdoor methods, summarized in \Cref{tab:dataset_summary}.\\

\noindent \textbf{Extended MNIST (EMNIST).}
The dataset is composed of handwritten character digits derived from the NIST Special Database 19 \cite{grother1995nist}.
The input images are in black-and-white with a size of $28 \times 28$ pixels.
In our experiments, we use the digits form of EMNIST, which has 10 class labels, each with 280,000 samples \cite{emnist}.
We split the dataset into 1,000 users in an independent and identically distributed (IID) manner, with 280 samples per user.
For the model architecture, we use a multi-layer perceptron (MLP), which contains three layers with 512, 512, and 10 neurons.
Using the Adam optimizer \cite{kingma2014adam}, we train the model with 20 epochs and a batch size of 128.
On a clean dataset, the model achieves $99.84\%$ training accuracy and $98.99\%$ test accuracy.
For the backdoor method, each user chooses a random target label and a backdoor trigger by randomly selecting 4 pixels and setting their values as 1.\\
%The user then applies this backdoor pattern to certain fraction of training samples and sets their labels to the chosen target label.

\noindent \textbf{Federated Extended MNIST (FEMNIST).}
The dataset augments Extended MNIST by providing a writer ID \cite{Caldas2018LEAFAB}.
We also use the digits, containing 10 class labels and 382,705 samples from 3,383 users, rendering it IID due to the unique writing style of each person.
Same as EMNIST, the input image is in black-and-white with $28 \times 28$ pixels.
Different from EMNIST, this dataset does not include additional preprocessing, such as size-normalization and centering.
Also, the pixel value is inverse: the value of 1.0 corresponds to the background, and 0.0 corresponds to the digits. Therefore, we use the same backdoor method as for EMNIST, but setting the pixels to 0 instead of 1.
For the model architecture, we use a convolutional neural network (CNN), containing two convolutional layers with $3 \times 3$ kernel size and filter numbers of 32 and 64, followed by two fully connected layers, containing 512 and 10 neurons.
We use the Adam optimizer \cite{kingma2014adam} to train the model with 20 epochs and batch size of 128.
Without backdooring, the model achieves $99.72\%$ training accuracy and $99.45\%$ test accuracy.\\

\noindent \textbf{CIFAR10.}
Providing $32 \times 32 \times 3$ color images  in 10 classes, 
%(airplane, automobile, bird, cat, deer, dog, frog, horse, ship, truck)
with 6,000 samples per class, we split this dataset in  into 500 users in an IID manner, 120 samples per user. Applying a residual network (ResNet) \cite{he_ResNet_CVPR16}, containing 3 groups of residual layers with number of filters set to (16, 32, 64), and 3 residual units for each group, and using Adam~\cite{kingma2014adam} for training with with 200 epochs and batch size of 32, this model achieves $98.98\%$ training accuracy and $91.03\%$ test accuracy on a clean dataset without backdoors.
We use standard data augmentation methods (e.g., shift, flip) for good accuracy performance.
%to feed the training samples.
The backdoor method is identical to EMNIST.
Note that we consider RGB channels as different pixels.\\
%The data augmentation, in fact, may corrupt backdoor patterns embedded in some poison samples.
%However, as shown in experiment results, we can still obtain high backdoor attack accuracy for the undeleted users.

\noindent \textbf{ImageNet.}
This widely used dataset contains $1331168$ differently sized images used for object recognition tasks, assigned to 1000 distinct labels~\cite{deng2009imagenet}. We removed 23 pictures due to incompatible jpeg-color schemes, and trained a ResNet50 model~\cite{He_2016_CVPR}. Due to computation power restrictions, we applied transfer-learning with a pre-trained model provided by the torch framework~\cite{torch}. Without backdooring, the pretrained model achieves $87.43\%$ training accuracy and $76.13\%$ test accuracy.
The generation of the backdoor is identical to CIFAR10, except that due to the varying sizes of ImageNet pictures, we colored 4 random color-channels of a transparent 32x32x3 pixel mask-image and then scale the mask-image up to the corresponding ImageNet picture size when applying the backdoor.\\

\noindent \textbf{AG News.}
This is a major benchmark dataset for text classification \cite{zhang2015character}.
The raw dataset~\cite{AGnews} contains $1,281,104$ news articles from more than $2,000$ news sources (users).
Similar to Zhang et al. \cite{zhang2015character}, we choose the 4 largest news categories (Business, Sci/Tech, Sports, World) as class labels and use the title and description fields of the news to predict its category.
We filter out samples with less than 15 words and only keep users with more than 30 samples to improve statistical evaluation reliability.
The final dataset has $549,714$ samples from $580$ users.
For the model architecture, we use a long short-term memory (LSTM) network \cite{hochreiter1997lstm}.
The model first turn words into 100-dimension vectors, and then uses a LSTM layer with 100 units to learn feature representations, and finally a linear layer with 4 neurons for classification.
We use the Adam optimizer \cite{kingma2014adam} to train the model with 5 epochs and batch size of 64.
Without backdooring, the model achieves $96.87\%$ training accuracy and $91.03\%$ test accuracy.
For the backdoor method, each user chooses a random target label and a backdoor pattern by randomly picking 4 word positions in last 15 words and replacing them with 4 user-specific words, which are randomly chosen from the whole word vocabulary.

% the illustration of decreasing alpha and betta with increasing number of testing samples.
\begin{figure}
\centering
  \centering
  \includegraphics[width=\linewidth]{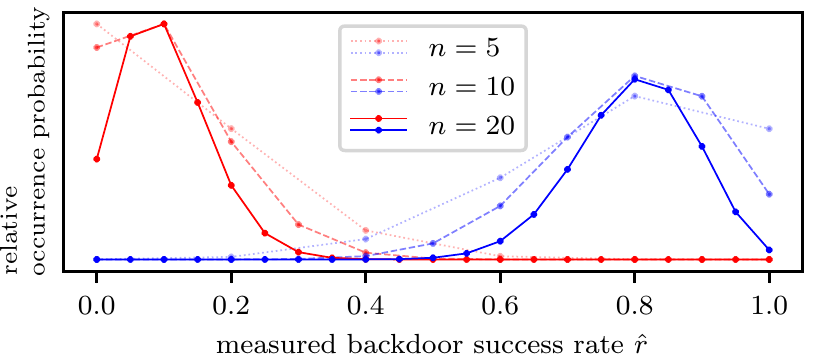}
  \caption{This figure intuitively shows how the confidence in distinguishing between $H_0$ and $H_1$ improves with the number of samples. With additional number of measured samples $n$, the distributions concentrate around the mean, thus simultaneously decreasing both $\alpha$ and $\beta$. Height of curves are adjusted for demonstration purposes. $q=0.1, p=0.8$}
  \label{fig:multiplesample}
\end{figure}

% we reference these alpha plots in the main body
\setlength{\spacehack}{-0.4em}
\setlength{\spacehacko}{-0.3em}

\newlength{\alphaswidth}

\begin{figure*}[!ht]
    \setlength{\alphaswidth}{0.35\linewidth}
	\centering
    %%%%%%%%%%%%%%%%%%%%%%%%%%%%%%%%%%%%EMNIST
    \verticaltext{~~~~~~~~~~~~~~~~EMNIST}
	\begin{subfigure}[t]{\alphaswidth}
		\raggedleft
		\includegraphics[width=\linewidth]{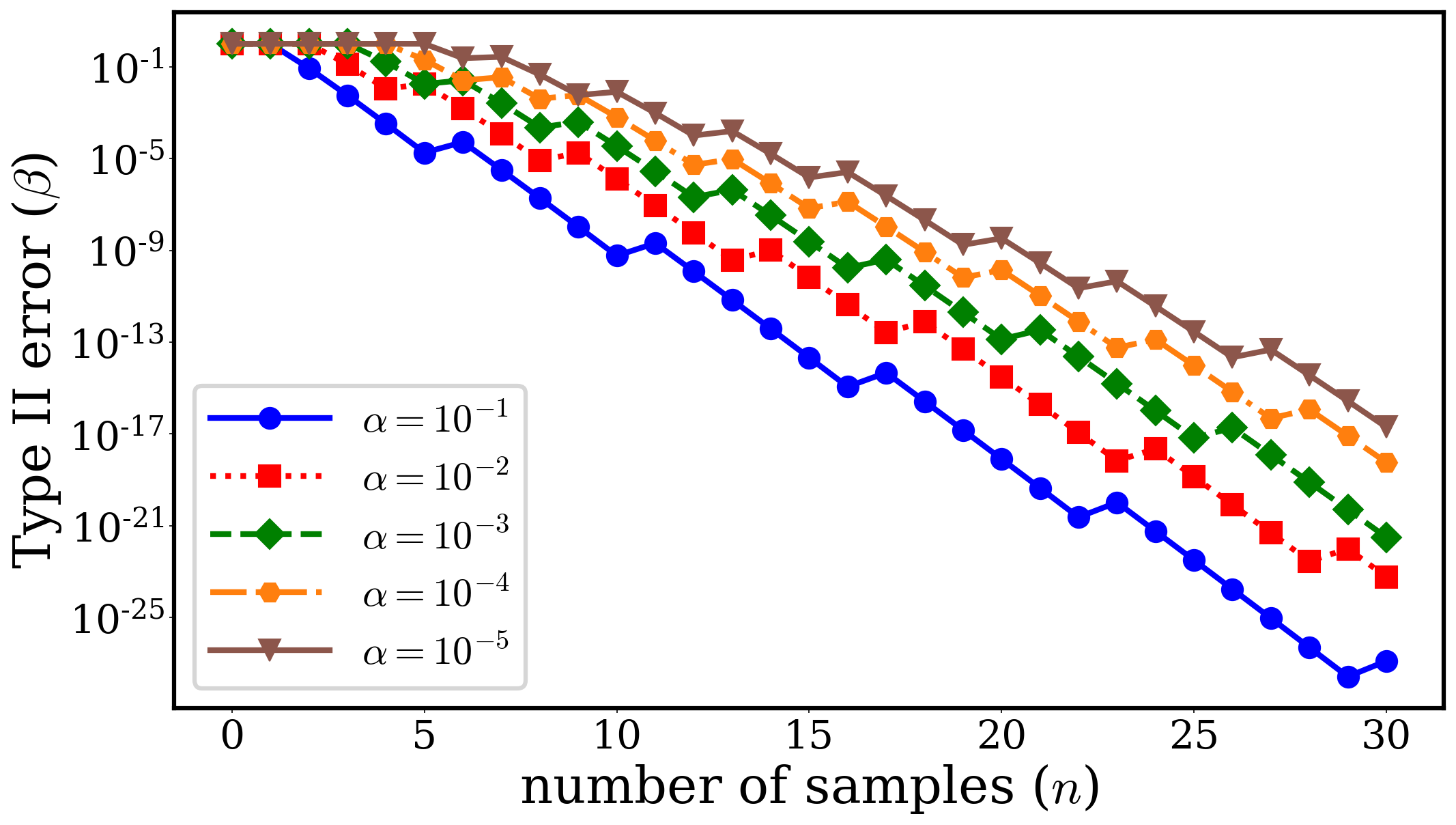}\vspace{\spacehacko}
	\end{subfigure}\hfill
	\begin{subfigure}[t]{\alphaswidth}
		\raggedright
		\includegraphics[width=\linewidth]{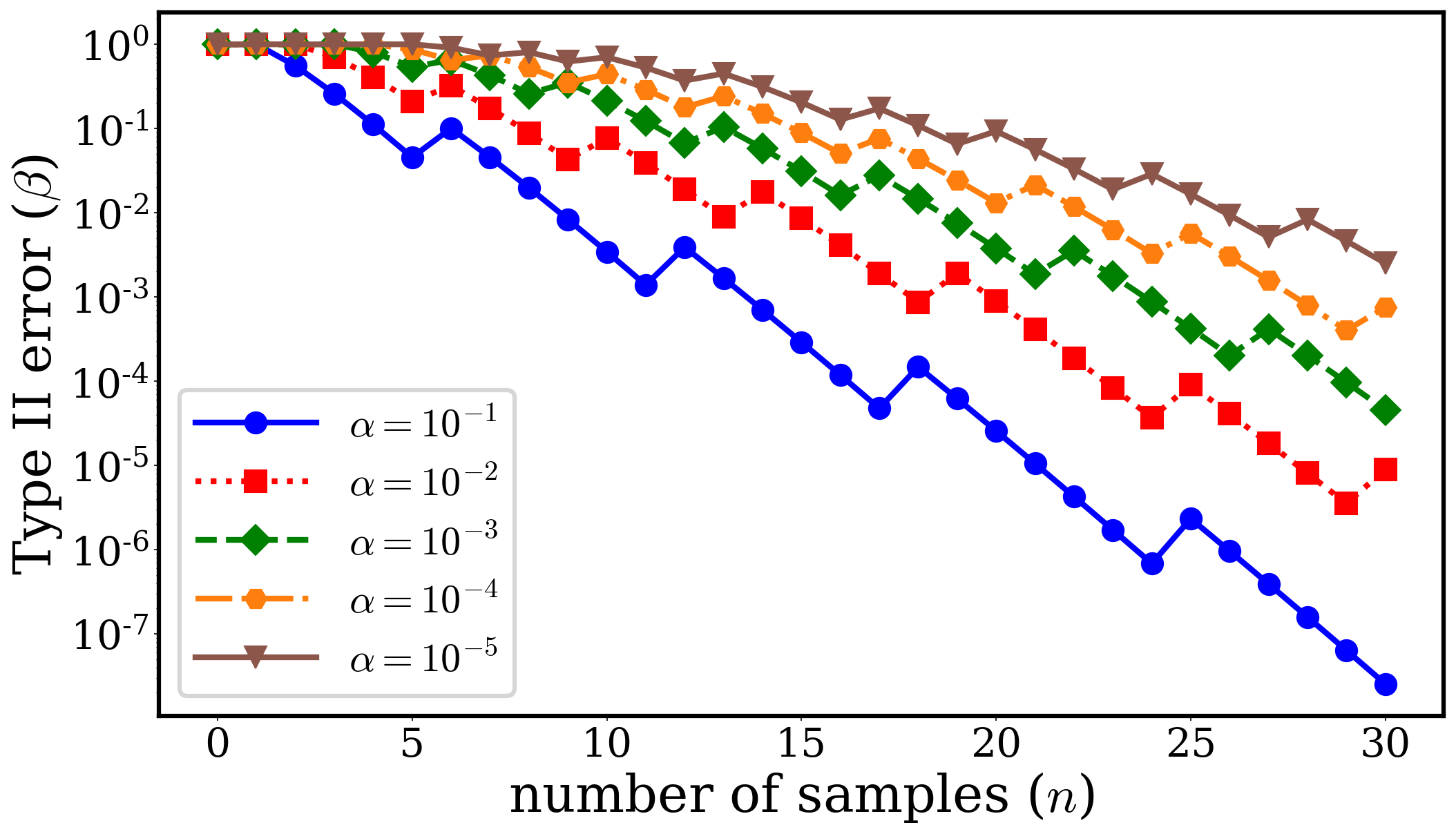}\vspace{\spacehacko}
	\end{subfigure}\hfill
	\par\medskip\vspace{\spacehack}
    %%%%%%%%%%%%%%%%%%%%%%%%%%%%%%%%%%%%FEMNIST
    \verticaltext{~~~~~~~~~~~~~FEMNIST}
	\begin{subfigure}[t]{\alphaswidth}
		\raggedleft
		\includegraphics[width=\linewidth]{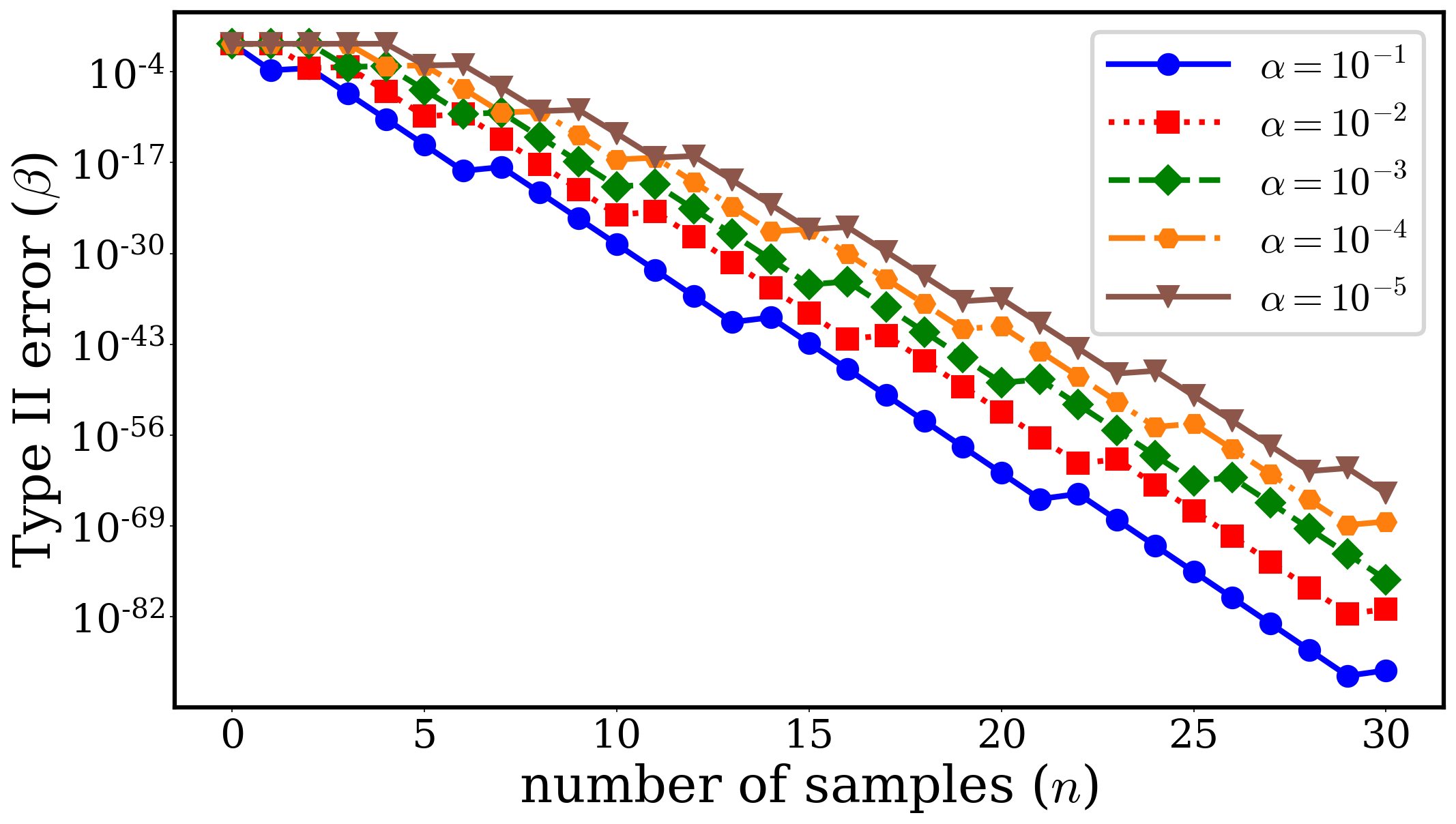}\vspace{\spacehacko}
	\end{subfigure}\hfill
	\begin{subfigure}[t]{\alphaswidth}
		\raggedright
		\includegraphics[width=\linewidth]{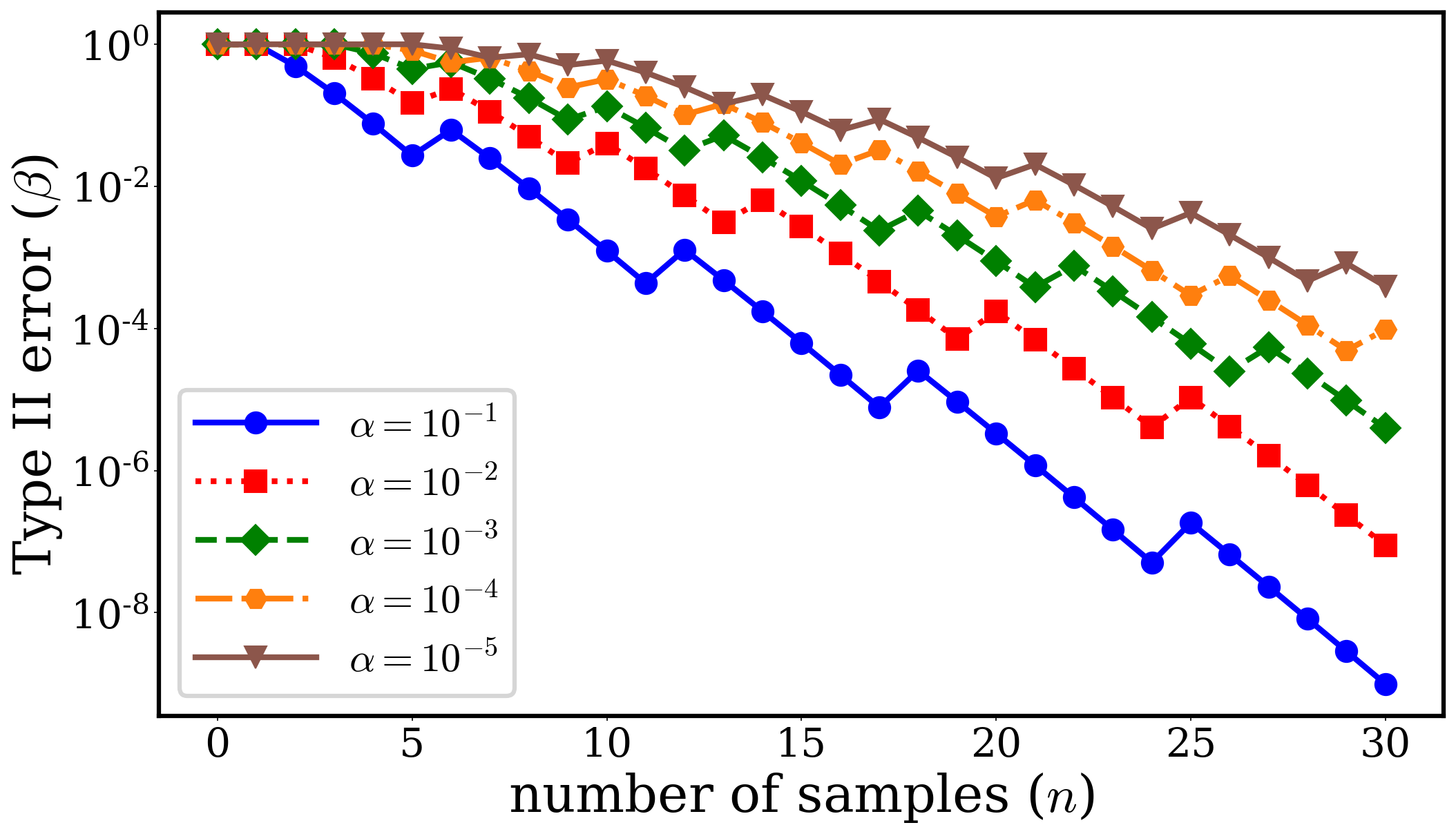}\vspace{\spacehacko}
	\end{subfigure}\hfill
	\par\medskip\vspace{\spacehack}
    %%%%%%%%%%%%%%%%%%%%%%%%%%%%%%%%%%%%CIFAR
    \verticaltext{~~~~~~~~~~~~~CIFAR10}
	\begin{subfigure}[t]{\alphaswidth}
		\raggedleft
		\includegraphics[width=\linewidth]{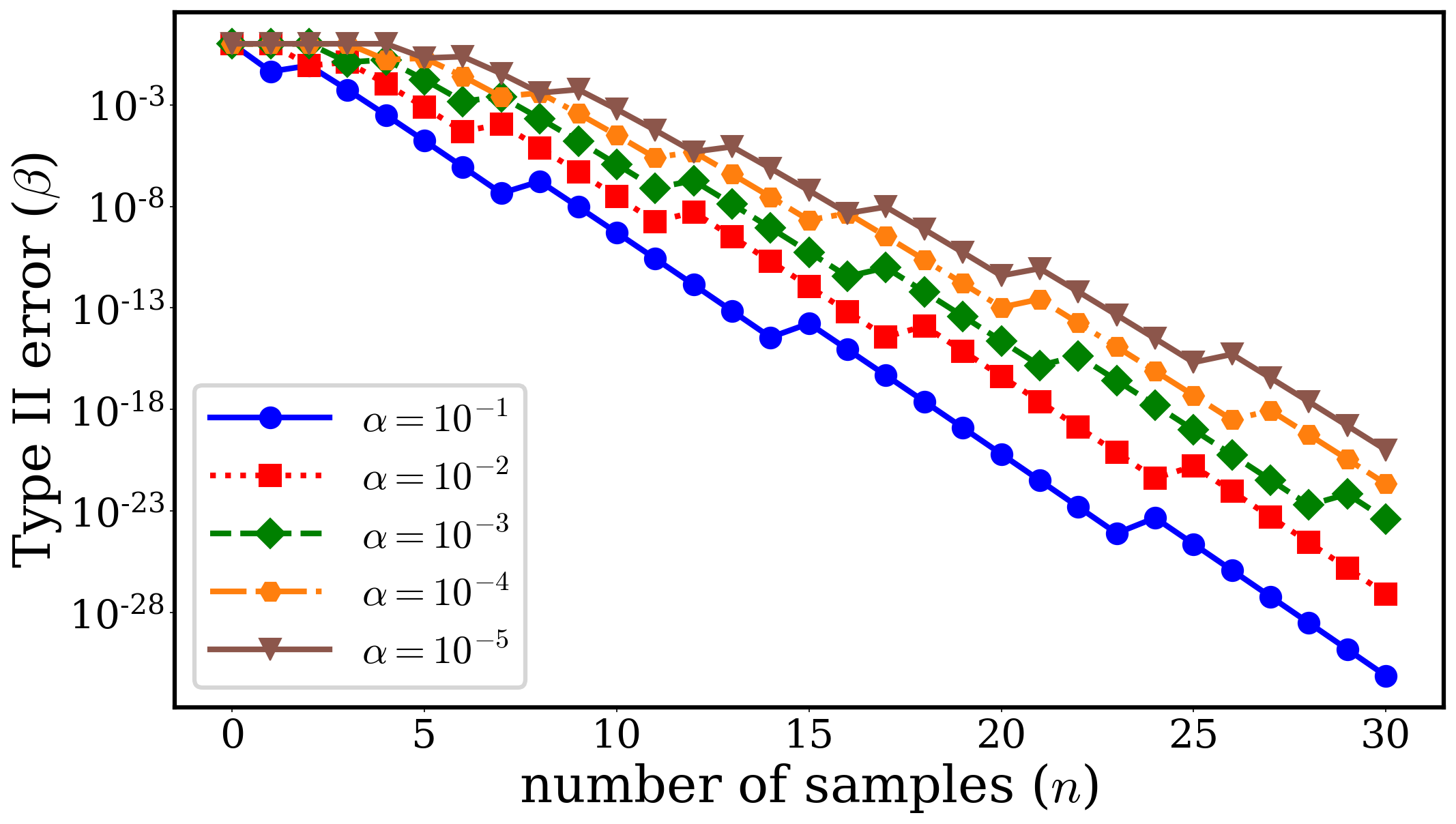}\vspace{\spacehacko}
	\end{subfigure}\hfill
	\begin{subfigure}[t]{\alphaswidth}
		\raggedright
		\includegraphics[width=\linewidth]{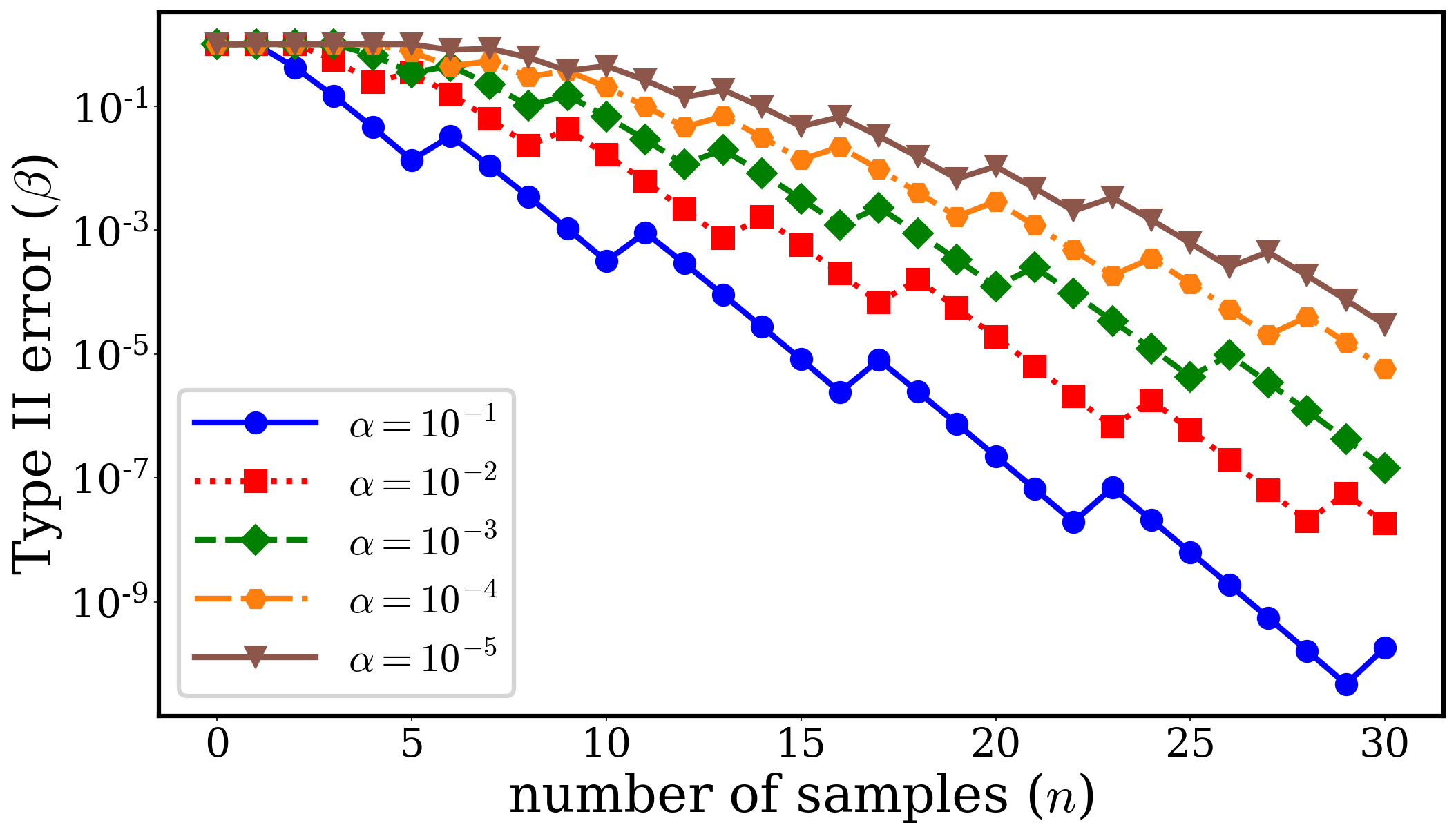}\vspace{\spacehacko}
	\end{subfigure}\hfill
	\par\medskip\vspace{\spacehack}
	%%%%%%%%%%%%%%%%%%%%%%%%%%%%%%%%%%%% ImageNet
    \verticaltext{~~~~~~~~~~~~~ImageNet}
	\begin{subfigure}[t]{\alphaswidth}
		\raggedleft
		\includegraphics[width=\linewidth]{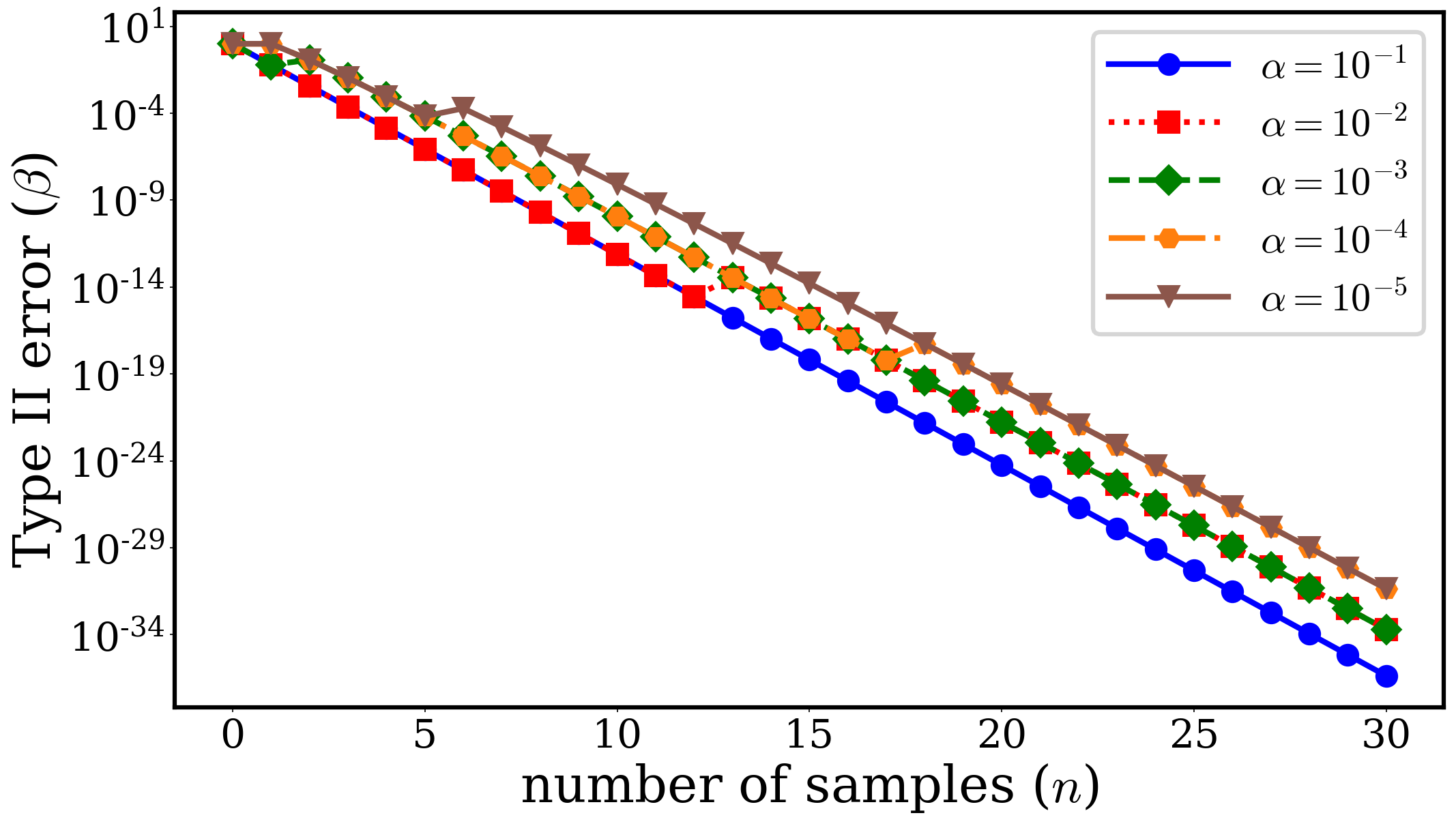}\vspace{\spacehacko}
	\end{subfigure}\hfill
	\begin{subfigure}[t]{\alphaswidth}
		\raggedright
		\includegraphics[width=\linewidth]{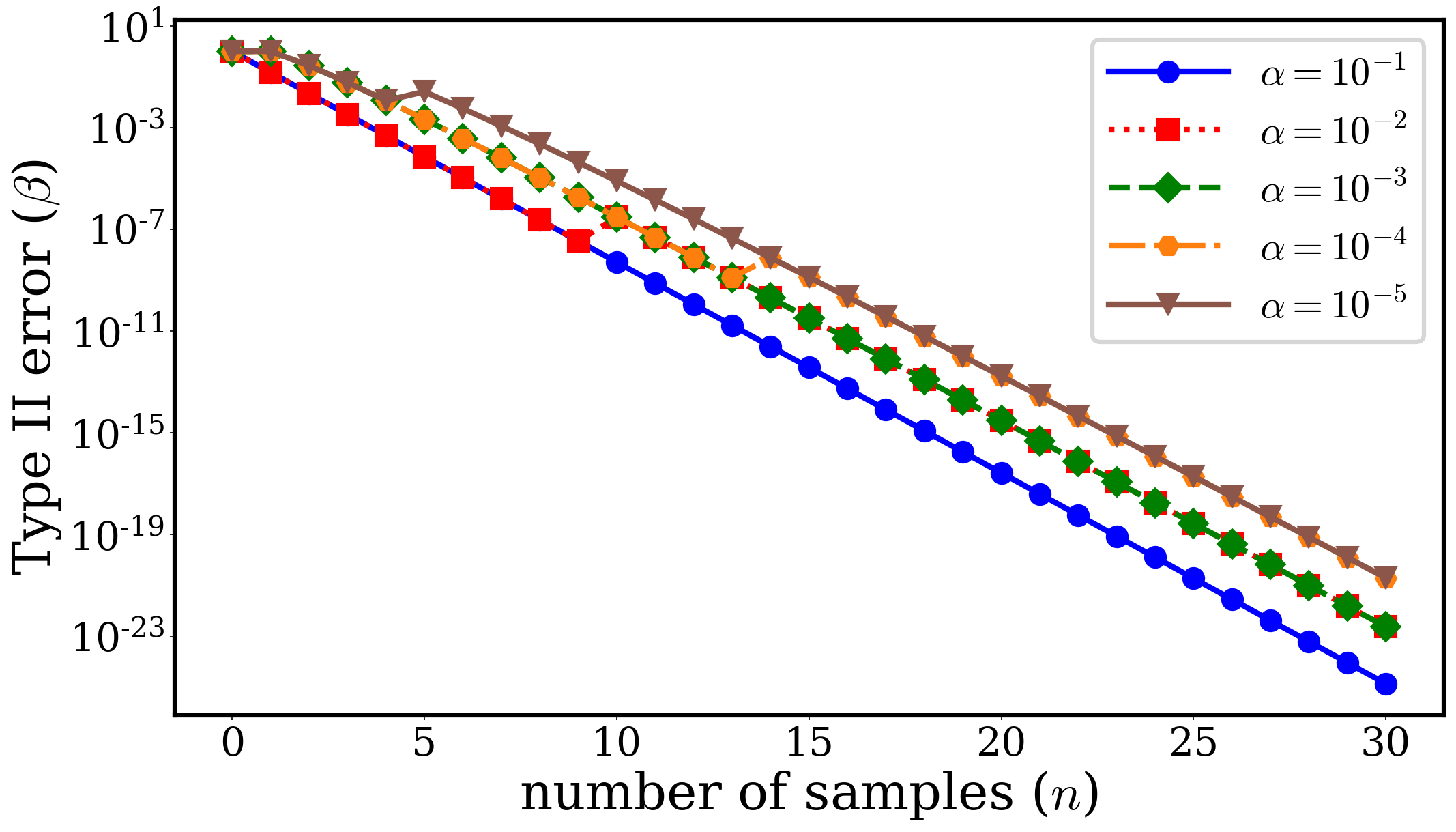}\vspace{\spacehacko}
	\end{subfigure}\hfill
	\par\medskip\vspace{\spacehack}
    %%%%%%%%%%%%%%%%%%%%%%%%%%%%%%%%%%%% AG News
    \verticaltext{~~~~~~~~~~~~~AG News}
    \newsavebox\IBoxB \newlength\IHeight
    \sbox\IBoxB{\includegraphics[width=\alphaswidth]{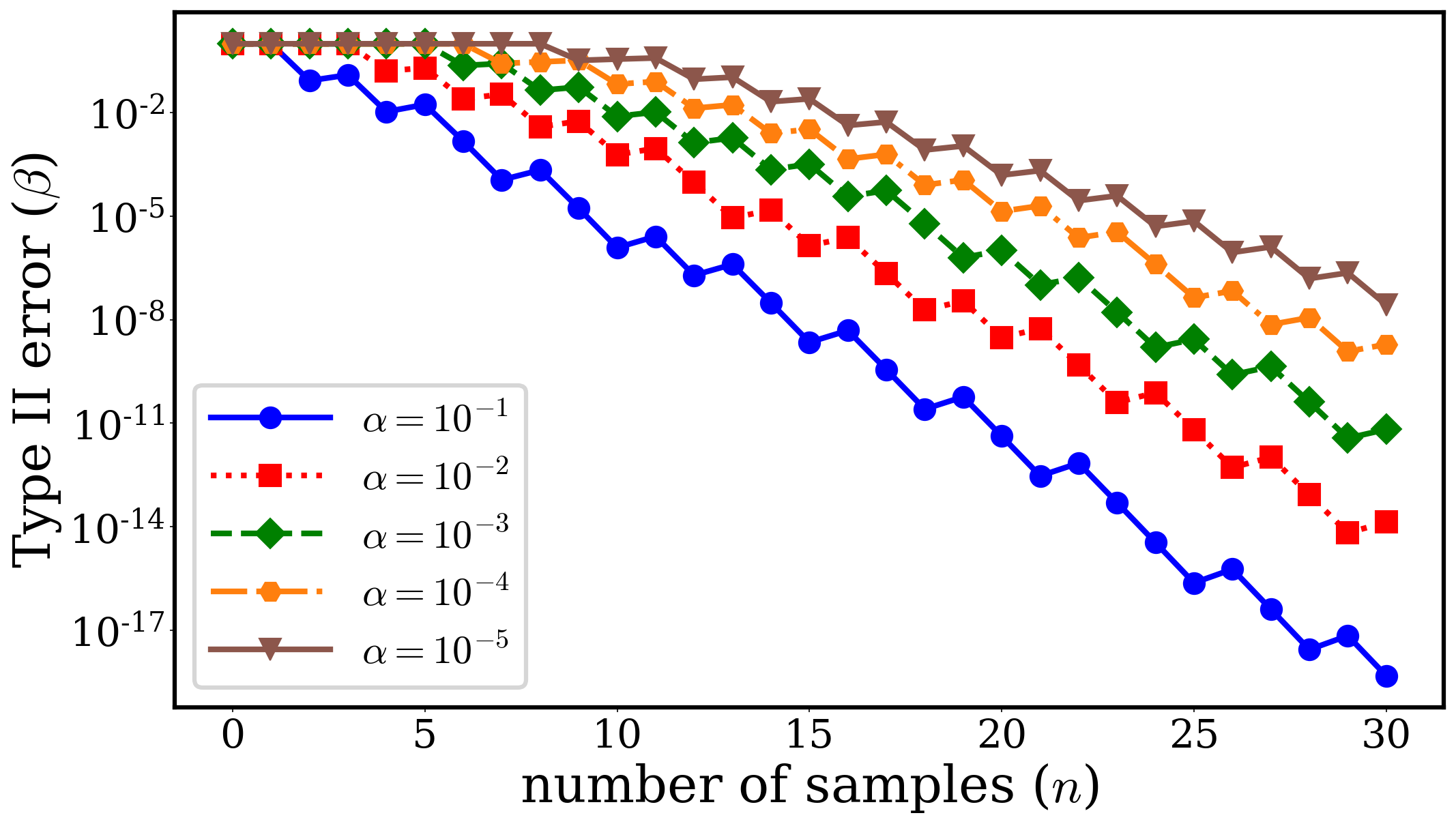}}\setlength\IHeight{\ht\IBoxB}%
	\begin{subfigure}[t]{\alphaswidth}
		\raggedleft
		\includegraphics[width=\linewidth]{Images/experiment_results/author_poison/AGNEWS_lstm_5_percent_author_poison_verify_poison_ratio_50.png}
		\caption{\textbf{\Naive{} Server}}
		\label{fig:nat_verify_alpha}
	\end{subfigure}\hfill
	\begin{subfigure}[t]{\alphaswidth}
    \vspace{-\IHeight}
	    \centering
	    \begin{minipage}[c][\IHeight][c]{0.8\linewidth}
    		    \textit{Not available}. State-of-the-art backdoor defense methods, such as the applied Neural Cleanse, do not generalise to non-continuous datasets such as AG-News. See \cref{sec:evaluation:advancedserver} for more details.
	    \end{minipage}
		\caption{\textbf{\Advanced{} Server}}
		\label{fig:def_verify_alpha}
	\end{subfigure}\hfill
	\caption{Verification performance for with different Type-I error $\alpha$. We have set the author poison fraction $\userratio{}=0.05$ and the poison ratio $\poisonratio{}=50\%$. Each row of plots is evaluated on the data-set specified at the most-left position. The left column depicts the result for the \Naive Server and the right column for the \Advanced Server.}
	\label{fig:verify_alpha}\vspace{\spacehack}
\end{figure*}

% The heterogeneity in users for adaptive server. we reference on that in the main body
\begin{figure*}[!b]
	\centering
	\begin{subfigure}[t]{0.24\linewidth}
		\centering
		\includegraphics[width=\linewidth]{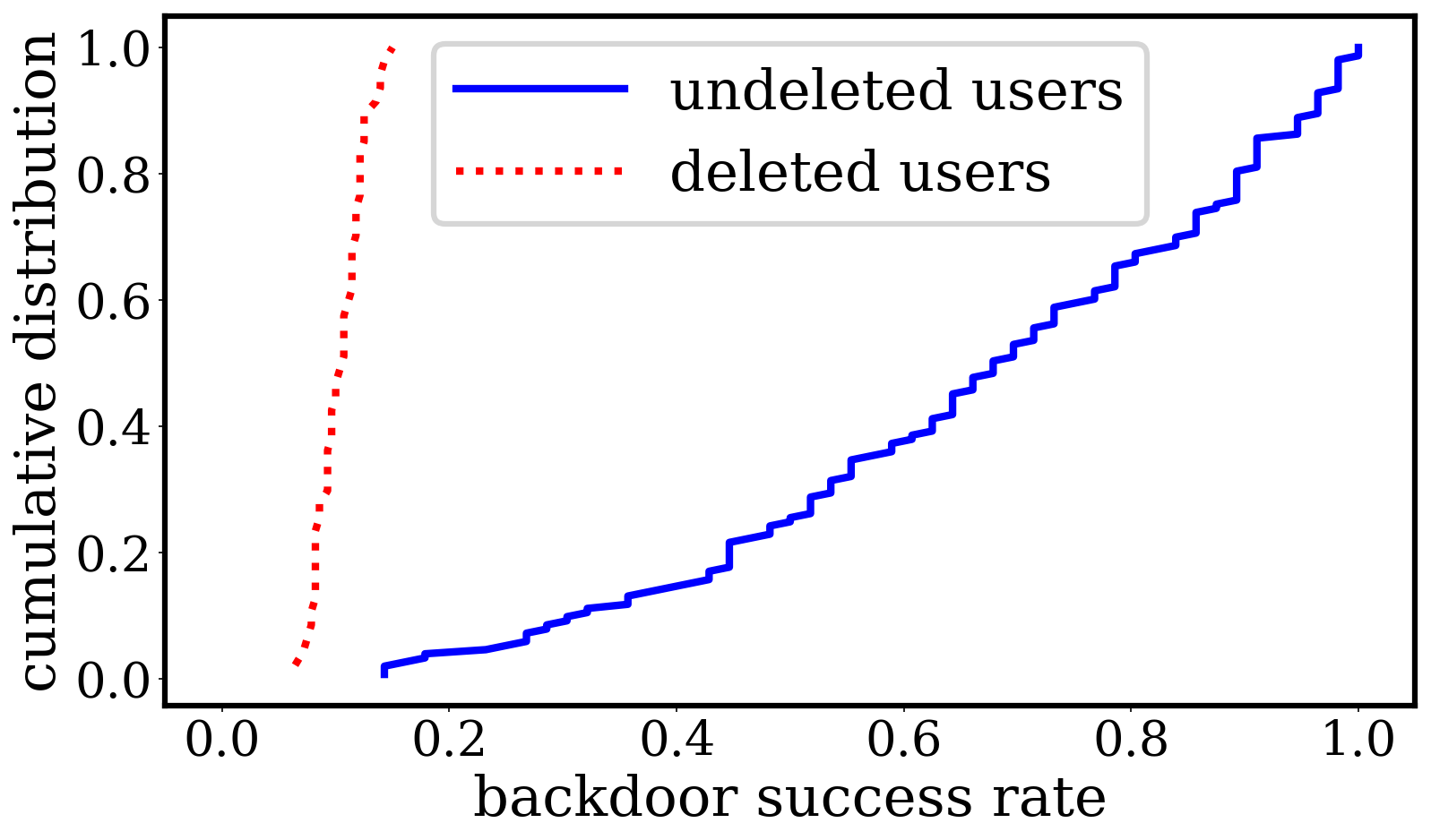}
		\caption{EMNIST}
		\label{fig:def_emnist_acc_cdf}
	\end{subfigure}\hfill
	\begin{subfigure}[t]{0.24\linewidth}
		\centering
		\includegraphics[width=\linewidth]{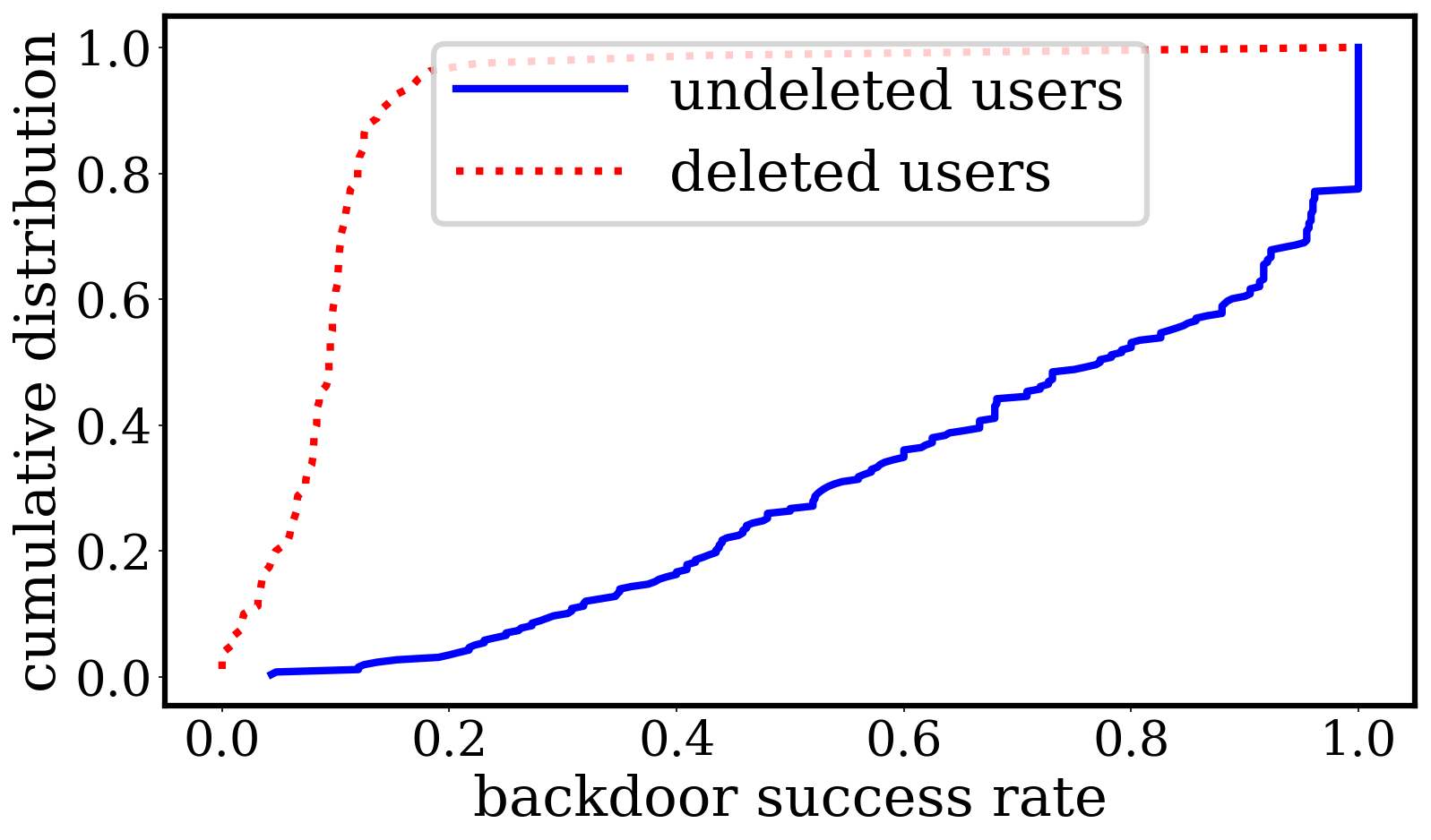}
		\caption{FEMNIST}
		\label{fig:def_femnist_acc_cdf}
	\end{subfigure}\hfill
	\begin{subfigure}[t]{0.24\linewidth}
		\centering
		\includegraphics[width=\linewidth]{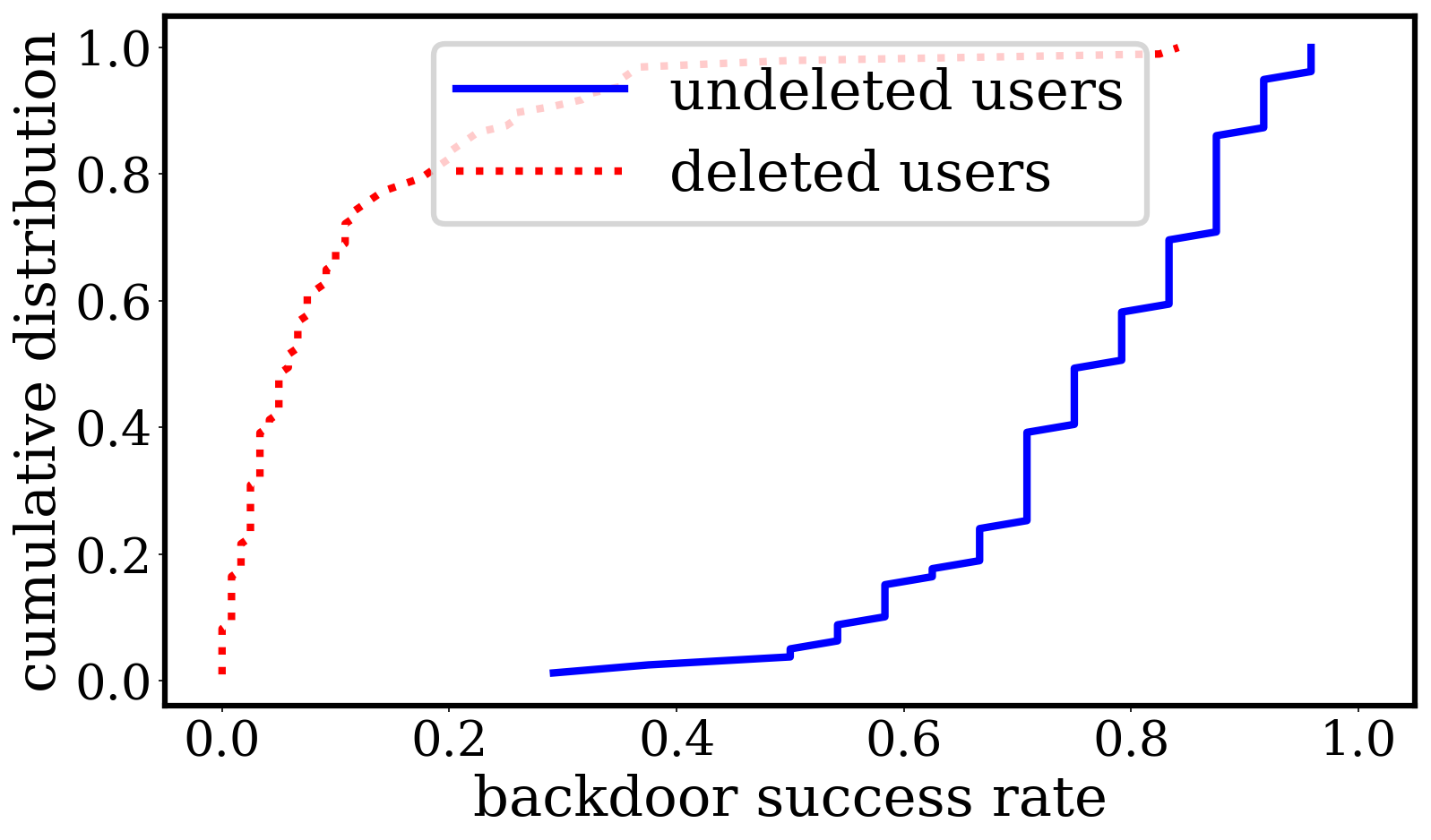}
		\caption{CIFAR10}
		\label{fig:def_cifar_acc_cdf}
	\end{subfigure}
	\begin{subfigure}[t]{0.24\linewidth}
		\centering
		\includegraphics[width=\linewidth]{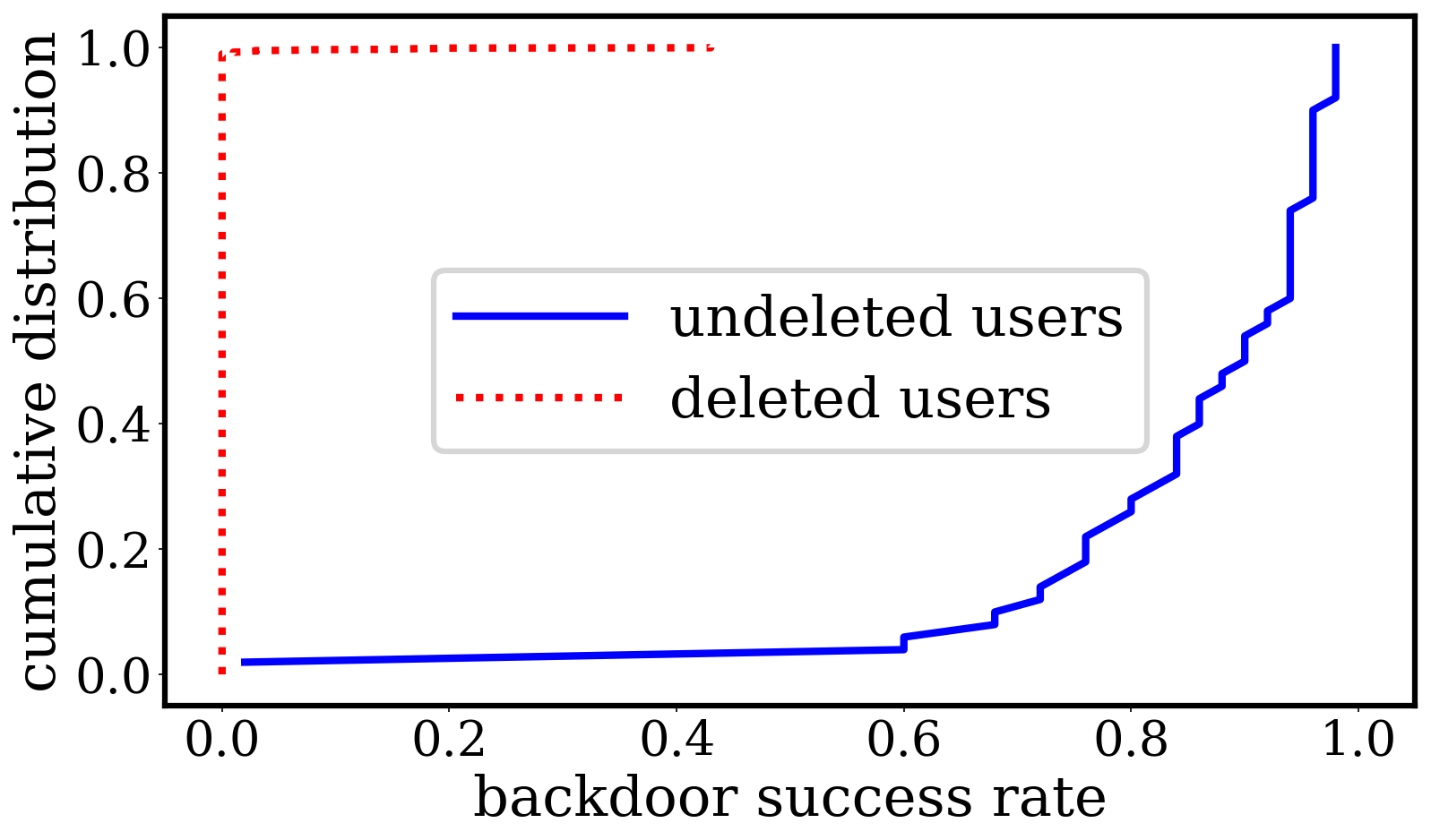}
		\caption{ImageNet}
		\label{fig:def_imagenet_acc_cdf}
	\end{subfigure}
	\caption{The cumulative distributions of backdoor attack success rates for deleted and undeleted users for several datasets in the presence of an adaptive server (data poison ratio $\poisonratio{}=50\%$ and privacy enthusiasts fraction $\userratio=0.05$).
	Similar to Section \ref{sec:individual_evaluation}, here we focus on the heterogeneity in performance across individual users when the \Advanced{} server uses Neural Cleanse \cite{backdoor_defense_wang_sp19} to defend against backdoor attacks.
	}
	\label{fig:def_model_acc_cdf}\vspace{9cm}
\end{figure*}

\section{Number of Users Sustainable}\label{appendix:numberofusers}
%\red{Consider to remove this due to space considerations. This is already somewhat covered by the user collaborate to be more sure.}
%\red{incorporate the fraction of user poisoning - f.}
%\red{What happens if we don't use data poisoning. If confidence is not leaked, probably won't work. If confidence is leaked, model inversion is the research domain and there is a lot of work there.}
%As the number of users using the system increases, it invariably leads to problems of backdoor collision. In other words, just
Given the finite space of backdoor patterns, one or more users can choose similar (similar and not exact because the ML algorithms are robust to small deviations) backdoors which can be a source of inaccuracies. It is important to have bounds on how many users can our mechanism sustain before such collisions start hampering the overall system performance. %This can be characterized using standard coding theory literature. 
For ease of exposition, we consider the domain of image classification. Let us consider a setting with binary images of size $n$, each backdoor has $w$ pixels set, and define dissimilar backdoors to be backdoors that differ in at least $d$ values. For instance, in our backdoor, when using EMNIST dataset, each image is $n = 784 = 28\times 28$, we have set $w = 4$ pixels and $d = 2$ (i.e., if two backdoors share 3 of the 4 pixels, they interfere with each others classification). We want to answer the following question:
\begin{displayquote}
How many backdoor patterns exist that are sufficiently dissimilar to each other? 
\end{displayquote}
This can be answered by an exact mapping to the following problem in coding theory: find the maximal number of binary vectors of length $n$, hamming distance $d$ apart, with constant weight $w$. Exactly computing this quantity, denoted by $A(n, d, w)$, is an open research question but there exist a number of bounds in the literature (Chapter 17 in MacWilliams and Sloane~\cite{macwilliams1977theory}). In our study, we need to compute the quantity: 
\begin{equation}
    \mathsf{\# Backdoors} = \sum_{i = d}^n A(n, i, w)    
\end{equation}
where the summation is because backdoors can differ arbitrarily as long as they are sufficiently dissimilar. Theorem 7 from~\cite{macwilliams1977theory} provides exact values for simple cases such as those required in our EMNIST example. We can then use a simple birthday paradox analysis to bound the number of users in the system to ensure low probability of backdoor collision. Note that the above analysis becomes more involved when using Convolutional Neural Networks as the convolution layers treat neighboring pixels with the same filter weight. %However, the above analysis still enables providing good bounds on the number of users sustainable by our system with low collision probability.

%In particular, for our EMNIST example above, using the exact values from Theorem 7 in~\cite{macwilliams1977theory}, we known that $\mathsf{\# Backdoors} \geq A(784, 4, 4) > 20$ Million. Using a simple application of birthday paradox, we know that when the number of users is $\leq \sqrt{20 \times 10^{6}} \approx 4,500$, we should see low collision of backdoors. 

% Given that the space of backdoors is finite, the number of users that can successfully use our backdoor-based verification system is constrained. 
% %the backdoor system inherently constraints the number of users that can successfully use this system. 
% We can use simple coding theory results to give a bound on the number of users as a function of the space of backdoor algorithms as well as the minimum distance required between two backdoors. This quantity, in the coding theory literature, is denoted by $A_q(n, d)$ where $n$ is the size of backdoor algorithm space (18-bits in our case of $4\times 4$ backdoor algorithm), $d$ is the minimum distance between them (say 2 or 4), and $q$ is the alphabet size.

%\input{Sections/appendix_full_f}

\end{document}